\documentclass[twocolumn,10pt,superscriptaddress,prl,aps]{revtex4-2}
\pdfoutput=1

\usepackage{times}
\usepackage[T1]{fontenc}
\usepackage{microtype}
\usepackage[utf8]{inputenc}
\usepackage{amsmath,amsthm,amsfonts,amsbsy,amssymb}
\usepackage{newtxmath,mathrsfs,bbold,bbm,dsfont}
\usepackage{enumitem}
\usepackage{graphicx,float}
\usepackage[caption=false]{subfig}
\usepackage{array,longtable,booktabs,makecell}
\usepackage{tikz,xcolor}
\usepackage{verbatim,algpseudocode,listings}
\usepackage{csquotes}
\usepackage{hyperref}
\usepackage{lipsum}
\usepackage{CJKutf8}

\newtheorem{fact}{Fact}
\newtheorem{lemma}{Lemma}
\newtheorem{proposition}{Proposition}
\newtheorem{theorem}{Theorem}

\theoremstyle{definition}

\newtheorem{remark}{Remark}
\newtheorem{strategy}{Strategy}

\newcommand{\1}{\mathbb{1}}
\newcommand{\aff}{\mathrm{aff}}
\newcommand{\id}{\mathrm{id}}
\newcommand{\rk}{\mathrm{rank}}
\newcommand{\spn}{\mathrm{span}}
\newcommand{\tr}{\mathrm{Tr}}
\newcommand{\wang}{\textnormal{\begin{CJK}{UTF8}{bsmi}王\end{CJK}}}

\newcommand{\abb}[1]{{\textnormal{#1}}} 
\newcommand{\ch}[1]{{\mathcal{#1}}} 
\newcommand{\g}[1]{{\widehat{#1}}} 
\newcommand{\gd}[1]{{\mathbf{#1}}} 
\newcommand{\s}[1]{{\mathscr{#1}}} 
\newcommand{\spa}[1]{{\mathds{#1}}} 

\newcommand{\bra}[1]{{\langle#1\rvert}}
\newcommand{\ket}[1]{{\lvert#1\rangle}}

\newcommand{\op}[2]{{\ket{#1}\!\bra{#2}}}

\newcommand{\mclose}{\mathclose{}}
\newcommand{\fleft}{\mathopen{}\left}
\newcommand{\fright}{\aftergroup\mclose\right}

\definecolor{darkblue}{rgb}{0,0,0.6}
\definecolor{darkgreen}{rgb}{0,0.45,0.1}
\definecolor{darkred}{rgb}{0.5,0,0}
\hypersetup{
	colorlinks = true,
	citecolor = darkgreen, 
	linkcolor = darkblue, 
	urlcolor  = darkblue, 
	filecolor = darkred 
}

\allowdisplaybreaks

\begin{document}


\title{Retrocausal Capacity of a Quantum Channel: \\ Communicating Through Noisy Closed Timelike Curves}

\author{Kaiyuan Ji}
\email{kj264@cornell.edu}
\affiliation{School of Electrical and Computer Engineering, Cornell University, Ithaca, New York 14850, USA}
\affiliation{Department of Mechanical Engineering, Massachusetts Institute of Technology, 77 Massachusetts Avenue, Cambridge, Massachusetts 02139, USA}

\author{Seth Lloyd}
\email{slloyd@mit.edu}
\affiliation{Department of Mechanical Engineering, Massachusetts Institute of Technology, 77 Massachusetts Avenue, Cambridge, Massachusetts 02139, USA}

\author{Mark M. Wilde}
\email{wilde@cornell.edu}
\affiliation{School of Electrical and Computer Engineering, Cornell University, Ithaca, New York 14850, USA}

\date{\today}


\begin{abstract}
We study the capacity of a quantum channel for retrocausal communication, where messages are transmitted backward in time, from a sender in the future to a receiver in the past, through a noisy postselected closed timelike curve mathematically represented by the channel.  We completely characterize the one-shot retrocausal quantum and classical capacities, and we show that the corresponding asymptotic capacities are equal to the average and sum, respectively, of the channel's max-information and its regularized Doeblin information.  This endows these information measures with a novel operational interpretation.  Furthermore, our characterization can be generalized beyond quantum channels to all completely positive maps.  This imposes information-theoretic limits on transmitting messages via postselected-teleportation-like mechanisms with arbitrary initial- and final-state boundary conditions, including those considered in various black-hole final-state models.
\end{abstract}



\maketitle
\let\oldaddcontentsline\addcontentsline
\renewcommand{\addcontentsline}[3]{}


\textbf{\textit{Introduction}}---The theory of general relativity admits the existence of closed timelike curves (CTCs)~\cite{vanstockum1938GravitationalFieldDistribution,godel1949ExampleNewType,bonnor1980RigidlyRotatingRelativistic,gott1991ClosedTimelikeCurves}.  Such a curve sends a ``chronology-violating'' system backward in time, which can then interact with surroundings of its own past.  The existence of CTCs could lead to perplexing consequences such as violation of causality and escape from a black hole~\cite{horowitz2004BlackHoleFinal,gottesman2004CommentBlackHole,lloyd2006AlmostCertainEscape,lloyd2014UnitarityBlackHole}, and it is challenged by grandfatherlike paradoxes.  To formalize the discussion of these consequences and mitigate paradoxes, abstract quantum-mechanical models of CTCs have been developed, with two notable examples being the Deutschian model~\cite{deutsch1991QuantumMechanicsClosed} based on a self-consistency condition and the postselected model~\cite{bennett2002LectureNotes,lloyd2011ClosedTimelikeCurves,lloyd2011QuantumMechanicsTime,svetlichny2009EffectiveQuantumTime} based on a postselected version of quantum teleportation~\cite{bennett1993TeleportingUnknownQuantum}.  Access to a noiseless CTC of one kind or the other has been shown to unleash stunning information-processing power, such as the ability to distinguish nonorthogonal states perfectly~\cite{brun2009LocalizedClosedTimelike,brun2012PerfectStateDistinguishability}, to clone~\cite{brun2013QuantumStateCloning}, and to speed up computations~\cite{aaronson2005QuantumComputingPostselection,aaronson2009ClosedTimelikeCurves,lloyd2011ClosedTimelikeCurves} (although its interpretation has been the subject of debate~\cite{bennett2009CanClosedTimelike,ralph2010InformationFlowQuantum,cavalcanti2012PreparationProblemNonlinear}).  Recent works have also demonstrated advantages brought by resources simulable using noiseless postselected CTCs (P-CTCs) in other tasks such as metrology~\cite{arvidsson-shukur2020QuantumAdvantagePostselected,arvidsson-shukur2023NonclassicalAdvantageMetrology,song2024AgnosticPhaseEstimation}, resource manipulation~\cite{regula2022ProbabilisticTransformationsQuantum}, hypothesis testing~\cite{regula2024PostselectedQuantumHypothesis}, and communication~\cite{ebler2018EnhancedCommunicationAssistance,chiribella2021IndefiniteCausalOrder,ji2024PostselectedCommunicationQuantum}.

In this paper, we investigate fundamental limits on the efficiency of communication~\cite{shannon1948MathematicalTheoryCommunication} through \emph{noisy} P-CTCs.  A noisy P-CTC can mathematically be represented by a quantum channel, in the sense that information traveling through the noisy P-CTC effectively evolves as if passing through the channel, except for ending up in the past.  Therefore, another way of phrasing the question is: how does the capacity of a quantum channel change, if it were to be used for transmitting messages backward in time (as opposed to forward in time as in conventional communication settings~\cite{shannon1948MathematicalTheoryCommunication,holevo1973BoundsQuantityInformation,schumacher1996SendingEntanglementNoisy,schumacher1997SendingClassicalInformation,lloyd1997CapacityNoisyQuantum,holevo1998CapacityQuantumChannel,bennett1999EntanglementassistedClassicalCapacity,shor2002QuantumChannelCapacity,devetak2005PrivateClassicalCapacity,cubitt2011ZeroerrorChannelCapacity,duan2016NosignallingassistedZeroerrorCapacity,khatri2024PrinciplesQuantumCommunication})?  We term this unconventional communication setting \emph{retrocausal communication}.  A distinctive feature of retrocausal communication is the possible formation of causal loops, which is counterintuitive, yet perfectly reasonable granted the existence of the noisy P-CTC.  In fact, the communication power of the noisy P-CTC can only be fully unlocked by exploiting such an exotic structure of spacetime.  Our main technical result is to establish the Shannon-theoretic limit of retrocausal communication, accompanied by an explicit construction of an optimal strategy in the one-shot regime.

Our use of the postselected model of CTCs is motivated by several considerations.  P-CTCs, based on postselected teleportation, can be simulated experimentally with reasonable overhead~\cite{lloyd2011ClosedTimelikeCurves}, and they have been shown to be consistent~\cite{lloyd2011QuantumMechanicsTime} with path-integral approaches~\cite{morris1988WormholesTimeMachines,kim1991VacuumFluctuationsPrevent,boulware1992QuantumFieldTheory,politzer1992SimpleQuantumSystems,politzer1994PathIntegralsDensity,hartle1994UnitarityCausalityGeneralized} to CTCs.  Most importantly, a noiseless P-CTC amounts to an identity channel from the future to the past, suggesting that the model is consistent with the formalism of quantum Shannon theory~\cite{wilde2017QuantumInformationTheory,hayashi2017QuantumInformationTheory,watrous2018TheoryQuantumInformation,khatri2024PrinciplesQuantumCommunication} and suitable for studying CTCs as media for communication.  This contrasts with the Deutschian model of CTCs~\cite{deutsch1991QuantumMechanicsClosed}, in which the chronology-violating system becomes decorrelated from the environment after traveling backward in time, so that partial information (such as one part of an entangled state) can never be faithfully transmitted from the future to the past.  Moreover, P-CTCs are closely related to the Horowitz-Maldacena final-state model of black-hole evaporation~\cite{horowitz2004BlackHoleFinal}, which was proposed to reconcile unitarity with Hawking's semiclassical arguments~\cite{hawking1976BreakdownPredictabilityGravitational} and inspired by a time-symmetric modification of quantum mechanics~\cite{aharonov1964TimeSymmetryQuantum,griffiths1984ConsistentHistoriesInterpretation,gell-mann1994TimeSymmetryAsymmetry}.  In variations of the model~\cite{gottesman2004CommentBlackHole,lloyd2006AlmostCertainEscape,lloyd2014UnitarityBlackHole} that allow for arbitrary initial- or final-state boundary conditions, nonunitary postselected-teleportation-like mechanisms resembling noisy P-CTCs could emerge.

\textbf{\textit{Noisy P-CTCs}}---A chronology-violating system traveling through a noiseless P-CTC~\cite{bennett2002LectureNotes,lloyd2011ClosedTimelikeCurves,lloyd2011QuantumMechanicsTime,svetlichny2009EffectiveQuantumTime} can be thought of as being teleported to its past under postselecting the event that the Bell measurement in the teleportation protocol~\cite{bennett1993TeleportingUnknownQuantum} yields the desirable outcome that requires no further action to complete the protocol.  This time-travel mechanism is ``noiseless,'' as the past and future versions of the system are identical.  Now consider a noisy mechanism of time travel, which turns a future system $A$ into a past system $B$ under the effect of a quantum channel $\ch{N}_{A\to B}$~\cite{genkina2012OptimalProbabilisticSimulation}.  A natural such mechanism is to let the future system travel through a noiseless P-CTC and then be acted upon ordinarily by the channel.  We refer to this mechanism as a ``noisy'' P-CTC and simply represent it by the channel $\ch{N}_{A\to B}$, with the implicit assumption that $B$ precedes $A$ in time.

P-CTCs can be combined with other ``chronology-respecting'' processes.  For instance, given a noisy P-CTC represented by $\ch{N}_{A\to B}$, one can feed $B$ into a bipartite channel $\ch{T}_{EB\to FA}$ as part of the latter's input and then feed the output $A$ of $\ch{T}_{EB\to FA}$ into the noisy P-CTC's entrance (see Fig.~\ref{fig:noisy-curve}).  This might result in a causal loop $B\to A\to B$, yet it is perfectly reasonable granted the existence of the noisy P-CTC because $B$ can influence $A$ via $\ch{T}_{EB\to FA}$ without violating chronology.  Potential paradoxes arising from such a loop can be addressed within the model of P-CTCs~\cite{lloyd2011ClosedTimelikeCurves}.  The resulting transformation from $E$ to $F$ is then represented by the following nonlinear map (see Fig.~\ref{fig:noisy-postselection})~\cite{lloyd2011ClosedTimelikeCurves,lloyd2011QuantumMechanicsTime}:
\begin{align}
	\rho_{RE}&\mapsto\frac{\tr_{AA'}\fleft[\Phi_{AA'}\ch{T}_{EB\to FA}\fleft[\rho_{RE}\otimes\ch{N}_{A\to B}\fleft[\Phi_{AA'}\fright]\fright]\fright]}{\tr\fleft[\Phi_{AA'}\ch{T}_{EB\to FA}\fleft[\rho_{RE}\otimes\ch{N}_{A\to B}\fleft[\Phi_{AA'}\fright]\fright]\fright]}, \label{eq:curve}
\end{align}
where $R$ is an arbitrary reference system and $\Phi_{AA'}\equiv(1/d_A)\sum_{i,j}\op{i}{j}_A\otimes\op{i}{j}_{A'}$ is the standard maximally entangled state between $A$ and $A'$.

\begin{figure}[t]
\subfloat[\label{fig:noisy-curve}]{\includegraphics[scale=0.24]{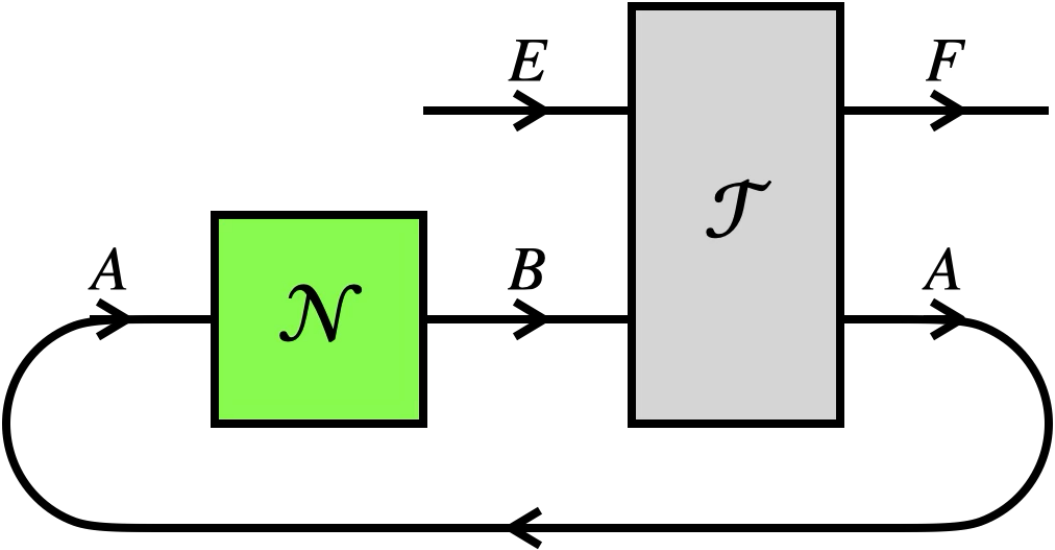}} \\
\subfloat[\label{fig:noisy-postselection}]{\includegraphics[scale=0.24]{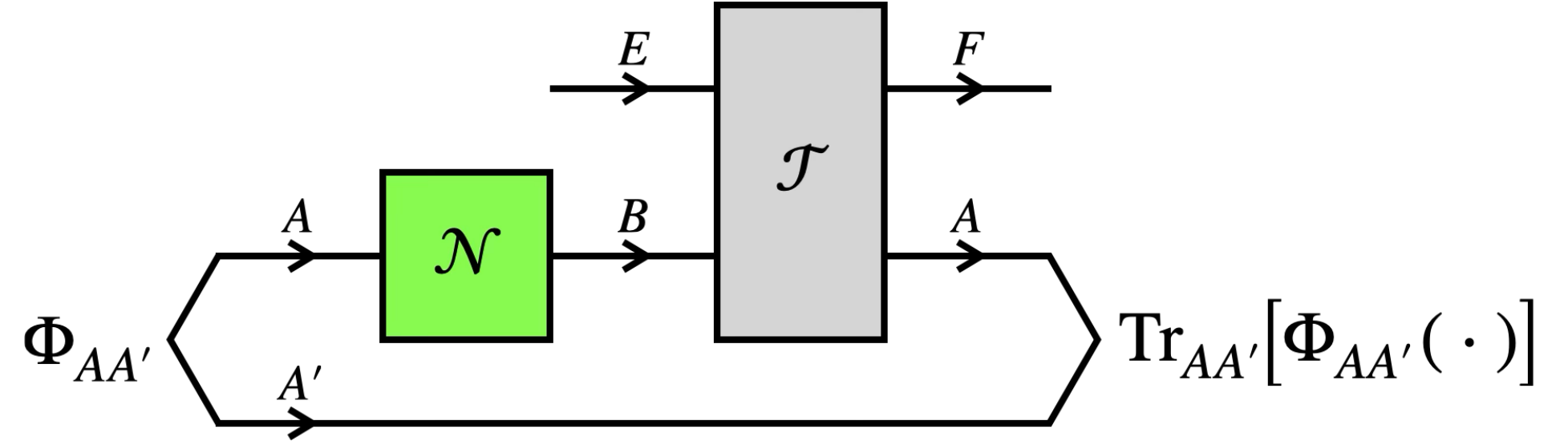}}
\caption{Enclosing a bipartite channel $\ch{T}_{EB\to FA}$ with a noisy P-CTC represented by a channel $\ch{N}_{A\to B}$ (Fig.~\ref{fig:noisy-curve}) is equivalent to preprocessing $\ch{T}_{EB\to FA}$ with $\ch{N}_{A\to B}$ and performing postselected teleportation from the output $A$ of $\ch{T}_{EB\to FA}$ to the input $A$ of $\ch{N}_{A\to B}$ (Fig.~\ref{fig:noisy-postselection} with renormalization).}
\label{fig:noisy}
\end{figure}

\textbf{\textit{Retrocausal communication}}---The setting of retrocausal communication can perhaps be best explained in reference to the film \textit{Interstellar}~\cite{thorne2014ScienceInterstellar} (see Fig.~\ref{fig:interstellar}).  In the film, a father in the future discovers that he can communicate to his young daughter in the past via some noisy mechanism.  Due to its chronology-violating nature, we model this mechanism as a noisy P-CTC represented by a quantum channel $\ch{N}_{A\to B}$.  We are concerned with the maximum possible efficiency at which messages can be reliably transmitted from the father to the daughter through the noisy P-CTC under an encoding and decoding of their choice.  This efficiency is measured by the \emph{retrocausal capacity} of $\ch{N}_{A\to B}$, which we will define shortly~\footnote{The retrocausal capacity of a channel is also referred to as its \emph{flux capacity}.}.  One may wonder how this exotic scenario differs from a conventional communication setting in which messages are transmitted forward in time.  The difference becomes manifest when taking into account the following possibility (as is also visible in the film~\cite{thorne2014ScienceInterstellar}): the father, who is in the future, may retrieve his memory of past events that he has witnessed, even including the daughter's decoding of the message which he is about to send!  It would thus not be surprising that he will consult his memory of the daughter's decoding when encoding his message, so as to maximize the efficiency of the communication.

\begin{figure}[t]
\includegraphics[width=\linewidth]{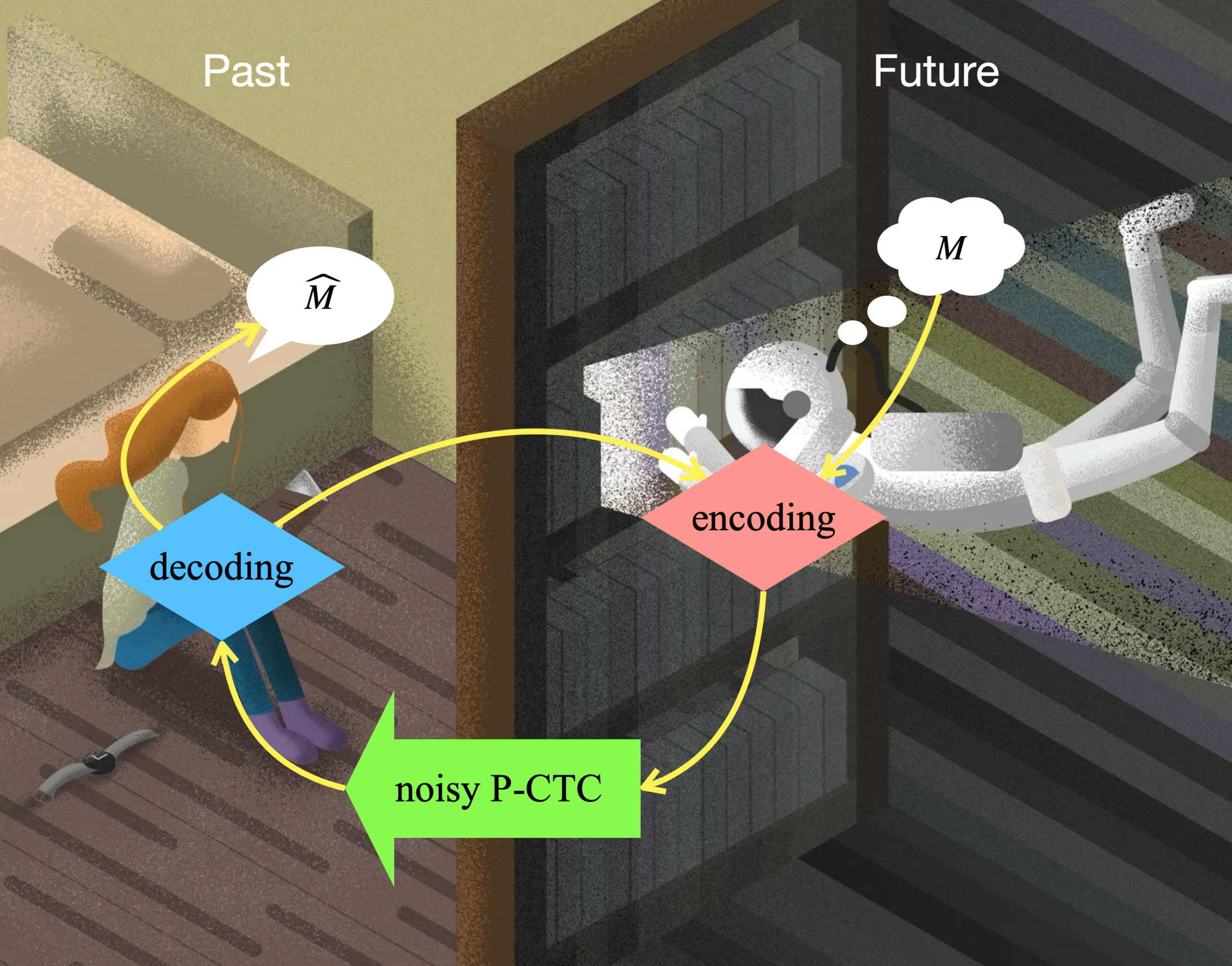}
\caption{In the film \textit{Interstellar}~\cite{thorne2014ScienceInterstellar}, a father, who is in the future, encodes his message to his daughter, who is in the past, in a system traveling through a noisy P-CTC.  The father can retrieve and consult his memory of past events that he has witnessed, so that his encoding could be influenced by what results from the daughter's decoding, forming a causal loop.  Background artwork: \textit{Bookshelf Scene} (2018), \copyright Alexey Nefed (motion.nalex), reproduced with permission~\cite{nefed2018BookshelfScene}.}
\label{fig:interstellar}
\end{figure}

Evidently, this allows what results from the daughter's decoding to influence the father's encoding, which in turn influences the daughter's decoding via the noisy P-CTC, thus forming a causal loop.  However, to stress again, granted the existence of the noisy P-CTC, the formation of this loop is physically justified: the father and the daughter operate in a perfectly chronology-respecting manner, including the father's retrieving and consulting his memory of past events; the only chronology-violating element here is the noisy P-CTC, which models a given mechanism beyond their control.  Note that, if the father is not allowed to retrieve or consult his memory, then the setting under consideration would not be different (at least mathematically) from conventional (forward-in-time) communication settings, since everything would anyway abide by a definite causal order and could thus be realized without necessarily violating chronology.

Following the intuition explained above, the most general form of strategy that the father and the daughter can adopt, without violating chronology (when not invoking the noisy P-CTC), can be described as follows (see Fig.~\ref{fig:communication-retrocausal}): (i) in the past, the daughter receives $B$ from the exit of the noisy P-CTC, decodes the message by applying a channel $\ch{D}_{B\to\g{M}L}$ with $\g{M}$ storing her recovery of the message, and leaves $L$ in a noiseless quantum memory for the father to retrieve in the future; (ii) in the future, the father retrieves $L$, encodes his message stored in $M$ by applying a channel $\ch{E}_{ML\to A}$, and feeds $A$ into the entrance of the noisy P-CTC.  The resulting transformation from the father's message system to the daughter's recovery system is represented by a nonlinear map $\ch{M}_{M\to\g{M}}$, whose action on an arbitrary pure state $\psi_{RM}$, with $R$ an arbitrary reference system, is given by
\begin{align}
	&\ch{M}_{M\to\g{M}}\fleft[\psi_{RM}\fright]\equiv \notag\\
	&\frac{\tr_{AA'}\fleft[\Phi_{AA'}\ch{E}_{ML\to A}\fleft[\psi_{RM}\otimes\ch{D}_{B\to L\g{M}}\circ\ch{N}_{A\to B}\fleft[\Phi_{AA'}\fright]\fright]\fright]}{\tr\fleft[\Phi_{AA'}\ch{E}_{ML\to A}\fleft[\psi_{RM}\otimes\ch{D}_{B\to L\g{M}}\circ\ch{N}_{A\to B}\fleft[\Phi_{AA'}\fright]\fright]\fright]}, \label{eq:communication}
\end{align}
following the machinery of Eq.~\eqref{eq:curve}.

\begin{figure}[t]
\subfloat[\label{fig:communication-retrocausal}]{\includegraphics[scale=0.24]{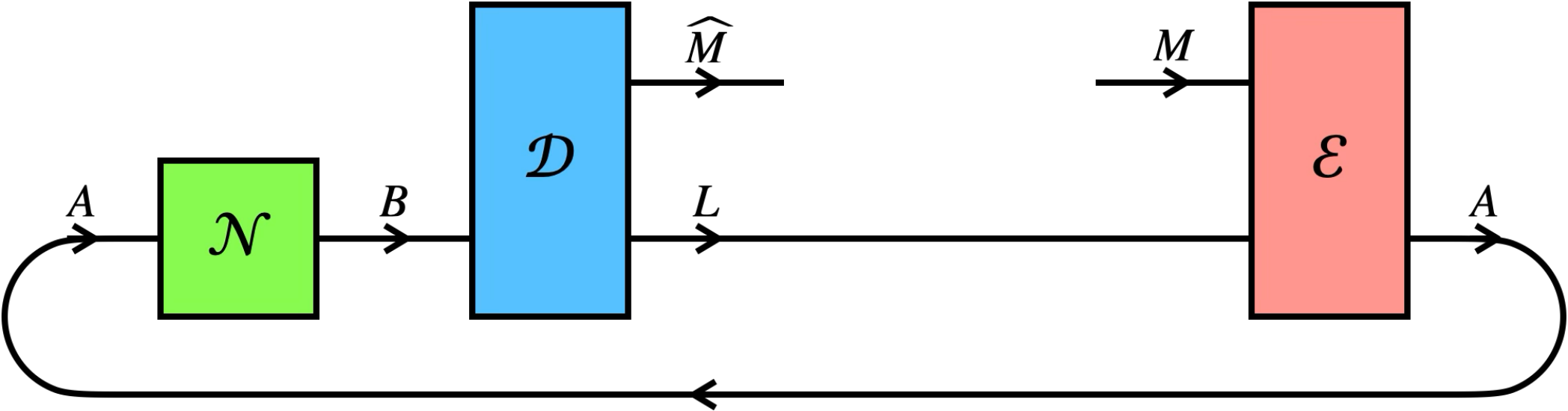}} \\
\subfloat[\label{fig:communication-assisted}]{\includegraphics[scale=0.24]{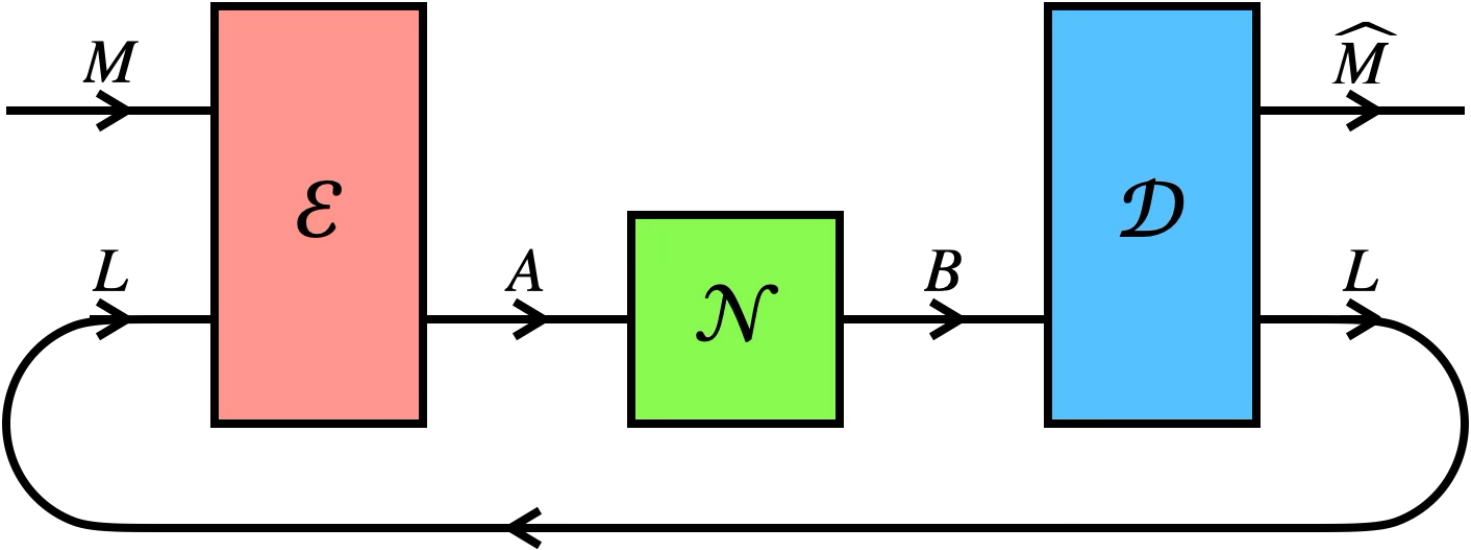}}
\caption{Retrocausal communication through a noisy P-CTC represented by a channel $\ch{N}_{A\to B}$ (Fig.~\ref{fig:communication-retrocausal}) is mathematically equivalent to forward-in-time communication through the channel $\ch{N}_{A\to B}$ with the assistance of a noiseless P-CTC (Fig.~\ref{fig:communication-assisted}).}
\label{fig:communication}
\end{figure}

Another observation is that, in the presence of a causal loop, the notion of past and future becomes elusive and can be treated flexibly.  For instance, if one imagines the ``first-hand experience'' of a message traveling through the noisy P-CTC, then the conclusion would be that the father resides in the ``past'' where it sets out and the daughter resides in the ``future'' where it arrives.  Judging from this experience, the noisy P-CTC would just appear as an ordinary channel, whereas the memory $L$ travels through a noiseless P-CTC.  This suggests that, mathematically, retrocausal communication (see Fig.~\ref{fig:communication-retrocausal}) is equivalent to forward-in-time communication assisted by a noiseless P-CTC (see Fig.~\ref{fig:communication-assisted}).  For coherence of presentation, we will continue with our original formulation of retrocausal communication (see Fig.~\ref{fig:communication-retrocausal}), where the father is in the future and the daughter is in the past.

\textbf{\textit{Retrocausal capacity}}---To quantify the quality of transmitting quantum messages via the nonlinear map $\ch{M}_{M\to\g{M}}$ [see Eq.~\eqref{eq:communication}], we define its worst-case quantum infidelity as
\begin{align}
	q_\abb{infid}\fleft(\ch{M}\fright)&\coloneq1-\inf_{\psi_{RM}}\tr\fleft[\psi_{R\g{M}}\ch{M}_{M\to\g{M}}\fleft[\psi_{RM}\fright]\fright],
\end{align}
where the infimization is over all pure states with $R$ of an arbitrary dimension.  We then define the \emph{one-shot retrocausal quantum capacity} of the quantum channel $\ch{N}_{A\to B}$ as the maximum number of qubits that can be transmitted within an infidelity $\varepsilon\in[0,1]$ through the noisy P-CTC that it represents:
\begin{align}
	Q_\abb{retro}^\varepsilon\fleft(\ch{N}\fright)&\coloneq\sup_{\ch{E}_{ML\to A},\ch{D}_{B\to\g{M}L}}\left\{\log_2d_M\colon q_\abb{infid}\fleft(\ch{M}\fright)\leq\varepsilon\right\},
\end{align}
where $\ch{M}_{M\to\g{M}}$ is given by Eq.~\eqref{eq:communication} and the supremization is over all choices of encoding and decoding (which are quantum channels) with $M$ and $L$ of arbitrary dimensions.  We define the \emph{asymptotic retrocausal quantum capacity} of $\ch{N}_{A\to B}$ as
\begin{align}
	Q_\abb{retro}\fleft(\ch{N}\fright)&\coloneq\inf_{\varepsilon\in(0,1)}\liminf_{n\to\infty}\frac{1}{n}Q_\abb{retro}^\varepsilon\fleft(\ch{N}^{\otimes n}\fright).
\end{align}

If the messages being transmitted are classical, then one can disregard the reference system $R$ and only consider all (perfectly distinguishable) symbols in an alphabet of size $d_M$.  The one-shot retrocausal classical capacity $C_\abb{retro}^\varepsilon\fleft(\ch{N}\fright)$ is then defined similarly to its quantum counterpart but with the worst-case classical error probability,
\begin{align}
	p_\abb{error}\fleft(\ch{M}\fright)&\coloneq1-\min_{m}\tr\fleft[\op{m}{m}_\g{M}\ch{M}_{M\to\g{M}}\fleft[\op{m}{m}_M\fright]\fright],
\end{align}
in place of $q_\abb{infid}(\ch{M})$, where the minimization is over all symbols in the alphabet.  The asymptotic retrocausal classical capacity $C_\abb{retro}(\ch{N})$ is defined accordingly.

To reiterate, the retrocausal capacities defined above assume that the channel $\ch{N}_{A\to B}$ represents a noisy P-CTC, with $B$ preceding $A$ in time.  They may not apply to other models of CTCs, nor do they suggest backward-in-time signaling in standard quantum theory, where chronology-violating resources are forbidden.

\textbf{\textit{Main result}}---For a quantum channel $\ch{N}_{A\to B}$ and a real number $\varepsilon\in(0,1)$, the one-shot retrocausal quantum and classical capacities are equal to, respectively,
\begin{align}
	Q_\abb{retro}^\varepsilon\fleft(\ch{N}\fright)&=\log_2\left\lfloor\sqrt{\frac{\varepsilon}{1-\varepsilon}2^{I_{\max}\fleft(\ch{N}\fright)+I_\abb{doe}\fleft(\ch{N}\fright)}+1}\right\rfloor, \label{eq:quantum}\\
	C_\abb{retro}^\varepsilon\fleft(\ch{N}\fright)&=\log_2\left\lfloor\frac{\varepsilon}{1-\varepsilon}2^{I_{\max}\fleft(\ch{N}\fright)+I_\abb{doe}\fleft(\ch{N}\fright)}+1\right\rfloor, \label{eq:classical}
\end{align}
where $\lfloor\cdot\rfloor$ denotes the floor function.  The max-information~\cite{konig2009OperationalMeaningMin} is defined as
\begin{align}
	I_{\max}\fleft(\ch{N}\fright)&\coloneq\inf_{\lambda,\sigma_B}\left\{\log_2\lambda\colon\ch{N}_{A\to B}\fleft[\Phi_{A'A}\fright]\leq\lambda\pi_{A'}\otimes\sigma_B\right\},
\end{align}
where $\pi_{A'}\equiv\1_{A'}/d_{A'}$ and the infimization is over all real numbers and all quantum states of $B$.  The Doeblin information~\cite{george2025QuantumDoeblinCoefficients} is defined as
\begin{align}
	I_\abb{doe}\fleft(\ch{N}\fright)&\coloneq\inf_{\lambda,\tau_B}\left\{\log_2\lambda\colon\pi_{A'}\otimes\tau_B\leq\lambda\ch{N}_{A\to B}\fleft[\Phi_{A'A}\fright]\right\},
\end{align}
where the infimization is over all real numbers and all unit-trace \emph{Hermitian} operators in $B$.

The asymptotic retrocausal quantum and classical capacities are given by
\begin{align}
	Q_\abb{retro}\fleft(\ch{N}\fright)&=\frac{1}{2}C_\abb{retro}\fleft(\ch{N}\fright)=\frac{1}{2}\left(I_{\max}\fleft(\ch{N}\fright)+I_\abb{doe}^\infty\fleft(\ch{N}\fright)\right), \label{eq:asymptotic}
\end{align}
where $I_\abb{doe}^\infty(\ch{N})\coloneq\lim_{n\to\infty}(1/n)I_\abb{doe}(\ch{N}^{\otimes n})$ is the regularized Doeblin information.

Several implications of our result are in order.  First, our result provides exact closed-form expressions for the one-shot retrocausal quantum and classical capacities.  This is notable considering that having a one-shot characterization up to this precision has been extremely rare in quantum Shannon theory.  Second, our result implies a complete characterization of the error exponent and the strong-converse exponent of retrocausal communication (see Supplemental Material~\cite[Secs.~\ref{sec:quantum-exponent} and \ref{sec:classical-exponent}]{Note2}\vphantom{\footnote{See Supplemental Material for technical background, detailed proofs of our results, numerical examples, and additional discussions, which includes Refs.~\cite{datta2009MinMaxrelativeEntropies,gupta2015MultiplicativityCompletelyBounded,diaz2018UsingReusingCoherence,nuradha2025MultivariateFidelities,ji2025BarycentricBoundsError,girardi2025QuantumUmlautInformation,fang2020QuantumChannelSimulation,wilde2020AmortizedChannelDivergence,chitambar2023CommunicationValueQuantum,chitambar2019QuantumResourceTheories,gour2025QuantumResourceTheories,vidal1999RobustnessEntanglement,harrow2003RobustnessQuantumGates,elitzur1992QuantumNonlocalityEach,lewenstein1998SeparabilityEntanglementComposite,datta2009MaxrelativeEntropyEntanglement,fekete1923UeberVerteilungWurzeln,chiribella2008TransformingQuantumOperations,burniston2020NecessarySufficientConditions,shor1995SchemeReducingDecoherence,schumacher1996QuantumDataProcessing,barnum1998InformationTransmissionNoisy,barnum2000QuantumFidelitiesChannel,bennett2002EntanglementassistedCapacityQuantum,holevo2002EntanglementassistedClassicalCapacity,matthews2012LinearProgramFinite,takagi2020ApplicationResourceTheory,gour2019ComparisonQuantumChannels,eggeling2002SemicausalOperationsAre,king2003CapacityQuantumDepolarizing,bennett1997CapacitiesQuantumErasure,giovannetti2005InformationcapacityDescriptionSpinchain}.}} for formal statements).  Third, our result provides a novel operational interpretation for the max- and Doeblin informations in the context of communication, contrasting, for instance, the previously known interpretation of the max-information in channel simulation (which is a setting reverse to communication)~\cite{duan2016NosignallingassistedZeroerrorCapacity} (see Supplemental Material~\cite[Sec.~\ref{sec:information}]{Note2} for more discussion).  Fourth, for several classes of channels (including measurement channels and certain covariant channels), the Doeblin information is additive~\cite{george2025QuantumDoeblinCoefficients}, implying that $I_\abb{doe}^\infty(\ch{N})=I_\abb{doe}(\ch{N})$.  This means that, for such channels, the asymptotic retrocausal capacities have a single-letter expression and are efficiently computable via semidefinite programs (SDPs) (see Supplemental Material~\cite[Sec.~\ref{sec:example}]{Note2} for numerical examples).  For general quantum channels, the Doeblin information is not always additive, and the efficient computability of its regularization is unknown.  In this case, we provide single-letter and SDP-computable upper and lower bounds on the asymptotic retrocausal capacities in Supplemental Material~\cite[Remark~\ref{rem:bounds}]{Note2}.

\textbf{\textit{Optimal strategy}}---To provide some intuition behind the derivation of our result, we outline a construction of an optimal strategy, which achieves the one-shot retrocausal quantum capacity precisely.  For a rigorous analysis of its performance and a proof of its optimality, see Supplemental Material~\cite[Sec.~\ref{sec:quantum}]{Note2}.  The strategy operates as follows (see Fig.~\ref{fig:strategy}).

\begin{figure}[t]
\includegraphics[width=\linewidth]{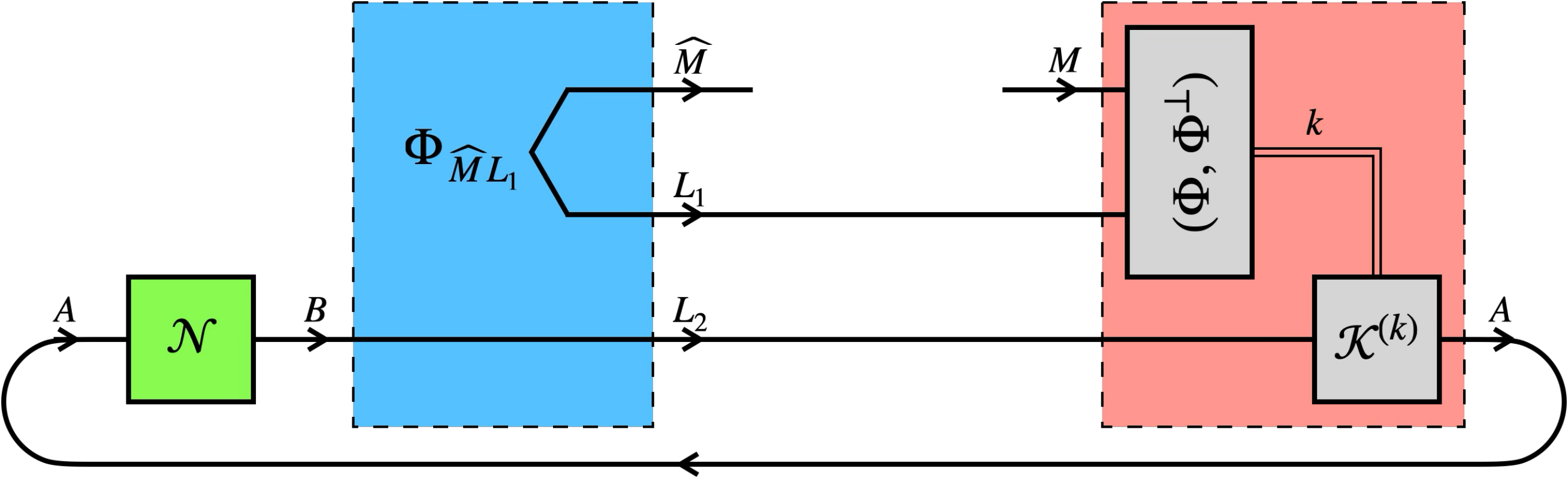}
\caption{Amplified probabilistic teleportation is an optimal strategy for retrocausal communication.  The blue and red boxes encapsulate the daughter's decoding and the father's encoding, respectively.}
\label{fig:strategy}
\end{figure}

In the past, the daughter receives $B$ from the exit of the noisy P-CTC.  She prepares the standard maximally entangled state between $\g{M}$ and a noiseless quantum memory $L_1$, keeps $B$ in another noiseless quantum memory $L_2$, and leaves both $L_1$ and $L_2$ for the father to retrieve in the future.  This concludes the daughter's decoding.  In the future, the father retrieves $L_1$ and $L_2$ and performs a Bell measurement on $M$ and $L_1$.  If the measurement ends up being a projection onto the standard maximally entangled state, then he applies a channel $\ch{K}_{L_2\to A}^{(0)}$ to $L_2$; otherwise he applies $\ch{K}_{L_2\to A}^{(1)}$.  Finally, he feeds $A$ into the entrance of the noisy P-CTC.  This concludes the father's encoding.

We call the above strategy \emph{amplified probabilistic teleportation}, and its underlying functioning can be understood as follows.  Instead of trying to code his message into $A$, which enters the noisy P-CTC, the father intends to ``teleport''~\cite{bennett1993TeleportingUnknownQuantum} the message probabilistically to the daughter via $L_1$.  Indeed, if the father's Bell measurement ends up being a projection onto the standard maximally entangled state, then what the daughter ends up with in $\g{M}$ is a perfect recovery of the message in $M$.  However, this ``desirable'' outcome occurs only with probability $1/d_M^2$, or at least so without invoking the noisy P-CTC.  Yet, as illustrated in Fig.~\ref{fig:strategy}, the remaining part of the strategy engineers a causal loop solely consisting of the noisy P-CTC, represented by $\ch{N}_{A\to B}$, and the channel $\ch{K}_{L_2\to A}^{(k)}$, conditioned on $k\in\{0,1\}$.  Due to the nonlinearity of such a loop (specifically, its capability of renormalizing the probabilities of certain events), the father can choose the channels $\ch{K}_{L_2\to A}^{(0)}$ and $\ch{K}_{L_2\to A}^{(1)}$ in such a way that the renormalized probability of the desirable outcome is maximally amplified.  If the renormalized probability of the desirable outcome can be amplified to at least $1-\varepsilon$ for a message system $M$ of dimension $d_M$, then $\log_2d_M$ qubits can indeed be ``teleported'' from $M$ to $\g{M}$ via $L_1$ subject to a failure probability at most $\varepsilon$, which upper bounds the overall infidelity of communication.  The max- and Doeblin informations of $\ch{N}_{A\to B}$ come to relevance in quantifying the noisy P-CTC's maximal amplification of the probability of the desirable outcome, due to their interpretations in terms of the maximum and minimum achievable singlet fractions, respectively~\cite{konig2009OperationalMeaningMin,george2025QuantumDoeblinCoefficients}.

\textbf{\textit{Discussion}}---We have proposed the setting of retrocausal communication based on the model of P-CTCs, demonstrating that exchanging the temporal locations of the sender and the receiver and ``reversing'' the chronological direction of the communication channel shifts the paradigm of quantum Shannon theory significantly.  In the one-shot regime, the fundamental limit on its efficiency can be determined precisely and achieved with a simple strategy contrasting with conventional coding schemes.  In the asymptotic regime, the limit can be incarnated as a resource inequality~\cite{devetak2008ResourceFrameworkQuantum} involving communication resources in both chronological directions:
\begin{align}
	\langle\overleftarrow{\ch{N}}\rangle+\infty[q\to q]&\geq\frac{1}{2}\left(I_{\max}\fleft(\ch{N}\fright)+I_\abb{doe}^\infty\fleft(\ch{N}\fright)\right)[q\leftarrow q],
\end{align}
where $\langle\overleftarrow{\ch{N}}\rangle$ denotes a noisy P-CTC represented by $\ch{N}_{A\to B}^{\otimes n}$ and $[q\to q]$ and $[q\leftarrow q]$ denote a noiseless quantum memory and a noiseless P-CTC, respectively, on $n$ qubits.  The noiseless quantum memory left by the receiver in the past and retrieved by the sender in the future plays an essential role in achieving this limit.  An interesting open question is how this limit is affected if the memory is noisy, say, represented by a channel $\ch{L}_{C\to D}$.  In other words, one may try to characterize the tradeoff between the rates $R$ and $Q$ in the following resource inequality:
\begin{align}
	\langle\overleftarrow{\ch{N}}\rangle+R\langle\overrightarrow{\ch{L}}\rangle&\geq Q[q\leftarrow q].
\end{align}

Another future direction is to explore the relation between our result and final-state models of black-hole evaporation~\cite{horowitz2004BlackHoleFinal,gottesman2004CommentBlackHole,lloyd2006AlmostCertainEscape,lloyd2014UnitarityBlackHole}.  The original final-state model proposed by Horowitz and Maldacena~\cite{horowitz2004BlackHoleFinal} postulates that black-hole evaporation can be described by postselected teleportation, which reconciles unitarity and Hawking's semiclassical arguments~\cite{hawking1976BreakdownPredictabilityGravitational}.  It was subsequently pointed out~\cite{gottesman2004CommentBlackHole,lloyd2006AlmostCertainEscape,lloyd2014UnitarityBlackHole} that, even assuming a final-state boundary condition, unitarity could still be spoiled if the initial- or final-state boundary condition can be arbitrary, and the evaporation would then be described by a postselected-teleportation-like mechanism represented by a completely positive map that is not necessarily trace preserving (see Supplemental Material~\cite[Sec.~\ref{sec:noisy}]{Note2} for details).  Notably, our characterization of the retrocausal capacity can be shown to apply to all completely positive maps and not just quantum channels (see Supplemental Material~\cite[Secs.~\ref{sec:quantum} and \ref{sec:classical}]{Note2} for proofs), and it is therefore general enough to account for final-state models with arbitrary initial- or final-state boundary conditions.  Further implications of our result for final-state models will be discussed in a future work~\cite{future}.

\begin{acknowledgements}
\textbf{\textit{Acknowledgments}}---We are grateful to Alexey Nefed (motion.nalex) for permission to reproduce his artwork \textit{Bookshelf Scene} (2018) as the background in Fig.~\ref{fig:interstellar}.  We thank Bartosz Regula, Eric Chitambar, and Theshani Nuradha for helpful discussions and comments.  K.J. and M.M.W. acknowledge support from the National Science Foundation under Grant No.~2329662 and from the Cornell School of Electrical and Computer Engineering.  This material is based upon work supported by, or in part by, the U. S. Army Research Laboratory and the U. S. Army Research Office under Contract/Grant No.~W911NF2310255 and by DoE under Contract No.~DE-SC0012704.
\end{acknowledgements}


\bibliographystyle{apsc}
\bibliography{Library}


\clearpage
\onecolumngrid
\setcounter{page}{1}
\setcounter{equation}{0}
\setcounter{figure}{0}
\renewcommand{\theequation}{S\arabic{equation}}
\renewcommand{\thefigure}{S\arabic{figure}}

\let\addcontentsline\oldaddcontentsline
\setcounter{secnumdepth}{3}

\makeatletter
\renewcommand\normalsize{\@setfontsize\normalsize{12}{14}}
\renewcommand\small{\@setfontsize\small{11}{13}}
\renewcommand\footnotesize{\@setfontsize\footnotesize{10}{12}}
\renewcommand\large{\@setfontsize\large{14}{17}}
\renewcommand\Large{\@setfontsize\Large{17}{20}}
\makeatother
\normalsize

\begin{center}
\textbf{\large Supplemental Material: Retrocausal Capacity of a Quantum Channel}
\end{center}

\tableofcontents

\section{Notation}
\label{sec:notation}

Let $\spa{R}$ denote the set of real numbers, $\spa{C}$ the set of complex numbers, and $\spa{N}$ the set of nonnegative integers.  For a nonnegative real number $x$, let $\lfloor x\rfloor\equiv\max_{d\in\spa{N}}\{d\colon d\leq x\}$ denote the largest integer not exceeding $x$.

Let $d_A\equiv\dim(\spa{H}_A)$ denote the dimension of a quantum system $A$ associated with a Hilbert space $\spa{H}_A$.  Let $\spa{B}_A$ denote the space of bounded linear operators acting on $\spa{H}_A$.  Let $\s{D}_A$ denote the set of quantum states of $A$ (i.e., the set of positive semidefinite unit-trace operators in $\spa{B}_A$), and let $\aff(\s{D}_A)$ denote the set of unit-trace Hermitian operators in $\spa{B}_A$ (i.e., the affine hull of $\s{D}_A$).  Let $\pi_A\equiv(1/d_A)\sum_{i=0}^{d_A-1}\op{i}{i}_A$ denote the uniform state of $A$.  Let $\Phi_{AA}\equiv(1/d_A)\sum_{i,j=0}^{d_A-1}\op{i}{j}_A\otimes\op{i}{j}_{A'}$ denote the standard maximally entangled state between two systems $A$ and $A'$ (with $A\cong A'$).  Let $\Upsilon_{M'M}\equiv(1/d_M)\sum_{m=0}^{d_M-1}\op{m}{m}_{M'}\otimes\op{m}{m}_M$ denote the standard classically maximally correlated state between two systems $M'$ and $M$ (with $M'\cong M$).  Let $\Omega_{M'\g{M}}\equiv\sum_{m=0}^{d_M-1}\op{m}{m}_{M'}\otimes\op{m}{m}_\g{M}$ denote the classical comparator between $M'$ and $\g{M}$ (with $M'\cong\g{M}$).

Let $\spa{L}_{A\to B}$ denote the space of linear maps from $\spa{B}_A$ to $\spa{B}_B$.  For a linear map $\ch{N}_{A\to B}\in\spa{L}_{A\to B}$, its Choi state is defined as $\Phi_{A'B}^\ch{N}\equiv\ch{N}_{A\to B}[\Phi_{A'A}]$.  Let $\s{C}_{A\to B}$ denote the set of quantum channels from $A$ to $B$ (i.e., the set of completely positive trace-preserving linear maps in $\spa{L}_{A\to B}$), and let $\s{L}_{A\to B}$ denote the set of completely positive maps in $\spa{L}_{A\to B}$.  For a unit-trace Hermitian operator $\sigma_B\in\aff(\s{D}_B)$, a corresponding replacement map $\ch{R}_{A\to B}^\sigma$ is defined by $\ch{R}_{A\to B}^\sigma[\cdot]\equiv\tr_A[\cdot]\sigma_B$, whose Choi state is given by $\Phi_{A'B}^{\ch{R}^\sigma}=\pi_{A'}\otimes\sigma_B$.

\section{Max-information and Doeblin information}
\label{sec:information}

The \emph{max-relative entropy} between two states $\rho_A,\sigma_A\in\s{D}_A$ is defined as~\cite{datta2009MinMaxrelativeEntropies}
\begin{align}
	D_{\max}\fleft(\rho\middle\|\sigma\fright)&\coloneq\inf_{\lambda\in\spa{R}}\left\{\log_2\lambda\colon\rho\leq\lambda\sigma\right\}.
\end{align}
We furthermore use the same definitions as above for the max-relative entropy when the arguments $\rho_A,\sigma_A\in\aff(\s{D}_A)$ are Hermitian operators that are not necessarily positive semidefinite.

The \emph{max-information} of a channel $\ch{N}_{A\to B}\in\s{C}_{A\to B}$ is defined as~\cite{konig2009OperationalMeaningMin,gupta2015MultiplicativityCompletelyBounded,diaz2018UsingReusingCoherence}
\begin{subequations}
\label{eq:max}
\begin{align}
	I_{\max}\fleft(\ch{N}\fright)&\coloneq\inf_{\sigma_B\in\s{D}_B}D_{\max}\fleft(\Phi^\ch{N}\middle\|\Phi^{\ch{R}^\sigma}\fright) \\
	&=\inf_{\substack{\lambda\in\spa{R}, \\ \sigma_B\in\s{D}_B}}\left\{\log_2\lambda\colon\Phi_{A'B}^\ch{N}\leq\lambda\pi_{A'}\otimes\sigma_B\right\}. \label{eq:max-SDP}
\end{align}
\end{subequations}
The \emph{Doeblin information} of a channel $\ch{N}_{A\to B}\in\s{C}_{A\to B}$ is defined as~\cite{george2025QuantumDoeblinCoefficients}
\begin{subequations}
\label{eq:doeblin}
\begin{align}
	I_\abb{doe}\fleft(\ch{N}\fright)&\coloneq\inf_{\tau_B\in\aff\fleft(\s{D}_B\fright)}D_{\max}\fleft(\Phi^{\ch{R}^\tau}\middle\|\Phi^\ch{N}\fright) \\
	&=\inf_{\substack{\lambda\in\spa{R}, \\ \tau_B\in\aff\fleft(\s{D}_B\fright)}}\left\{\log_2\lambda\colon\pi_{A'}\otimes\tau_B\leq\lambda\Phi_{A'B}^\ch{N}\right\}. \label{eq:doeblin-SDP}
\end{align}
\end{subequations}
Note that the infimization in $\tau_B$ above is over all unit-trace Hermitian operators.  We furthermore use the same definitions as above, for the max- and Doeblin informations, when the argument $\ch{N}_{A\to B}\in\s{L}_{A\to B}$ is a completely positive map that is not necessarily trace preserving.

\begin{remark}
For a function $\gd{D}\colon\bigcup_A\aff(\s{D}_A)\times\aff(\s{D}_A)\to\spa{R}\cup\{\infty\}$ satisfying the data-processing inequality (DPI) on pairs of unit-trace Hermitian operators, one can define a corresponding mutual information $\gd{I}$ of a channel $\ch{N}_{A\to B}\in\s{C}_{A\to B}$ as follows (see, e.g., Ref.~\cite[Definition~7.8.5]{khatri2024PrinciplesQuantumCommunication} for when $\gd{D}$ satisfies the DPI on pairs of states):
\begin{subequations}
\label{eq:mutual}
\begin{align}
	\gd{I}\fleft(\ch{N}\fright)&\coloneq\sup_{\rho_{RA}\in\s{D}_{RA}}\inf_{\sigma_B\in\aff\fleft(\s{D}_B\fright)}\gd{D}\fleft(\ch{N}_{A\to B}\fleft[\rho_{RA}\fright]\middle\|\rho_R\otimes\sigma_B\fright) \\
	&=\sup_{\rho_{A'}\in\s{D}_{A'}}\inf_{\sigma_B\in\aff\fleft(\s{D}_B\fright)}\gd{D}\fleft(\rho_{A'}^\frac{1}{2}\Phi_{A'B}^\ch{N}\rho_{A'}^\frac{1}{2}\middle\|\rho_{A'}^\frac{1}{2}\Phi_{A'B}^{\ch{R}^\sigma}\rho_{A'}^\frac{1}{2}\fright). \label{eq:mutual-1}
\end{align}
\end{subequations}
One can also define a corresponding ``reverse'' mutual information $\gd{I}^\abb{rev}$ of a channel as the mutual information corresponding to the reverse function $\gd{D}^\abb{rev}\colon(\rho,\sigma)\mapsto\gd{D}(\sigma\|\rho)$ (see, e.g., Refs.~\cite{nuradha2025MultivariateFidelities,ji2025BarycentricBoundsError,girardi2025QuantumUmlautInformation} for specific instances):
\begin{subequations}
\label{eq:reverse}
\begin{align}
	\gd{I}^\abb{rev}\fleft(\ch{N}\fright)&\coloneq\sup_{\rho_{RA}\in\s{D}_{RA}}\inf_{\tau_B\in\aff\fleft(\s{D}_B\fright)}\gd{D}\fleft(\rho_R\otimes\tau_B\middle\|\ch{N}_{A\to B}\fleft[\rho_{RA}\fright]\fright) \\
	&=\sup_{\rho_{A'}\in\s{D}_{A'}}\inf_{\tau_B\in\aff\fleft(\s{D}_B\fright)}\gd{D}\fleft(\rho_{A'}^\frac{1}{2}\Phi_{A'B}^{\ch{R}^\tau}\rho_{A'}^\frac{1}{2}\middle\|\rho_{A'}^\frac{1}{2}\Phi_{A'B}^\ch{N}\rho_{A'}^\frac{1}{2}\fright), \label{eq:rev-1}
\end{align}
\end{subequations}
where Eq.~\eqref{eq:rev-1} follows from applying Eq.~\eqref{eq:mutual-1} to $\gd{D}^\abb{rev}$.  Then following a similar argument to Ref.~\cite[Remark~2]{fang2020QuantumChannelSimulation} (also see Ref.~\cite[Lemma~12]{wilde2020AmortizedChannelDivergence}), it can be inferred that the max- and Doeblin informations in Eqs.~\eqref{eq:max} and \eqref{eq:doeblin} are equivalent, respectively, to the mutual and reverse mutual information in Eqs.~\eqref{eq:mutual} and \eqref{eq:reverse} corresponding to the max-relative entropy defined on pairs of unit-trace Hermitian operators.

The max-information has several operational interpretations in quantum Shannon theory, most notably as the one-shot zero-error nonsignaling-assisted classical simulation cost~\cite{duan2016NosignallingassistedZeroerrorCapacity} and as the entanglement-assisted logarithmic communication value~\cite{chitambar2023CommunicationValueQuantum}.  The Doeblin information has found applications in bounding contraction coefficients and consequently a variety of other contexts such as quantum machine learning, quantum error mitigation, and sample complexity of noisy quantum hypothesis testing~\cite{george2025QuantumDoeblinCoefficients}.  From a resource-theoretic perspective~\cite{chitambar2019QuantumResourceTheories,gour2025QuantumResourceTheories}, the max- and Doeblin informations can also be interpreted, respectively, as the logarithmic resource robustness~\cite{vidal1999RobustnessEntanglement,harrow2003RobustnessQuantumGates} and the logarithmic resource weight~\cite{elitzur1992QuantumNonlocalityEach,lewenstein1998SeparabilityEntanglementComposite} against the affine hull of replacement channels~\cite{datta2009MaxrelativeEntropyEntanglement,george2025QuantumDoeblinCoefficients}.
\end{remark}

For our goal of characterizing the retrocausal capacity of a channel, it is of particular relevance that the max- and Doeblin informations have operational interpretations in terms of the maximum and minimum achievable singlet fractions, respectively~\cite{konig2009OperationalMeaningMin,george2025QuantumDoeblinCoefficients}, as reviewed in the following lemma.

\begin{lemma}[\cite{konig2009OperationalMeaningMin,george2025QuantumDoeblinCoefficients}]
\label{lem:fraction}
The max- and Doeblin informations can be interpreted in terms of the maximum and minimum achievable singlet fractions, respectively: for a completely positive map $\ch{N}_{A\to B}\in\s{L}_{A\to B}$,
\begin{align}
	I_{\max}\fleft(\ch{N}\fright)&=2\log_2d_A+\log_2\max_{\ch{K}_{B\to A}\in\s{C}_{B\to A}}\tr\fleft[\Phi_{AA'}\left(\ch{K}_{B\to A}\circ\ch{N}_{A\to B}\right)\fleft[\Phi_{AA'}\fright]\fright], \\
	I_\abb{doe}\fleft(\ch{N}\fright)&=-2\log_2d_A-\log_2\min_{\ch{K}_{B\to A}\in\s{C}_{B\to A}}\tr\fleft[\Phi_{AA'}\left(\ch{K}_{B\to A}\circ\ch{N}_{A\to B}\right)\fleft[\Phi_{AA'}\fright]\fright].
\end{align}
\end{lemma}

The max-information is additive in the sense that $I_{\max}(\ch{N}_1\otimes\ch{N}_2)=I_{\max}(\ch{N}_1)+I_{\max}(\ch{N}_2)$ for two arbitrary completely positive maps $\ch{N}_1$ and $\ch{N}_2$~\cite[p.~4342]{konig2009OperationalMeaningMin} (also see Ref.~\cite[Lemma~6]{gupta2015MultiplicativityCompletelyBounded}).  This implies that, for every $n\in\spa{N}$,
\begin{align}
	\frac{1}{n}I_{\max}\fleft(\ch{N}^{\otimes n}\fright)&=I_{\max}\fleft(\ch{N}\fright). \label{eq:max-regularized}
\end{align}
On the other hand, the Doeblin information is nonadditive in general, despite being superadditive~\cite[Proposition~20]{george2025QuantumDoeblinCoefficients}: $I_\abb{doe}(\ch{N}_1\otimes\ch{N}_2)\geq I_\abb{doe}(\ch{N}_1)+I_\abb{doe}(\ch{N}_2)$.  By Fekete's lemma~\cite{fekete1923UeberVerteilungWurzeln}, the following limit exists, which we term the \emph{regularized Doeblin information} of $\ch{N}_{A\to B}$:
\begin{align}
	I_\abb{doe}^\infty\fleft(\ch{N}\fright)&\coloneq\lim_{n\to\infty}\frac{1}{n}I_\abb{doe}\fleft(\ch{N}^{\otimes n}\fright). \label{eq:doeblin-regularized}
\end{align}
A slight generalization of Ref.~\cite[Theorems~3 and 4]{george2025QuantumDoeblinCoefficients}, shows that, for quantum-to-classical completely positive maps (which include measurement channels) and mictodiactic covariant completely positive maps (which include depolarizing channels), the Doeblin information is additive and thus equal to its regularization.  For general completely positive maps, the Doeblin information is upper bounded by the following quantity~\cite{george2025QuantumDoeblinCoefficients}:
\begin{align}
	I_\abb{doe}\fleft(\ch{N}\fright)&\leq\inf_{\substack{\lambda\in\spa{R}, \\ \tau_B\in\aff\fleft(\s{D}_B\fright)}}\left\{\log_2\lambda\colon-\lambda\Phi_{A'B}^\ch{N}\leq\pi_{A'}\otimes\tau_B\leq\lambda\Phi_{A'B}^\ch{N}\right\} \\
	&\eqcolon I_\wang\fleft(\ch{N}\fright), \label{eq:wang}
\end{align}
which can be shown to be additive by slightly generalizing Ref.~\cite[Proposition~24]{george2025QuantumDoeblinCoefficients}.  Therefore, since the Doeblin information is superadditive, we can bound the regularized Doeblin information from both sides in terms of single-letter quantities that are efficiently computable via semidefinite programs (SDPs):
\begin{align}
	I_\abb{doe}\fleft(\ch{N}\fright)&\leq I_\abb{doe}^\infty\fleft(\ch{N}\fright)\leq I_\wang\fleft(\ch{N}\fright). \label{eq:bounds}
\end{align}

\section{Postselected closed timelike curves (P-CTC\lowercase{s})}
\label{sec:curves}

\subsection{Noiseless P-CTCs}
\label{sec:noiseless}

As proposed in Ref.~\cite{bennett2002LectureNotes} and formalized in Refs.~\cite{lloyd2011ClosedTimelikeCurves,lloyd2011QuantumMechanicsTime,svetlichny2009EffectiveQuantumTime}, enclosing a linear map $\ch{T}_{EA\to FA}\in\spa{L}_{EA\to FA}$ with a noiseless P-CTC on the system $A$ can be simulated, up to renormalization, by:
\begin{enumerate}
	\item preparing a standard maximally entangled state $\Phi_{AA'}$ between $A$ and $A'$;
	\item letting the map $\ch{T}_{EA\to FA}$ act on $A$ along with $E$, receiving $A$ back along with $F$;
	\item projecting $A$ and $A'$ onto $\Phi_{AA'}$ and discarding them.
\end{enumerate}
The above simulation can be represented by a linear supermap (i.e., a linear map on linear maps~\cite{chiribella2008TransformingQuantumOperations}) $\Gamma_A$ from $\spa{L}_{A\to A}$ to $\spa{C}$, which we term the \emph{loop supermap}, defined as follows:
\begin{subequations}
\label{eq:loop}
\begin{align}
	\Gamma_A\colon\ch{T}_{EA\to FA}&\mapsto\Gamma_A\fleft\{\ch{T}_{EA\to FA}\fright\}, \\
	\Gamma_A\fleft\{\ch{T}_{EA\to FA}\fright\}\colon\rho_{RE}&\mapsto\tr_{AA'}\fleft[\Phi_{AA'}\ch{T}_{EA\to FA}\fleft[\rho_{RE}\otimes\Phi_{AA'}\fright]\fright],
\end{align}
\end{subequations}
where we use the curly bracket to indicate the action of a supermap on a map.  Then the actual transformation from $E$ to $F$, resulting from enclosing $\ch{T}_{EA\to FA}$ with the noiseless P-CTC, is represented by the following nonlinear map from $\spa{B}_{RE}$ to $\spa{B}_{RF}$, with $R$ an arbitrary reference system, which renormalizes $\Gamma_A\{\ch{T}_{EA\to FA}\}$:
\begin{align}
	\rho_{RE}&\mapsto\frac{\Gamma_A\fleft\{\ch{T}_{EA\to FA}\fright\}\fleft[\rho_{RE}\fright]}{\tr\fleft[\Gamma_A\fleft\{\ch{T}_{EA\to FA}\fright\}\fleft[\rho_{RE}\fright]\fright]}. \label{eq:noiseless}
\end{align}
Note that, if $\ch{T}_{EA\to FA}$ is completely positive, then so is $\Gamma_A\{\ch{T}_{EA\to FA}\}$~\cite{burniston2020NecessarySufficientConditions}.  It is thus immediate that the nonlinear map in Eq.~\eqref{eq:noiseless} preserves the set of states in the presence of an arbitrary reference system $R$.

\begin{figure}[t]
\subfloat[\label{fig:noiseless-curve}]{\includegraphics[scale=0.3]{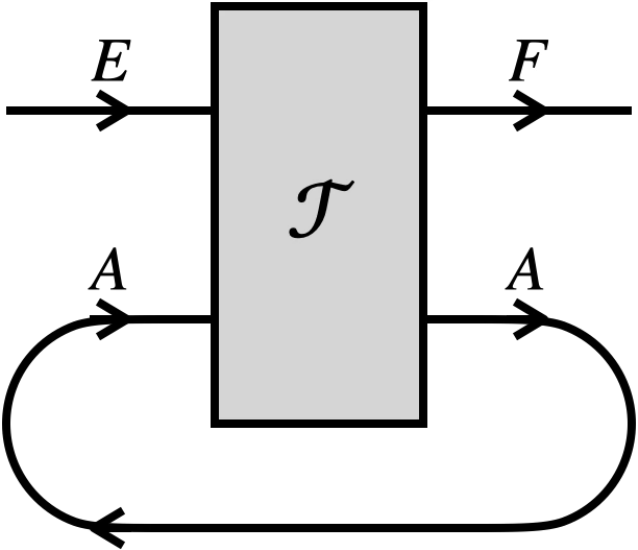}} \qquad\qquad\qquad
\subfloat[\label{fig:noiseless-postselection}]{\includegraphics[scale=0.3]{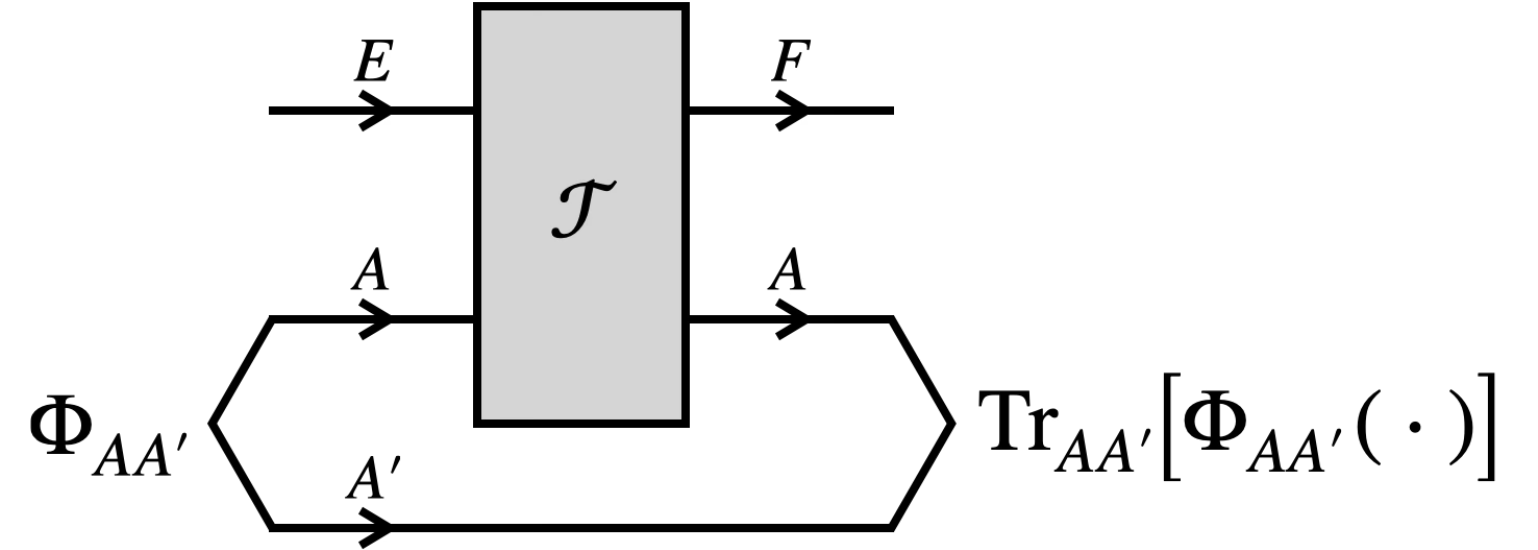}}
\caption{Enclosing a linear map $\ch{T}_{EA\to FA}$ with a noiseless P-CTC on the system $A$ (Fig.~\ref{fig:noiseless-curve}) is equivalent to performing postselected teleportation from its output $A$ to its input $A$ (Fig.~\ref{fig:noiseless-postselection} with renormalization).}
\label{fig:noiseless}
\end{figure}

\begin{fact}[Noiselessness]
\label{fact:noiselessness}
Surrounding a swap channel between two systems with a noiseless P-CTC on one system results in an identity channel on the other: for two systems $A$ and $B$ with $A\cong B$ and a linear operator $\rho_{RA}\in\spa{B}_{RA}$ with $R$ a reference system,
\begin{align}
	\frac{\Gamma_B\fleft\{\ch{P}_{AB\to AB}\fright\}\fleft[\rho_{RA}\fright]}{\tr\fleft[\Gamma_B\fleft\{\ch{P}_{AB\to AB}\fright\}\fleft[\rho_{RA}\fright]\fright]}&=\rho_{RA},
\end{align}
where $\ch{P}_{A_IB_I\to A_OB_O}\equiv\id_{A_I\to B_O}\otimes\id_{B_I\to A_O}$ is the swap channel between $A$ and $B$.
\end{fact}

\begin{proof}
It follows from the definition of the loop supermap [Eq.~\eqref{eq:loop}] that
\begin{align}
	\Gamma_B\fleft\{\ch{P}_{AB\to AB}\fright\}\fleft[\rho_{RA}\fright]&=\tr_{BB'}\fleft[\Phi_{BB'}\ch{P}_{AB\to AB}\fleft[\rho_{RA}\otimes\Phi_{BB'}\fright]\fright] \\
	&=\tr_{BB'}\fleft[\Phi_{BB'}\left(\rho_{RB}\otimes\Phi_{AB'}\right)\fright] \\
	&=\frac{1}{d_A^2}\rho_{RA}, \label{pf:noiselessness-1}
\end{align}
where Eq.~\eqref{pf:noiselessness-1} follows from probabilistic teleportation~\cite{bennett1993TeleportingUnknownQuantum}.  The desired statement is then immediate.
\end{proof}

\begin{fact}[Composition]
\label{fact:composition}
Parallel composition of noiseless P-CTCs is represented by tensor product: for two systems $A$ and $B$, we have that $\Gamma_{AB}=\Gamma_A\otimes\Gamma_B$.
\end{fact}

\begin{proof}
The desired statement follows immediately from the fact that $\Phi_{ABA'B'}=\Phi_{AA'}\otimes\Phi_{BB'}$.
\end{proof}

\begin{lemma}[Cyclicity]
\label{lem:cyclicity}
Noiseless P-CTCs have the following cyclicity property (see Fig.~\ref{fig:cyclicity}): for two linear maps $\ch{T}_{EB\to FA}\in\spa{L}_{EB\to FA}$ and $\ch{W}_{GA\to HB}\in\spa{L}_{GA\to HB}$ and a linear operator $\rho_{REG}\in\spa{B}_{REG}$ with $R$ a reference system, we have that
\begin{align}
	\frac{\Gamma_A\fleft\{\ch{T}_{EB\to FA}\circ\ch{W}_{GA\to HB}\fright\}\fleft[\rho_{REG}\fright]}{\tr\fleft[\Gamma_A\fleft\{\ch{T}_{EB\to FA}\circ\ch{W}_{GA\to HB}\fright\}\fleft[\rho_{REG}\fright]\fright]}&=\frac{\Gamma_B\fleft\{\ch{W}_{GA\to HB}\circ\ch{T}_{EB\to FA}\fright\}\fleft[\rho_{REG}\fright]}{\tr\fleft[\Gamma_B\fleft\{\ch{W}_{GA\to HB}\circ\ch{T}_{EB\to FA}\fright\}\fleft[\rho_{REG}\fright]\fright]}.
\end{align}
\end{lemma}

\begin{figure}[t]
\subfloat[\label{fig:cyclicity-1}]{\includegraphics[scale=0.3]{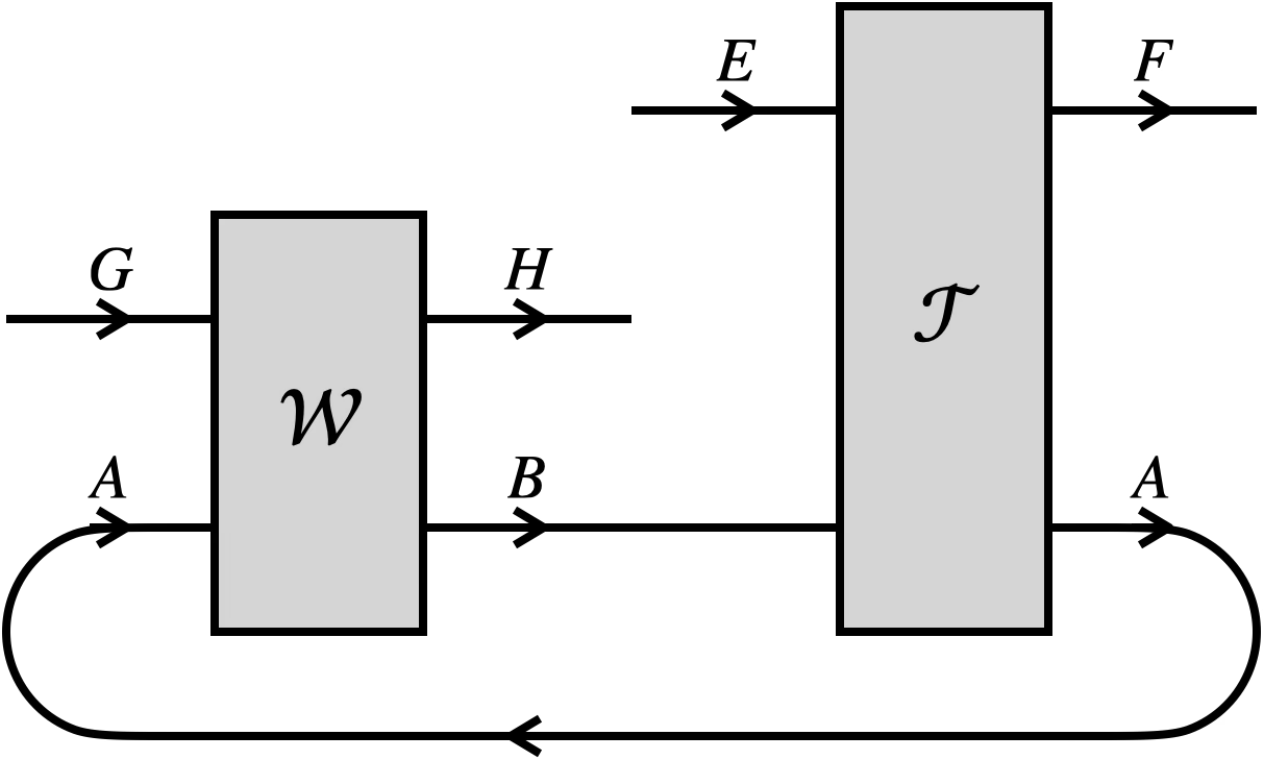}} \qquad\qquad\qquad
\subfloat[\label{fig:cyclicity-2}]{\includegraphics[scale=0.3]{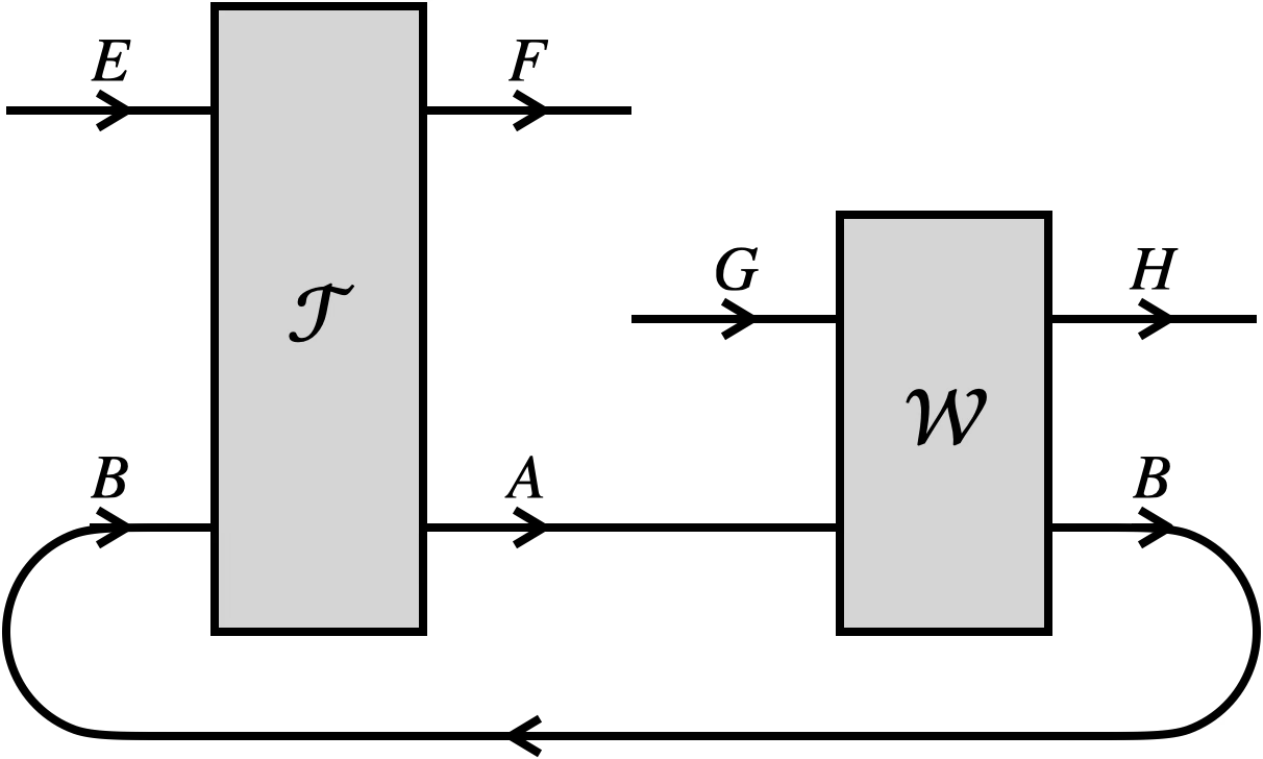}}
\caption{The cyclicity of noiseless P-CTCs dictates that Figs.~\ref{fig:cyclicity-1} and \ref{fig:cyclicity-2} are equivalent.}
\label{fig:cyclicity}
\end{figure}

\begin{proof}
Since Eq.~\eqref{pf:noiselessness-1} implies that $\Gamma_{B'}\{\ch{P}_{BB'\to BB'}\}=\frac{1}{d_B^2}\id_{B\to B}$ for $B'\cong B$, where $\ch{P}_{BB'\to BB'}$ is the swap channel between $B$ and $B'$, it follows that
\begin{align}
	\Gamma_A\fleft\{\ch{T}_{EB\to FA}\circ\ch{W}_{GA\to HB}\fright\}&=d_B^2\Gamma_A\fleft\{\ch{T}_{EB\to FA}\circ\Gamma_{B'}\fleft\{\ch{P}_{BB'}\fright\}\circ\ch{W}_{GA\to HB}\fright\} \\
	&=d_B^2\left(\Gamma_A\otimes\Gamma_{B'}\right)\fleft\{\ch{T}_{EB\to FA}\circ\ch{P}_{BB'\to BB'}\circ\ch{W}_{GA\to HB}\fright\} \\
	&=d_B^2\Gamma_{AB'}\fleft\{\ch{T}_{EB\to FA}\circ\ch{P}_{BB'\to BB'}\circ\ch{W}_{GA\to HB}\fright\}. \label{pf:cyclicity-1}
\end{align}
Considering that
\begin{align}
	\ch{T}_{EB_O\to FA_O}\circ\ch{P}_{B_IB_I'\to B_OB_O'}\circ\ch{W}_{GA_I\to HB_I}&=\ch{T}_{EB_O\to FA_O}\circ\left(\id_{B_I\to B_O'}\otimes\id_{B_I'\to B_O}\right)\circ\ch{W}_{GA_I\to HB_I} \\
	&=\ch{T}_{EB_I'\to FA_O}\circ\ch{W}_{GA_I\to HB_O'} \\
	&=\ch{T}_{EB_I'\to FA_O}\otimes\ch{W}_{GA_I\to HB_O'},
\end{align}
it follows from Eq.~\eqref{pf:cyclicity-1} that
\begin{align}
	\Gamma_A\fleft\{\ch{T}_{EB\to FA}\circ\ch{W}_{GA\to HB}\fright\}&=d_B^2\Gamma_{AB'}\fleft\{\ch{T}_{EB'\to FA}\otimes\ch{W}_{GA\to HB'}\fright\}. \label{pf:cyclicity-2}
\end{align}
Following the same reasoning, we can also show for $A'\cong A$ that
\begin{align}
	\Gamma_B\fleft\{\ch{W}_{GA\to HB}\circ\ch{T}_{EB\to FA}\fright\}&=d_A^2\Gamma_{A'B}\fleft\{\ch{W}_{GA'\to HB}\otimes\ch{T}_{EB\to FA'}\fright\}. \label{pf:cyclicity-3}
\end{align}
Combining Eqs.~\eqref{pf:cyclicity-2} and \eqref{pf:cyclicity-3} yields the desired statement.
\end{proof}

\subsection{Noisy P-CTCs}
\label{sec:noisy}

Compared to the main text, here we consider a slightly more generalized notion of a noisy P-CTC, which can be represented by a completely positive (but not necessarily trace-preserving) map.  As illustrated in Fig.~\ref{fig:noisy-curve-SM}, a noisy P-CTC from $A$ to $B$, represented by a completely positive map $\ch{N}_{A\to B}\in\s{L}_{A\to B}$, can be understood as the concatenation of a noiseless P-CTC on $A$ and the map $\ch{N}_{A\to B}$.  In other words, enclosing a linear map $\ch{T}_{EB\to FA}\in\spa{L}_{EB\to FA}$ with the noisy P-CTC amounts to preprocessing $\ch{T}_{EB\to FA}$ with $\ch{N}_{A\to B}$ and enclosing $\ch{T}_{EB\to FA}\circ\ch{N}_{A\to B}$ with a noiseless P-CTC on $A$.  Formally, the resulting transformation from $E$ to $F$ is represented by the following nonlinear map from $\spa{B}_{RE}$ to $\spa{B}_{RF}$, with $R$ an arbitrary reference system:
\begin{align}
	\rho_{RE}&\mapsto\frac{\Gamma_A\fleft\{\ch{T}_{EB\to FA}\circ\ch{N}_{A\to B}\fright\}\fleft[\rho_{RE}\fright]}{\tr\fleft[\Gamma_A\fleft\{\ch{T}_{EB\to FA}\circ\ch{N}_{A\to B}\fright\}\fleft[\rho_{RE}\fright]\fright]}. \label{eq:noisy}
\end{align}
Note that the noisy P-CTC can equivalently be realized by postprocessing $\ch{T}_{EB\to FA}$ with $\ch{N}_{A\to B}$ and enclosing $\ch{N}_{A\to B}\circ\ch{T}_{EB\to FA}$ with a noiseless P-CTC on $B$, due to the cyclicity of noiseless P-CTCs (Lemma~\ref{lem:cyclicity}).

\begin{figure}[t]
\subfloat[\label{fig:noisy-curve-SM}]{\includegraphics[scale=0.3]{figures/noisy-curve}} \qquad\qquad\qquad
\subfloat[\label{fig:noisy-final}]{\includegraphics[scale=0.3]{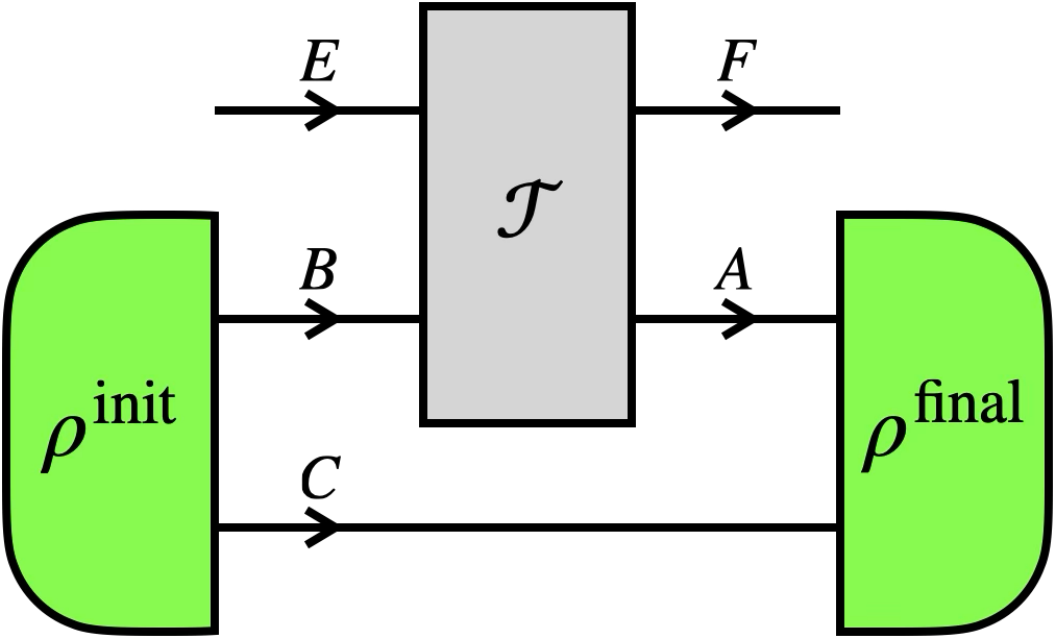}}
\caption{Enclosing a linear map $\ch{T}_{EB\to FA}$ with a noisy P-CTC represented by a completely positive map $\ch{N}_{A\to B}$ is equivalent to preprocessing $\ch{T}_{EB\to FA}$ with $\ch{N}_{A\to B}$ and enclosing $\ch{T}_{EB\to FA}\circ\ch{N}_{A\to B}$ with a noiseless P-CTC on the system $A$ (Fig.~\ref{fig:noisy-curve-SM}).  It is also equivalent to imposing on $\ch{T}_{EB\to FA}$  initial- and final-state boundary conditions satisfying Eq.~\eqref{eq:final}  (Fig.~\ref{fig:noisy-final} with renormalization).}
\label{fig:noisy-SM}
\end{figure}

A major reason for our interest in noisy P-CTCs represented by completely positive maps, and not just by channels, is that they naturally arise in a time-symmetric modification of quantum mechanics~\cite{aharonov1964TimeSymmetryQuantum,griffiths1984ConsistentHistoriesInterpretation,gell-mann1994TimeSymmetryAsymmetry}, where arbitrary initial- and final-state boundary conditions can be specified.  A general configuration to consider is one where a linear map $\ch{T}_{EB\to FA}\in\spa{L}_{EB\to FA}$ is subject to an initial state $\rho_{BC}^\abb{init}\in\s{D}_{BC}$ and  a final state $\rho_{AC}^\abb{final}\in\s{D}_{AC}$, where $C$ is a side quantum memory (see Fig.~\ref{fig:noisy-final}).  The resulting transformation from $E$ to $F$ is then represented by the following nonlinear map, with $R$ an arbitrary reference system:
\begin{align}
	\rho_{RE}&\mapsto\frac{\tr_{AC}\fleft[\rho_{AC}^\abb{final}\ch{T}_{EB\to FA}\fleft[\rho_{RE}\otimes\rho_{BC}^\abb{init}\fright]\fright]}{\tr\fleft[\rho_{AC}^\abb{final}\ch{T}_{EB\to FA}\fleft[\rho_{RE}\otimes\rho_{BC}^\abb{init}\fright]\fright]}. \label{eq:final}
\end{align}
It is easy to verify that Eq.~\eqref{eq:final} recovers Eq.~\eqref{eq:noisy} upon defining the following completely positive but not necessarily trace-preserving map:
\begin{align}
	\ch{N}_{A\to B}\fleft[\cdot\fright]&\coloneq\tr_{AC}\fleft[\rho_{AC}^\abb{final}\left(\left(\cdot\right)_A\otimes\rho_{BC}^\abb{init}\right)\fright],
\end{align}
which represents a noisy P-CTC from $A$ to $B$.  Conversely, every noisy P-CTC represented by a completely positive map can be realized with some initial- and final-state boundary condition.  To see this, for a completely positive map $\ch{N}_{A\to B}\in\s{L}_{A\to B}$, by setting
\begin{align}
	\rho_{BC}^\abb{init}&\coloneq\frac{\ch{N}_{A\to B}\fleft[\Phi_{AC}\fright]}{\tr\fleft[\ch{N}_{A\to B}\fleft[\Phi_{AC}\fright]\fright]}, \\
	\rho_{AC}^\abb{final}&\coloneq\Phi_{AC},
\end{align}
one realizes a noisy P-CTC represented by $\ch{N}_{A\to B}/\tr\fleft[\ch{N}_{A\to B}\fleft[\Phi_{AC}\fright]\fright]$, and this noisy P-CTC is equivalent to the one represented by $\ch{N}_{A\to B}$, due to renormalization.

\section{Retrocausal communication}
\label{sec:communication}

\subsection{General framework}
\label{sec:framework}

The setting of retrocausal communication involves two parties, namely a sender in the future and a receiver in the past.  Following the convention of the main text, we refer to the sender as a father and the receiver as his daughter.  The father is connected to the daughter by a noisy P-CTC represented by a given completely positive map $\ch{N}_{A\to B}\in\s{L}_{A\to B}$, and they aim to transmit quantum or classical messages backward in time from the father to the daughter through this noisy P-CTC.  They can adopt any strategy for communication as long as they do not violate chronology when not invoking the noisy P-CTC.

The freedom for the father and the daughter to adopt any strategy without violating chronology brings forth the following possibility.  Since the daughter precedes the father in time, she is naturally able to signal to him, most generally by storing certain quantum information in a noiseless quantum memory of an arbitrary dimension and leaving the memory for him to retrieve in the future.  Note that such a memory, which enables signaling from the daughter to the father, combined with the noisy P-CTC from the father to daughter, could incur a causal loop.  However, grandfather-type paradoxes can be effectively circumvented within the model of P-CTCs~\cite{lloyd2011ClosedTimelikeCurves}, guaranteeing that the setting we consider is well defined.  In fact, the potential formation of such a causal loop is what fundamentally distinguishes retrocausal communication from conventional (i.e., forward-in-time) communication settings (such as unassisted classical~\cite{shannon1948MathematicalTheoryCommunication,holevo1973BoundsQuantityInformation,holevo1998CapacityQuantumChannel,schumacher1997SendingClassicalInformation} and quantum~\cite{shor1995SchemeReducingDecoherence,schumacher1996SendingEntanglementNoisy,schumacher1996QuantumDataProcessing,lloyd1997CapacityNoisyQuantum,barnum1998InformationTransmissionNoisy,barnum2000QuantumFidelitiesChannel,shor2002QuantumChannelCapacity,devetak2005PrivateClassicalCapacity} communication, entanglement-assisted communication~\cite{bennett1999EntanglementassistedClassicalCapacity,bennett2002EntanglementassistedCapacityQuantum,holevo2002EntanglementassistedClassicalCapacity}, and nonsignaling-assisted communication~\cite{cubitt2011ZeroerrorChannelCapacity,matthews2012LinearProgramFinite,duan2016NosignallingassistedZeroerrorCapacity,takagi2020ApplicationResourceTheory}), and it is the reason why a channel's retrocausal capacity can exceed its capacity in any conventional communication setting.

We now put the discussion in formal terms.  As illustrated in Fig.~\ref{fig:retrocausal}, the father and the daughter's strategy is, in general, represented by a pair of channels, $(\ch{E}_{ML\to A},\ch{D}_{B\to\g{M}L})$ (with $M\cong\g{M}$), corresponding to the father's encoding and the daughter's decoding, respectively, where:
\begin{itemize}
	\item $M$ is the system storing the father's message;
	\item $\g{M}$ is the system storing the daughter's recovery of the message;
	\item $L$ is the noiseless quantum memory that the daughter leaves to the father.
\end{itemize}
The strategy operates as follows.
\begin{itemize}
	\item \textbf{Encoding.}  In the future, the father retrieves $L$ (which was left by the daughter in the past), encodes his message with the assistance of $L$ by applying the channel $\ch{E}_{ML\to A}$, and feeds $A$ into the entrance of the noisy P-CTC.
	\item \textbf{Decoding.}  In the past, the daughter receives $B$ from the exit of the noisy P-CTC, decodes the message by applying the channel $\ch{D}_{B\to\g{M}L}$, and leaves $L$ for the father to retrieve in the future.
\end{itemize}

\begin{figure}[t]
\includegraphics[scale=0.3]{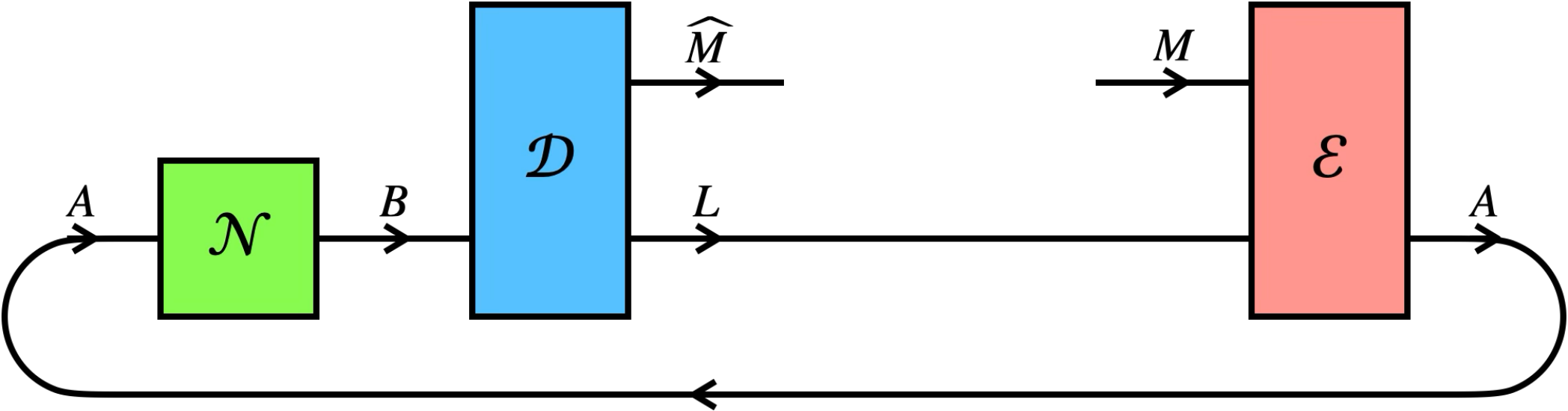}
\caption{In retrocausal communication, a father in the future (i.e., the sender) is connected to  his daughter in the past (i.e., the receiver) by a noisy P-CTC represented by a completely positive map $\ch{N}_{A\to B}$ (green).  The father applies a channel $\ch{E}_{ML\to A}$ as his encoding (red), and the daughter applies a channel $\ch{D}_{B\to\g{M}L}$ as her decoding (blue).  As part of the strategy, the daughter can signal to the father via a noiseless quantum memory $L$.}
\label{fig:retrocausal}
\end{figure}

Combining the strategy $(\ch{E}_{ML\to A},\ch{D}_{B\to\g{M}L})$ as described above with the noisy P-CTC represented by the given completely positive map $\ch{N}_{A\to B}$, a transformation from $M$ to $\g{M}$ is realized, represented by a nonlinear map $\ch{M}_{M\to\g{M}}$ from $\spa{B}_{RM}$ to $\spa{B}_{R\g{M}}$, with $R$ an arbitrary reference system, defined as follows:
\begin{align}
	\ch{M}_{M\to\g{M}}&\colon\rho_{RM}\mapsto\frac{\Gamma_A\fleft\{\ch{E}_{ML\to A}\circ\ch{D}_{B\to\g{M}L}\circ\ch{N}_{A\to B}\fright\}\fleft[\rho_{RM}\fright]}{\tr\fleft[\Gamma_A\fleft\{\ch{E}_{ML\to A}\circ\ch{D}_{B\to\g{M}L}\circ\ch{N}_{A\to B}\fright\}\fleft[\rho_{RM}\fright]\fright]}, \label{eq:retrocausal}
\end{align}
where $\Gamma_A$ is the loop supermap defined in Eq.~\eqref{eq:loop}.

\begin{remark}[Forward-in-time communication assisted by a noiseless P-CTC]
\label{rem:assisted}
The setting of retrocausal communication, where communication is conducted \emph{backward} in time (from a sender in the future to a receiver in the past), is equivalent to the setting where communication is conducted \emph{forward} in time (from a sender in the past to a receiver in the future) but assisted by a noiseless P-CTC from the receiver to the sender (see Fig.~\ref{fig:assisted}).  To see this, it follows from the cyclicity of noiseless P-CTCs (Lemma~\ref{lem:cyclicity}) that, for every linear operator $\rho_{RM}\in\spa{B}_{RM}$ with $R$ a reference system,
\begin{align}
	\frac{\Gamma_A\fleft\{\ch{E}_{ML\to A}\circ\ch{D}_{B\to\g{M}L}\circ\ch{N}_{A\to B}\fright\}\fleft[\rho_{RM}\fright]}{\tr\fleft[\Gamma_A\fleft\{\ch{E}_{ML\to A}\circ\ch{D}_{B\to\g{M}L}\circ\ch{N}_{A\to B}\fright\}\fleft[\rho_{RM}\fright]\fright]}&=\frac{\Gamma_L\fleft\{\ch{D}_{B\to\g{M}L}\circ\ch{N}_{A\to B}\circ\ch{E}_{ML\to A}\fright\}\fleft[\rho_{RM}\fright]}{\tr\fleft[\Gamma_L\fleft\{\ch{D}_{B\to\g{M}L}\circ\ch{N}_{A\to B}\circ\ch{E}_{ML\to A}\fright\}\fleft[\rho_{RM}\fright]\fright]}. \label{eq:assisted}
\end{align}
Note that the map from $\rho_{RM}$ to the right-hand side of Eq.~\eqref{eq:assisted} precisely represents the transformation from $M$ to $\g{M}$ resulting from adopting the strategy $(\ch{E}_{ML\to A},\ch{D}_{B\to\g{M}L})$ in a setting where communication is conducted forward in time through $\ch{N}_{A\to B}$ and assisted by a noiseless P-CTC on $L$.  Due to this equivalence, it can clearly be recognized that, mathematically, retrocausal communication reduces to entanglement-assisted communication~\cite{bennett1999EntanglementassistedClassicalCapacity,bennett2002EntanglementassistedCapacityQuantum,holevo2002EntanglementassistedClassicalCapacity} when $L$ is only used for entanglement distribution, and it reduces to unassisted communication~\cite{shannon1948MathematicalTheoryCommunication,holevo1973BoundsQuantityInformation,holevo1998CapacityQuantumChannel,schumacher1997SendingClassicalInformation,shor1995SchemeReducingDecoherence,schumacher1996SendingEntanglementNoisy,schumacher1996QuantumDataProcessing,lloyd1997CapacityNoisyQuantum,barnum1998InformationTransmissionNoisy,barnum2000QuantumFidelitiesChannel,shor2002QuantumChannelCapacity,devetak2005PrivateClassicalCapacity} when $L$ is not utilized at all.
\end{remark}

\begin{figure}[t]
\includegraphics[scale=0.3]{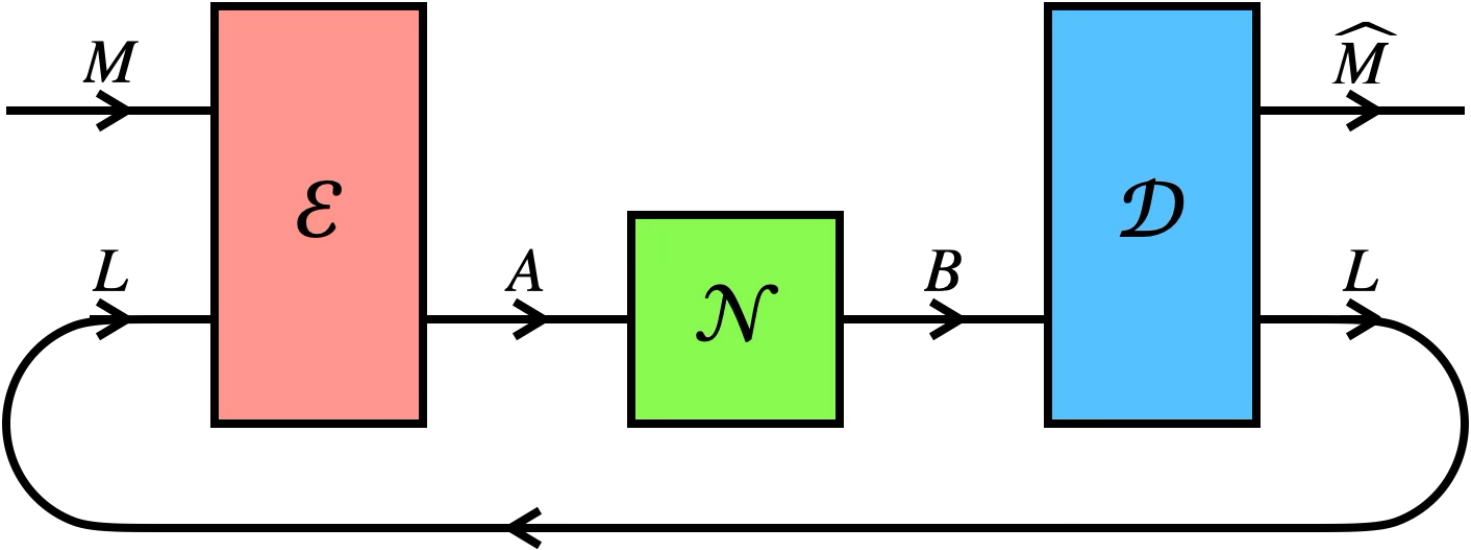}
\caption{In forward-in-time communication assisted by a noiseless P-CTC, the sender in the past is connected to the receiver in the future by a completely positive map $\ch{N}_{A\to B}$, and the receiver is connected to the sender by a noiseless P-CTC.  This setting is mathematically equivalent to retrocausal communication (see Fig.~\ref{fig:retrocausal}).}
\label{fig:assisted}
\end{figure}

\begin{remark}[Axiomatic interpretation]
\label{rem:axiomatic}
The setting of retrocausal communication can also be understood from an axiomatic perspective.  Specifically, in any communication setting, the sender and the receiver's strategy can be interpreted as a bipartite channel $\Theta_{MB\to A\g{M}}$, where $M$ and $A$ are held by the sender and $B$ and $\g{M}$ are held by the receiver (see, e.g., Ref.~\cite[Fig.~1]{duan2016NosignallingassistedZeroerrorCapacity} for an illustration).  Communication through $\ch{N}_{A\to B}$ by adopting this strategy can then be understood as connecting the output $A$ to the input $B$ of $\Theta_{MB\to A\g{M}}$ with $\ch{N}_{A\to B}$, resulting in a transformation from the sender's message system $M$ to the receiver's recovery system $\g{M}$.  Then the following can be said.
\begin{itemize}
	\item In standard (i.e., chronology-respecting) quantum theory, to avoid the formation of a causal loop when combining $\Theta_{MB\to A\g{M}}$ with $\ch{N}_{A\to B}$ for all $\ch{N}_{A\to B}$, a minimal constraint on the bipartite channel $\Theta_{MB\to A\g{M}}$ is that it should be nonsignaling from $B$ to $A$.  Under this constraint, combining $\Theta_{MB\to A\g{M}}$ with $\ch{N}_{A\to B}$ amounts to applying a superchannel to $\ch{N}_{A\to B}$~\cite{chiribella2008TransformingQuantumOperations,gour2019ComparisonQuantumChannels}.
	\item In addition to the above, in conventional communication settings, to ensure that the strategy itself does not introduce any communication resource from the sender to the receiver, it is also necessary to require that the bipartite channel $\Theta_{MB\to A\g{M}}$ be nonsignaling from $M$ to $\g{M}$.  This renders $\Theta_{MB\to A\g{M}}$ a two-way nonsignaling correlation and, if without further constraints, corresponds to the setting of nonsignaling-assisted communication~\cite{cubitt2011ZeroerrorChannelCapacity,matthews2012LinearProgramFinite,duan2016NosignallingassistedZeroerrorCapacity,takagi2020ApplicationResourceTheory}.  As more restrictive settings than nonsignaling-assisted communication, entanglement-assisted communication~\cite{bennett1999EntanglementassistedClassicalCapacity,bennett2002EntanglementassistedCapacityQuantum,holevo2002EntanglementassistedClassicalCapacity} corresponds to when $\Theta_{MB\to A\g{M}}$ is furthermore required to be realized by local operations and shared entanglement, and unassisted communication~\cite{shannon1948MathematicalTheoryCommunication,holevo1973BoundsQuantityInformation,holevo1998CapacityQuantumChannel,schumacher1997SendingClassicalInformation,shor1995SchemeReducingDecoherence,schumacher1996SendingEntanglementNoisy,schumacher1996QuantumDataProcessing,lloyd1997CapacityNoisyQuantum,barnum1998InformationTransmissionNoisy,barnum2000QuantumFidelitiesChannel,shor2002QuantumChannelCapacity,devetak2005PrivateClassicalCapacity} corresponds to when $\Theta_{MB\to A\g{M}}$ is realized by just local operations.
	\item In the setting of retrocausal communication (or equivalently, communication assisted by a noiseless P-CTC; see Remark~\ref{rem:assisted} and Fig.~\ref{fig:assisted}) that we consider, the nonsignaling constraint from $B$ to $A$ on $\Theta_{MB\to A\g{M}}$ is released, which enables the possible formation of a causal loop and renders the strategy indescribable by a superchannel; however, the nonsignaling constraint from $M$ to $\g{M}$ must be kept, so that the strategy, despite its inclusiveness, does not introduce any communication resource from the sender to the receiver.  Note that the sole requirement of $\Theta_{MB\to A\g{M}}$ being nonsignaling from $M$ to $\g{M}$ is equivalent to our assumption that the strategy can be decomposed into a pair of channels $(\ch{E}_{ML\to A},\ch{D}_{B\to\g{M}L})$ with $L$ a quantum memory, due to Ref.~\cite{eggeling2002SemicausalOperationsAre}.
\end{itemize}
\end{remark}

\subsection{Retrocausal quantum communication}
\label{sec:communication-quantum}

We first consider the case where the messages being transmitted are quantum.  A quantum message can be represented by a pure state $\psi_{RM}$, where $R$ is an inaccessible reference system that purifies the message system $M$.  For a strategy $(\ch{E}_{ML\to A},\ch{D}_{B\to\g{M}L})$, the fidelity of the daughter's recovery of the message is given by the overlap between $\psi_{R\g{M}}$ and the actual state of $R\g{M}$ that she recovers, which by Eq.~\eqref{eq:retrocausal} equals
\begin{align}
	F\fleft(\psi_{R\g{M}}\middle|\psi_{RM};\fleft(\ch{E}_{ML\to A},\ch{D}_{B\to\g{M}L}\fright),\ch{N}_{A\to B}\fright)&=\frac{\tr\fleft[\psi_{R\g{M}}\Gamma_A\fleft\{\ch{E}_{ML\to A}\circ\ch{D}_{B\to\g{M}L}\circ\ch{N}_{A\to B}\fright\}\fleft[\psi_{RM}\fright]\fright]}{\tr\fleft[\Gamma_A\fleft\{\ch{E}_{ML\to A}\circ\ch{D}_{B\to\g{M}L}\circ\ch{N}_{A\to B}\fright\}\fleft[\psi_{RM}\fright]\fright]}. \label{eq:fidelity}
\end{align}
The \emph{worst-case quantum infidelity} associated with the strategy is then given by
\begin{align}
	q_\abb{infid}\fleft(\fleft(\ch{E}_{ML\to A},\ch{D}_{B\to\g{M}L}\fright),\ch{N}_{A\to B}\fright)&\coloneq1-\inf_{\substack{d_R\in\spa{N}, \\
	\psi_{RM}\in\s{D}_{RM}\colon \\ \rk\fleft(\psi_{RM}\fright)=1}}F\fleft(\psi_{R\g{M}}\middle|\psi_{RM};\fleft(\ch{E}_{ML\to A},\ch{D}_{B\to\g{M}L}\fright),\ch{N}_{A\to B}\fright), \label{eq:infidelity}
\end{align}
which upper bounds the infidelity for all possible quantum messages.

The \emph{one-shot retrocausal quantum capacity} of a completely positive map $\ch{N}_{A\to B}\in\s{L}_{A\to B}$ within infidelity $\varepsilon\in[0,1]$ is defined as
\begin{align}
	Q_\abb{retro}^\varepsilon\fleft(\ch{N}\fright)&\coloneq\sup_{\substack{d_M,d_L\in\spa{N}, \\
	\ch{E}_{ML\to A}\in\s{C}_{ML\to A}, \\ \ch{D}_{B\to\g{M}L}\in\s{C}_{B\to\g{M}L}}}\left\{\log_2d_M\colon q_\abb{infid}\fleft(\fleft(\ch{E}_{ML\to A},\ch{D}_{B\to\g{M}L}\fright),\ch{N}_{A\to B}\fright)\leq\varepsilon\right\}, \label{eq:quantum-oneshot}
\end{align}
which quantifies the maximum number of qubits that can be transmitted within infidelity $\varepsilon$ from the father to the daughter through the noisy P-CTC represented by $\ch{N}_{A\to B}$ under an encoding and decoding of their choice.

The \emph{asymptotic retrocausal quantum capacity} of a completely positive map $\ch{N}_{A\to B}\in\s{L}_{A\to B}$ is defined as
\begin{align}
	Q_\abb{retro}\fleft(\ch{N}\fright)&\coloneq\inf_{\varepsilon\in(0,1)}\liminf_{n\to\infty}\frac{1}{n}Q_\abb{retro}^\varepsilon\fleft(\ch{N}^{\otimes n}\fright), \label{eq:quantum-asymptotic}
\end{align}
which quantifies the maximum rate at which the number of reliably transmitted qubits increases with respect to the number of parallel uses of the noisy P-CTC.

The \emph{error exponent of retrocausal quantum communication} for a completely positive map $\ch{N}_{A\to B}\in\s{L}_{A\to B}$ at a transmission rate $r\in[0,\infty)$ is defined as
\begin{align}
	E_\abb{retro,Q}^r\fleft(\ch{N}\fright)&\coloneq\liminf_{n\to\infty}\sup_{\substack{d_M,d_L\in\spa{N}, \\
	\ch{E}_{ML\to A^n}\in\s{C}_{ML\to A^n}, \\ \ch{D}_{B^n\to\g{M}L}\in\s{C}_{B^n\to\g{M}L}}}\left\{-\frac{1}{n}\log_2q_\abb{infid}\fleft(\fleft(\ch{E}_{ML\to A^n},\ch{D}_{B^n\to\g{M}L}\fright),\fleft(\ch{N}^{\otimes n}\fright)_{A^n\to B^n}\fright)\colon d_M\geq 2^{nr}\right\}, \label{eq:quantum-error}
\end{align}
which quantifies the maximum rate at which the infidelity exponentially decays as the number of transmitted qubits increases at rate $r$.

The \emph{strong-converse exponent of retrocausal quantum communication} for a completely positive map $\ch{N}_{A\to B}\in\s{L}_{A\to B}$ at a transmission rate $r\in[0,\infty)$ is defined as
\begin{align}
	&S_\abb{retro,Q}^r\fleft(\ch{N}\fright) \notag\\
	&\coloneq\limsup_{n\to\infty}\inf_{\substack{d_M,d_L\in\spa{N}, \\
	\ch{E}_{ML\to A^n}\in\s{C}_{ML\to A^n} \\ \ch{D}_{B^n\to\g{M}L}\in\s{C}_{B^n\to\g{M}L}}}\left\{-\frac{1}{n}\log_2\left(1-q_\abb{infid}\fleft(\fleft(\ch{E}_{ML\to A^n},\ch{D}_{B^n\to\g{M}L}\fright),\fleft(\ch{N}^{\otimes n}\fright)_{A^n\to B^n}\fright)\right)\colon d_M\geq 2^{nr}\right\}, \label{eq:quantum-strong}
\end{align}
which quantifies the minimum rate at which the infidelity exponentially approaches $1$ as the number of transmitted qubits increases at rate $r$.

\subsection{Retrocausal classical communication}
\label{sec:communication-classical}

We now consider the case where the messages being transmitted are classical.  A classical message $m$ from an alphabet of size $d_M$ is represented by the classical state $\op{m}{m}_M$, with $\{\ket{m}\}_{m=0}^{d_M-1}$ an orthonormal basis of $\spa{H}_M$.  For a strategy $(\ch{E}_{ML\to A},\ch{D}_{B\to\g{M}L})$, the probability of the daughter correctly recovering the message is given by the overlap between $\op{m}{m}_\g{M}$ and the actual state of $\g{M}$ that she recovers, which by Eq.~\eqref{eq:retrocausal} equals
\begin{align}
	\Pr\fleft\{m\middle|m;\fleft(\ch{E}_{ML\to A},\ch{D}_{B\to\g{M}L}\fright),\ch{N}_{A\to B}\fright\}&=\frac{\tr\fleft[\op{m}{m}_\g{M}\Gamma_A\fleft\{\ch{E}_{ML\to A}\circ\ch{D}_{B\to\g{M}L}\circ\ch{N}_{A\to B}\fright\}\fleft[\op{m}{m}_M\fright]\fright]}{\tr\fleft[\Gamma_A\fleft\{\ch{E}_{ML\to A}\circ\ch{D}_{B\to\g{M}L}\circ\ch{N}_{A\to B}\fright\}\fleft[\op{m}{m}_M\fright]\fright]}. \label{eq:probability}
\end{align}
The \emph{worst-case classical error probability} associated with the strategy is then given by
\begin{align}
	p_\abb{error}\fleft(\fleft(\ch{E}_{ML\to A},\ch{D}_{B\to\g{M}L}\fright),\ch{N}_{A\to B}\fright)&\coloneq1-\min_{m}\Pr\fleft\{m\middle|m;\fleft(\ch{E}_{ML\to A},\ch{D}_{B\to\g{M}L}\fright),\ch{N}_{A\to B}\fright\}, \label{eq:error}
\end{align}
where the minimization is over all symbols in the alphabet.  It upper bounds the error probability for all possible classical messages.

The \emph{one-shot retrocausal classical capacity} of a completely positive map $\ch{N}_{A\to B}\in\s{L}_{A\to B}$ within error $\varepsilon\in[0,1]$ is defined as
\begin{align}
	C_\abb{retro}^\varepsilon\fleft(\ch{N}\fright)&\coloneq\sup_{\substack{d_M,d_L\in\spa{N}, \\
	\ch{E}_{ML\to A}\in\s{C}_{ML\to A}, \\ \ch{D}_{B\to\g{M}L}\in\s{C}_{B\to\g{M}L}}}\left\{\log_2d_M\colon p_\abb{error}\fleft(\fleft(\ch{E}_{ML\to A},\ch{D}_{B\to\g{M}L}\fright),\ch{N}_{A\to B}\fright)\leq\varepsilon\right\}, \label{eq:classical-oneshot}
\end{align}
which quantifies the maximum number of qubits that can be transmitted within error $\varepsilon$ from the father to the daughter through the noisy P-CTC represented by $\ch{N}_{A\to B}$ under an encoding and decoding of their choice.

The \emph{asymptotic retrocausal classical capacity} of a completely positive map $\ch{N}_{A\to B}\in\s{L}_{A\to B}$ is defined as
\begin{align}
	C_\abb{retro}\fleft(\ch{N}\fright)&\coloneq\inf_{\varepsilon\in(0,1)}\liminf_{n\to\infty}\frac{1}{n}C_\abb{retro}^\varepsilon\fleft(\ch{N}^{\otimes n}\fright), \label{eq:classical-asymptotic}
\end{align}
which quantifies the maximum rate at which the number of reliably transmitted qubits increases with respect to the number of parallel uses of the noisy P-CTC.
It quantifies the maximum rate at which the number of reliably transmitted classical bits increases with respect to the number of parallel uses of the noisy P-CTC.

The \emph{error exponent of retrocausal classical communication} for a completely positive map $\ch{N}_{A\to B}\in\s{L}_{A\to B}$ at a transmission rate $r\in[0,\infty)$ is defined as
\begin{align}
	E_\abb{retro,C}^r\fleft(\ch{N}\fright)&\coloneq\liminf_{n\to\infty}\sup_{\substack{d_M,d_L\in\spa{N}, \\
	\ch{E}_{ML\to A^n}\in\s{C}_{ML\to A^n}, \\ \ch{D}_{B^n\to\g{M}L}\in\s{C}_{B^n\to\g{M}L}}}\left\{-\frac{1}{n}\log_2p_\abb{error}\fleft(\fleft(\ch{E}_{ML\to A^n},\ch{D}_{B^n\to\g{M}L}\fright),\fleft(\ch{N}^{\otimes n}\fright)_{A^n\to B^n}\fright)\colon d_M\geq 2^{nr}\right\}, \label{eq:classical-error}
\end{align}
which quantifies the maximum rate at which the error probability exponentially decays as the number of transmitted classical bits increases at rate $r$.

The \emph{strong-converse exponent of retrocausal classical communication} for a completely positive map $\ch{N}_{A\to B}\in\s{L}_{A\to B}$ at a transmission rate $r\in[0,\infty)$ is defined as
\begin{align}
	&S_\abb{retro,C}^r\fleft(\ch{N}\fright) \notag\\
	&\coloneq\limsup_{n\to\infty}\inf_{\substack{d_M,d_L\in\spa{N}, \\
	\ch{E}_{ML\to A^n}\in\s{C}_{ML\to A^n}, \\ \ch{D}_{B^n\to\g{M}L}\in\s{C}_{B^n\to\g{M}L}}}\left\{-\frac{1}{n}\log_2\left(1-p_\abb{error}\fleft(\fleft(\ch{E}_{ML\to A^n},\ch{D}_{B^n\to\g{M}L}\fright),\fleft(\ch{N}^{\otimes n}\fright)_{A^n\to B^n}\fright)\right)\colon d_M\geq 2^{nr}\right\}, \label{eq:classical-strong}
\end{align}
which quantifies the minimum rate at which the error probability exponentially approaches $1$ as the number of transmitted classical bits increases at rate $r$.

\section{Retrocausal quantum capacity}
\label{sec:quantum}

In this section, we precisely determine the one-shot retrocausal quantum capacity.

\begin{theorem}[One-shot retrocausal quantum capacity]
\label{thm:quantum-oneshot}
The one-shot retrocausal quantum capacity has a closed-form expression in terms of the max- and Doeblin informations: for a completely positive map $\ch{N}_{A\to B}\in\s{L}_{A\to B}$ and a real number $\varepsilon\in(0,1)$,
\begin{align}
	Q_\abb{retro}^\varepsilon\fleft(\ch{N}\fright)&=\log_2\left\lfloor\sqrt{\frac{\varepsilon}{1-\varepsilon}2^{I_{\max}\fleft(\ch{N}\fright)+I_\abb{doe}\fleft(\ch{N}\fright)}+1}\right\rfloor,
\end{align}
where the max- and Doeblin informations are defined in Eqs.~\eqref{eq:max} and \eqref{eq:doeblin}, respectively.
\end{theorem}

\begin{remark}[Efficient computability]
\label{rem:computable}
Since the max- and Doeblin informations are both defined via semidefinite programs (SDPs) [see Eqs.~\eqref{eq:max-SDP} and \eqref{eq:doeblin-SDP}], the one-shot retrocausal quantum capacity is efficiently computable as a closed-form formula thereof.
\end{remark}

In what follows, we prove Theorem~\ref{thm:quantum-oneshot} by first establishing the desired converse bound on the one-shot retrocausal quantum capacity and then showing that it is achievable by construction.  The basic idea of the proof is similar to that of the results in Ref.~\cite{ji2024PostselectedCommunicationQuantum}.

\subsection{Converse bound}
\label{sec:quantum-converse}

As shown in the following lemma, if the father and the daughter were connected by a noisy P-CTC represented by a replacement channel, then the resulting transformation from $M$ to $\g{M}$ is necessarily also a replacement channel and thus useless for communication, regardless of the strategy they adopt.  A similar statement can be made for affine combinations of replacement channels as well.

\begin{lemma}[Useless noisy P-CTC]
\label{lem:useless}
For a strategy $(\ch{E}_{ML\to A},\ch{D}_{B\to\g{M}L})$ and a unit-trace Hermitian operator $\tau_B\in\aff(\s{D}_B)$, there exists another unit-trace Hermitian operator $\eta_\g{M}\in\aff(\s{D}_\g{M})$ such that
\begin{align}
	\Gamma_A\fleft\{\ch{E}_{ML\to A}\circ\ch{D}_{B\to\g{M}L}\circ\ch{R}_{A\to B}^\tau\fright\}&=\frac{1}{d_A^2}\ch{R}_{M\to\g{M}}^\eta.
\end{align}
Furthermore, if $\tau_B$ is a state, then so is $\eta_\g{M}$.
\end{lemma}

\begin{proof}
Define a unit-trace Hermitian operator $\eta_\g{M}\coloneq\ch{D}_{B\to\g{M}}[\tau_B]$.  Note that $\eta_\g{M}$ is a state if $\tau_B$ is.  Then it follows from the definition of the loop supermap [Eq.~\eqref{eq:loop}] that, for every linear operator $\rho_{RM}\in\spa{B}_{RM}$,
\begin{align}
	\Gamma_A\fleft\{\ch{E}_{ML\to A}\circ\ch{D}_{B\to\g{M}L}\circ\ch{R}_{A\to B}^\tau\fright\}\fleft[\rho_{RM}\fright]&=\tr_{AA'}\fleft[\Phi_{AA'}\left(\ch{E}_{ML\to A}\circ\ch{D}_{B\to\g{M}L}\circ\ch{R}_{A\to B}^\tau\right)\fleft[\rho_{RM}\otimes\Phi_{AA'}\fright]\fright] \\
	&=\tr_{AA'}\fleft[\Phi_{AA'}\left(\ch{E}_{ML\to A}\circ\ch{D}_{B\to\g{M}L}\right)\fleft[\rho_{RM}\otimes\pi_{A'}\otimes\tau_B\fright]\fright] \\
	&=\frac{1}{d_A}\tr_A\fleft[\pi_A\left(\ch{E}_{ML\to A}\circ\ch{D}_{B\to\g{M}L}\right)\fleft[\rho_{RM}\otimes\tau_B\fright]\fright] \\
	&=\frac{1}{d_A^2}\tr_{ML}\fleft[\ch{D}_{B\to\g{M}L}\fleft[\rho_{RM}\otimes\tau_B\fright]\fright] \\
	&=\frac{1}{d_A^2}\rho_R\otimes\ch{D}_{B\to\g{M}}\fleft[\tau_B\fright] \\
	&=\frac{1}{d_A^2}\ch{R}_{M\to\g{M}}^\eta\fleft[\rho_{RM}\fright],
\end{align}
which implies the desired statement.
\end{proof}

We now prove the desired converse bound on the one-shot retrocausal quantum capacity.

\begin{proposition}[Converse bound on retrocausal quantum capacity]
\label{prop:quantum-converse}
For a completely positive map $\ch{N}_{A\to B}\in\s{L}_{A\to B}$ and a real number $\varepsilon\in(0,1)$,
\begin{align}
	Q_\abb{retro}^\varepsilon\fleft(\ch{N}\fright)&\leq\log_2\left\lfloor\sqrt{\frac{\varepsilon}{1-\varepsilon}2^{I_{\max}\fleft(\ch{N}\fright)+I_\abb{doe}\fleft(\ch{N}\fright)}+1}\right\rfloor.
\end{align}
\end{proposition}

\begin{proof}
Let $(\ch{E}_{ML\to A},\ch{D}_{B\to\g{M}L})$ be a strategy such that
\begin{align}
	q_\abb{infid}\fleft(\fleft(\ch{E}_{ML\to A},\ch{D}_{B\to\g{M}L}\fright),\ch{N}_{A\to B}\fright)&\leq\varepsilon. \label{pf:quantum-converse-1}
\end{align}
Let $\lambda,\mu\in\spa{R}$ be two real numbers, $\sigma_B\in\s{D}_B$ a state, and $\tau_B\in\aff(\s{D}_B)$ a unit-trace Hermitian operator, such that
\begin{align}
	\Phi_{A'B}^\ch{N}&\leq\lambda\Phi_{A'B}^{\ch{R}^\sigma}, \label{pf:quantum-converse-2}\\
	\Phi_{A'B}^{\ch{R}^\tau}&\leq\mu\Phi_{A'B}^\ch{N}. \label{pf:quantum-converse-3}
\end{align}
Equation~\eqref{pf:quantum-converse-2} means that $\lambda\ch{R}_{A\to B}^\sigma-\ch{N}_{A\to B}$ is a completely positive map, and therefore so is $\lambda\Gamma_A\fleft\{\ch{E}_{ML\to A}\circ\ch{D}_{B\to\g{M}L}\circ\ch{R}_{A\to B}^\sigma\fright\}-\Gamma_A\fleft\{\ch{E}_{ML\to A}\circ\ch{D}_{B\to\g{M}L}\circ\ch{N}_{A\to B}\fright\}$.  This implies that
\begin{align}
	\Gamma_A\fleft\{\ch{E}_{ML\to A}\circ\ch{D}_{B\to\g{M}L}\circ\ch{N}_{A\to B}\fright\}\fleft[\Phi_{M'M}\fright]&\leq\lambda\Gamma_A\fleft\{\ch{E}_{ML\to A}\circ\ch{D}_{B\to\g{M}L}\circ\ch{R}_{A\to B}^\sigma\fright\}\fleft[\Phi_{M'M}\fright] \\
	&=\frac{\lambda}{d_A^2}\ch{R}_{M\to\g{M}}^\omega\fleft[\Phi_{M'M}\fright] \\
	&=\frac{\lambda}{d_A^2}\pi_{M'}\otimes\omega_\g{M}, \label{pf:quantum-converse-4}
\end{align}
where $\omega_\g{M}\in\s{D}_\g{M}$ is a state, whose existence is guaranteed by Lemma~\ref{lem:useless}.  Likewise, it follows from Eq.~\eqref{pf:quantum-converse-3} that
\begin{align}
	\mu\Gamma_A\fleft\{\ch{E}_{ML\to A}\circ\ch{D}_{B\to\g{M}L}\circ\ch{N}_{A\to B}\fright\}\fleft[\Phi_{M'M}\fright]&\geq\Gamma_A\fleft\{\ch{E}_{ML\to A}\circ\ch{D}_{B\to\g{M}L}\circ\ch{R}_{A\to B}^\tau\fright\}\fleft[\Phi_{M'M}\fright] \\
	&=\frac{1}{d_A^2}\ch{R}_{M\to\g{M}}^\eta\fleft[\Phi_{M'M}\fright] \\
	&=\frac{1}{d_A^2}\pi_{M'}\otimes\eta_\g{M}, \label{pf:quantum-converse-5}
\end{align}
where $\eta_\g{M}\in\aff(\s{D}_\g{M})$ is a unit-trace Hermitian operator, whose existence is guaranteed by Lemma~\ref{lem:useless}.  Then it follows from Eq.~\eqref{pf:quantum-converse-1} that
\begin{align}
	1-\varepsilon&\leq1-q_\abb{infid}\fleft(\fleft(\ch{E}_{ML\to A},\ch{D}_{B\to\g{M}L}\fright),\ch{N}_{A\to B}\fright) \label{pf:quantum-converse-6}\\
	&\leq\frac{\tr\fleft[\Phi_{M'\g{M}}\Gamma_A\fleft\{\ch{E}_{ML\to A}\circ\ch{D}_{B\to\g{M}L}\circ\ch{N}_{A\to B}\fright\}\fleft[\Phi_{M'M}\fright]\fright]}{\tr\fleft[\Gamma_A\fleft\{\ch{E}_{ML\to A}\circ\ch{D}_{B\to\g{M}L}\circ\ch{N}_{A\to B}\fright\}\fleft[\Phi_{M'M}\fright]\fright]} \label{pf:quantum-converse-7}\\
	&\leq\frac{\lambda\tr\fleft[\Phi_{M'\g{M}}\left(\pi_{M'}\otimes\omega_\g{M}\right)\fright]}{d_A^2\tr\fleft[\Gamma_A\fleft\{\ch{E}_{ML\to A}\circ\ch{D}_{B\to\g{M}L}\circ\ch{N}_{A\to B}\fright\}\fleft[\Phi_{M'M}\fright]\fright]} \label{pf:quantum-converse-8}\\
	&=\frac{\lambda}{d_M^2d_A^2\tr\fleft[\Gamma_A\fleft\{\ch{E}_{ML\to A}\circ\ch{D}_{B\to\g{M}L}\circ\ch{N}_{A\to B}\fright\}\fleft[\Phi_{M'M}\fright]\fright]}, \label{pf:quantum-converse-9}
\end{align}
where Eq.~\eqref{pf:quantum-converse-7} follows from Eqs.~\eqref{eq:infidelity} and \eqref{eq:fidelity} and choosing $R=M'$ and $\psi_{RM}=\Phi_{M'M}$; Eq.~\eqref{pf:quantum-converse-8} follows from Eq.~\eqref{pf:quantum-converse-4}.  Similarly,
\begin{align}
	\mu\varepsilon&\geq\mu q_\abb{infid}\fleft(\fleft(\ch{E}_{ML\to A},\ch{D}_{B\to\g{M}L}\fright),\ch{N}_{A\to B}\fright) \label{pf:quantum-converse-10}\\
	&\geq\frac{\mu\tr\fleft[\left(\1_{M'\g{M}}-\Phi_{M'\g{M}}\right)\Gamma_A\fleft\{\ch{E}_{ML\to A}\circ\ch{D}_{B\to\g{M}L}\circ\ch{N}_{A\to B}\fright\}\fleft[\Phi_{M'M}\fright]\fright]}{\tr\fleft[\Gamma_A\fleft\{\ch{E}_{ML\to A}\circ\ch{D}_{B\to\g{M}L}\circ\ch{N}_{A\to B}\fright\}\fleft[\Phi_{M'M}\fright]\fright]} \\
	&\geq\frac{\tr\fleft[\left(\1_{M'\g{M}}-\Phi_{M'\g{M}}\right)\left(\pi_{M'}\otimes\eta_\g{M}\right)\fright]}{d_A^2\tr\fleft[\Gamma_A\fleft\{\ch{E}_{ML\to A}\circ\ch{D}_{B\to\g{M}L}\circ\ch{N}_{A\to B}\fright\}\fleft[\Phi_{M'M}\fright]\fright]} \label{pf:quantum-converse-11}\\
	&=\left(1-\frac{1}{d_M^2}\right)\frac{1}{d_A^2\tr\fleft[\Gamma_A\fleft\{\ch{E}_{ML\to A}\circ\ch{D}_{B\to\g{M}L}\circ\ch{N}_{A\to B}\fright\}\fleft[\Phi_{M'M}\fright]\fright]}, \label{pf:quantum-converse-12}
\end{align}
where Eq.~\eqref{pf:quantum-converse-11} follows from Eq.~\eqref{pf:quantum-converse-5}.  Note that combining Eqs.~\eqref{pf:quantum-converse-9} and \eqref{pf:quantum-converse-12} implies that
\begin{align}
	d_M^2\leq\frac{\varepsilon\lambda\mu}{1-\varepsilon}+1,
\end{align}
which, due to $d_M$ being an integer, further implies that
\begin{align}
	\log_2d_M\leq\log_2\left\lfloor\sqrt{\frac{\varepsilon\lambda\mu}{1-\varepsilon}+1}\right\rfloor.
\end{align}
Supremizing the left-hand side over every $d_M,d_L\in\spa{N}$ and every strategy $(\ch{E}_{ML\to A},\ch{D}_{B\to\g{M}L})$ satisfying Eq.~\eqref{pf:quantum-converse-1}, and infimizing the right-hand side over every $\lambda,\mu\in\spa{R}$, $\sigma_B\in\s{D}_B$, and $\tau_B\in\aff(\s{D}_B)$ satisfying Eqs.~\eqref{pf:quantum-converse-2} and \eqref{pf:quantum-converse-3}, the desired statement thus follows from the definitions of the one-shot retrocausal quantum capacity and the max- and Doeblin informations [Eqs.~\eqref{eq:quantum-oneshot}, \eqref{eq:max}, and \eqref{eq:doeblin}].
\end{proof}

\subsection{Achievability}
\label{sec:quantum-achievability}

We now explicitly construct a strategy that saturates the converse bound derived in Proposition~\ref{prop:quantum-converse}.

\begin{strategy}[Amplified probabilistic quantum teleportation]
\label{str:quantum}
The strategy operates as follows (see Fig.~\ref{fig:quantum-achievability}).
\begin{itemize}
	\item \textbf{Encoding.}  In the future, the father retrieves two quantum memories $L_1$ and $L_2$ (which were left by Bella in the past) with $L_1\cong M$ and $L_2\cong B$.  He performs a measurement on $M$ and $L_1$ according to the POVM $(\Phi_{ML_1},\1_{ML_1}-\Phi_{ML_1})$, obtaining one of the two possible outcomes, $k\in\{0,1\}$.  Conditioned on $k$, he applies a channel $\ch{K}_{L_2\to A}^{(k)}\in\s{C}_{L_2\to A}$ to $L_2$.  He then feeds $A$ into the entrance of the noisy P-CTC.
	\item \textbf{Decoding.}  In the past, the daughter receives $B$ from the exit of the noisy P-CTC.  She prepares a standard maximally entangled state $\Phi_{\g{M}L_1}$ between $\g{M}$ and $L_1$, renames $B$ as $L_2$, and leaves both $L_1$ and $L_2$ for the father to retrieve in the future.  She then declares that $\g{M}$ stores her recovery of the message.
\end{itemize}
\end{strategy}

\begin{figure}[t]
\includegraphics[scale=0.3]{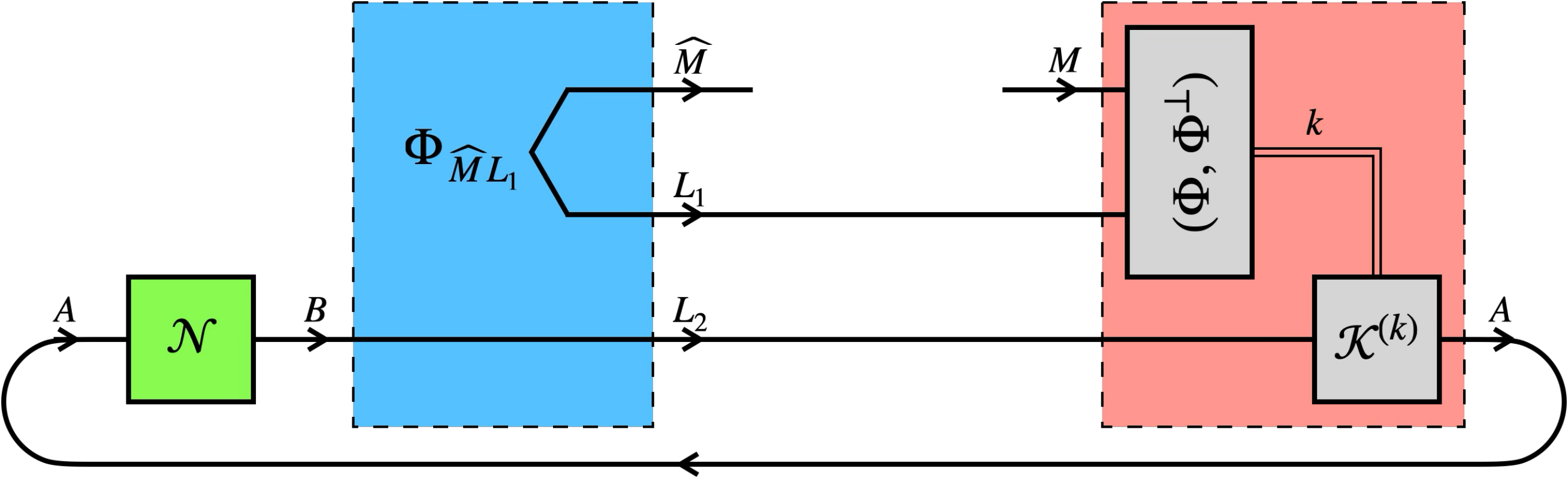}
\caption{The basic idea of amplified probabilistic quantum teleportation is to first perform probabilistic teleportation from the father to the daughter, and then engineer a causal loop involving the noisy P-CTC such that it maximally amplifies the renormalized probability of the teleportation being successful.  The red and blue boxes encapsulate the father's encoding and the daughter's decoding, respectively.}
\label{fig:quantum-achievability}
\end{figure}

The intuition behind Strategy~\ref{str:quantum} is as follows.  Instead of trying to code his message into $A$ which enters the noisy P-CTC, the father intends to ``teleport''~\cite{bennett1993TeleportingUnknownQuantum} the message probabilistically to the daughter via $L_1$.  Indeed, by probababilistic teleportation, whenever the father obtains the outcome $k=0$ after performing the measurement on $ML_1$, what the daughter ends up with in $\g{M}$ is a perfect recovery of the message in $M$.  However, the outcome $k=0$ is not guaranteed to occur; in fact, it occurs only with probability $1/d_M^2$, or at least so without invoking the noisy P-CTC.  Yet, as illustrated in Fig.~\ref{fig:quantum-achievability}, the remaining part of the strategy engineers a causal loop solely consisting of the noisy P-CTC, represented by $\ch{N}_{A\to B}$, and the channel $\ch{K}_{L_2\to A}^{(k)}$, conditioned on $k\in\{0,1\}$.  Due to the nonlinearity of such a loop (specifically, its capability of renormalizing the probabilities of certain events), the father can choose the two channels $\ch{K}_{L_2\to A}^{(0)}$ and $\ch{K}_{L_2\to A}^{(1)}$ in such a way that the renormalized probability of the event $k=0$ is maximally amplified (or equivalently, the renormalized probability of $k=1$ is maximally diminished).  If the renormalized probability of $k=0$ can be amplified to at least $1-\varepsilon$ for a message system $M$ of dimension $d_M$, then $\log_2d_M$ qubits can indeed be ``teleported'' from $M$ to $\g{M}$ via $L_1$ subject to a failure probability at most $\varepsilon$, which upper bounds the overall infidelity of communication.

We now formalize this analysis and show that Strategy~\ref{str:quantum} can indeed achieve the maximum possible efficiency allowed by the converse bound established in Proposition~\ref{prop:quantum-converse}.

\begin{proposition}[Achievability of retrocausal quantum capacity]
\label{prop:quantum-achievability}
For a completely positive map $\ch{N}_{A\to B}\in\s{L}_{A\to B}$ and a real number $\varepsilon\in(0,1)$,
\begin{align}
	Q_\abb{retro}^\varepsilon\fleft(\ch{N}\fright)&\geq\log_2\left\lfloor\sqrt{\frac{\varepsilon}{1-\varepsilon}2^{I_{\max}\fleft(\ch{N}\fright)+I_\abb{doe}\fleft(\ch{N}\fright)}+1}\right\rfloor. \label{eq:quantum-achievability}
\end{align}
Furthermore, the efficiency on the right-hand side of Eq.~\eqref{eq:quantum-achievability} is achievable by Strategy~\ref{str:quantum} for some channels $\ch{K}_{L_2\to A}^{(0)}$ and $\ch{K}_{L_2\to A}^{(1)}$.
\end{proposition}

\begin{proof}
The encoding and decoding of Strategy~\ref{str:quantum} are represented by the following channels, respectively:
\begin{align}
	\ch{E}_{ML_1L_2\to A}\fleft[\cdot\fright]&\coloneq\tr_{ML_1}\fleft[\Phi_{ML_1}\ch{K}_{L_2\to A}^{(0)}\fleft[\cdot\fright]\fright]+\tr_{ML_1}\fleft[\left(\1_{ML_1}-\Phi_{ML_1}\right)\ch{K}_{L_2\to A}^{(1)}\fleft[\cdot\fright]\fright], \label{pf:quantum-achievability-1}\\
	\ch{D}_{B\to\g{M}L_1L_2}\fleft[\cdot\fright]&\coloneq\Phi_{\g{M}L_1}\otimes\id_{B\to L_2}\fleft[\cdot\fright]. \label{pf:quantum-achievability-2}
\end{align}
Let $\psi_{RM}$ be a pure state.  Then it follows from the definition of the loop supermap [Eq.~\eqref{eq:loop}] that
\begin{align}
	&\Gamma_A\fleft\{\ch{E}_{ML_1L_2\to A}\circ\ch{D}_{B\to\g{M}L_1L_2}\circ\ch{N}_{A\to B}\fright\}\fleft[\psi_{RM}\fright] \notag\\
	&=\tr_{AA'}\fleft[\Phi_{AA'}\left(\ch{E}_{ML_1L_2\to A}\circ\ch{D}_{B\to\g{M}L_1L_2}\circ\ch{N}_{A\to B}\right)\fleft[\psi_{RM}\otimes\Phi_{AA'}\fright]\fright] \\
	&=\tr_{ML_1AA'}\fleft[\left(\Phi_{ML_1}\otimes\Phi_{AA'}\right)\left(\ch{K}_{B\to A}^{(0)}\circ\ch{N}_{A\to B}\right)\fleft[\psi_{RM}\otimes\Phi_{\g{M}L_1}\otimes\Phi_{AA'}\fright]\fright] \notag\\
	&\quad+\tr_{ML_1AA'}\fleft[\left(\left(\1_{ML_1}-\Phi_{ML_1}\right)\otimes\Phi_{AA'}\right)\left(\ch{K}_{B\to A}^{(1)}\circ\ch{N}_{A\to B}\right)\fleft[\psi_{RM}\otimes\Phi_{\g{M}L_1}\otimes\Phi_{AA'}\fright]\fright] \\
	&=\frac{1}{d_M^2}\tr\fleft[\Phi_{AA'}\left(\ch{K}_{B\to A}^{(0)}\circ\ch{N}_{A\to B}\right)\fleft[\Phi_{AA'}\fright]\fright]\psi_{R\g{M}} \notag\\
	&\quad+\frac{1}{d_M^2}\tr\fleft[\Phi_{AA'}\left(\ch{K}_{B\to A}^{(1)}\circ\ch{N}_{A\to B}\right)\fleft[\Phi_{AA'}\fright]\fright]\left(d_M^2\tr_M\fleft[\psi_{RM}\fright]\otimes\pi_\g{M}-\psi_{R\g{M}}\right). \label{pf:quantum-achievability-3}
\end{align}
Considering the fact that $d_M^2\tr_M[\psi_{RM}]\otimes\pi_\g{M}-\psi_{R\g{M}}\geq0$, it follows from Eq.~\eqref{eq:fidelity} that
\begin{align}
	&1-F\fleft(\psi_{R\g{M}}\middle|\psi_{RM};\fleft(\ch{E}_{ML_1L_2\to A},\ch{D}_{B\to\g{M}L_1L_2}\fright),\ch{N}_{A\to B}\fright) \notag\\
	&=1-\frac{\tr\fleft[\psi_{R\g{M}}\Gamma_A\fleft\{\ch{E}_{ML_1L_2\to A}\circ\ch{D}_{B\to\g{M}L_1L_2}\circ\ch{N}_{A\to B}\fright\}\fleft[\psi_{RM}\fright]\fright]}{\tr\fleft[\Gamma_A\fleft\{\ch{E}_{ML_1L_2\to A}\circ\ch{D}_{B\to\g{M}L_1L_2}\circ\ch{N}_{A\to B}\fright\}\fleft[\psi_{RM}\fright]\fright]} \\
	&\leq1-\frac{\frac{1}{d_M^2}\tr\fleft[\Phi_{AA'}\left(\ch{K}_{B\to A}^{(0)}\circ\ch{N}_{A\to B}\right)\fleft[\Phi_{AA'}\fright]\fright]}{\frac{1}{d_M^2}\tr\fleft[\Phi_{AA'}\left(\ch{K}_{B\to A}^{(0)}\circ\ch{N}_{A\to B}\right)\fleft[\Phi_{AA'}\fright]\fright]+\frac{d_M^2-1}{d_M^2}\tr\fleft[\Phi_{AA'}\left(\ch{K}_{B\to A}^{(1)}\circ\ch{N}_{A\to B}\right)\fleft[\Phi_{AA'}\fright]\fright]} \label{pf:quantum-achievability-4}\\
	&=\left(\frac{\tr\fleft[\Phi_{AA'}\left(\ch{K}_{B\to A}^{(0)}\circ\ch{N}_{A\to B}\right)\fleft[\Phi_{AA'}\fright]\fright]}{\left(d_M^2-1\right)\tr\fleft[\Phi_{AA'}\left(\ch{K}_{B\to A}^{(1)}\circ\ch{N}_{A\to B}\right)\fleft[\Phi_{AA'}\fright]\fright]}+1\right)^{-1}, \label{pf:quantum-achievability-5}
\end{align}
where Eq.~\eqref{pf:quantum-achievability-4} follows from Eq.~\eqref{pf:quantum-achievability-3}.  Since this holds for every pure state $\psi_{RM}$, it follows from Eq.~\eqref{eq:infidelity} that, as long as
\begin{align}
	\left(\frac{\tr\fleft[\Phi_{AA'}\left(\ch{K}_{B\to A}^{(0)}\circ\ch{N}_{A\to B}\right)\fleft[\Phi_{AA'}\fright]\fright]}{\left(d_M^2-1\right)\tr\fleft[\Phi_{AA'}\left(\ch{K}_{B\to A}^{(1)}\circ\ch{N}_{A\to B}\right)\fleft[\Phi_{AA'}\fright]\fright]}+1\right)^{-1}&\leq\varepsilon, \label{pf:quantum-achievability-6}
\end{align}
we have that
\begin{align}
	q_\abb{infid}\fleft(\fleft(\ch{E}_{ML_1L_2\to A},\ch{D}_{B\to\g{M}L_1L_2}\fright),\ch{N}_{A\to B}\fright)&=1-\inf_{\substack{\psi_{RM}\in\s{D}_{RM}\colon \\ \rk\fleft(\psi_{RM}\fright)=1}}F\fleft(\psi_{R\g{M}}\middle|\psi_{RM};\fleft(\ch{E}_{ML_1L_2\to A},\ch{D}_{B\to\g{M}L_1L_2}\fright),\ch{N}_{A\to B}\fright) \\
	&\leq\varepsilon. \label{pf:quantum-achievability-7}
\end{align}
Since Eq.~\eqref{pf:quantum-achievability-6} is equivalent to
\begin{align}
	d_M^2&\leq\frac{\varepsilon\tr\fleft[\Phi_{AA'}\left(\ch{K}_{B\to A}^{(0)}\circ\ch{N}_{A\to B}\right)\fleft[\Phi_{AA'}\fright]\fright]}{\left(1-\varepsilon\right)\tr\fleft[\Phi_{AA'}\left(\ch{K}_{B\to A}^{(1)}\circ\ch{N}_{A\to B}\right)\fleft[\Phi_{AA'}\fright]\fright]}+1, \label{pf:quantum-achievability-8}
\end{align}
by choosing
\begin{align}
	\ch{K}_{B\to A}^{(0)}&\coloneq\arg\max_{\ch{K}_{B\to A}\in\s{C}_{B\to A}}\tr\fleft[\Phi_{AA'}\left(\ch{K}_{B\to A}\circ\ch{N}_{A\to B}\right)\fleft[\Phi_{AA'}\fright]\fright], \label{pf:quantum-achievability-9}\\
	\ch{K}_{B\to A}^{(1)}&\coloneq\arg\min_{\ch{K}_{B\to A}\in\s{C}_{B\to A}}\tr\fleft[\Phi_{AA'}\left(\ch{K}_{B\to A}\circ\ch{N}_{A\to B}\right)\fleft[\Phi_{AA'}\fright]\fright], \label{pf:quantum-achievability-10}
\end{align}
which maximize the right-hand side of Eq.~\eqref{pf:quantum-achievability-8}, the following efficiency can be achieved within infidelity $\varepsilon$ (i.e., subject to Eq.~\eqref{pf:quantum-achievability-7}) and is thus a lower bound on the one-shot retrocausal quantum capacity [see Eq.~\eqref{eq:quantum-oneshot}]:
\begin{align}
	\log_2d_M&=\log_2\left\lfloor\sqrt{\frac{\varepsilon\tr\fleft[\Phi_{AA'}\left(\ch{K}_{B\to A}^{(0)}\circ\ch{N}_{A\to B}\right)\fleft[\Phi_{AA'}\fright]\fright]}{\left(1-\varepsilon\right)\tr\fleft[\Phi_{AA'}\left(\ch{K}_{B\to A}^{(1)}\circ\ch{N}_{A\to B}\right)\fleft[\Phi_{AA'}\fright]\fright]}+1}\right\rfloor \\
	&=\log_2\left\lfloor\sqrt{\frac{\varepsilon\max_{\ch{K}_{B\to A}\in\s{C}_{B\to A}}\tr\fleft[\Phi_{AA'}\left(\ch{K}_{B\to A}\circ\ch{N}_{A\to B}\right)\fleft[\Phi_{AA'}\fright]\fright]}{\left(1-\varepsilon\right)\min_{\ch{K}_{B\to A}\in\s{C}_{B\to A}}\tr\fleft[\Phi_{AA'}\left(\ch{K}_{B\to A}\circ\ch{N}_{A\to B}\right)\fleft[\Phi_{AA'}\fright]\fright]}+1}\right\rfloor \\
	&=\log_2\left\lfloor\sqrt{\frac{\varepsilon}{1-\varepsilon}2^{I_{\max}\fleft(\ch{N}\fright)+I_\abb{doe}\fleft(\ch{N}\fright)}+1}\right\rfloor, \label{pf:quantum-achievability-11}
\end{align}
where Eq.~\eqref{pf:quantum-achievability-11} follows from the interpretation of the max- and Doeblin informations in terms of the maximum and minimum achievable singlet fractions, respectively (Lemma~\ref{lem:fraction}), implying the desired statement.
\end{proof}

\begin{remark}[On infinite retrocausal quantum capacity]
\label{rem:infinite}
While the max-information of a completely positive map is always finite, the Doeblin information can be infinite.  As can be inferred from the proof of Proposition~\ref{prop:quantum-achievability}, for a completely positive map whose Doeblin information is equal to infinity, its one-shot retrocausal quantum capacity is equal to infinity for all $\varepsilon\in[0,1]$, including the zero-error case.
\end{remark}

Combining Propositions~\ref{prop:quantum-converse} and \ref{prop:quantum-achievability} completes the proof of Theorem~\ref{thm:quantum-oneshot}.

\subsection{Asymptotics}
\label{sec:quantum-asymptotics}

With the one-shot retrocausal quantum capacity determined, the asymptotic retrocausal quantum capacity immediately follows.

\begin{theorem}[Asymptotic retrocausal quantum capacity]
\label{thm:quantum-asymptotic}
The asymptotic retrocausal quantum capacity is equal to one half of the sum of the max- and regularized Doeblin informations: for a completely positive map $\ch{N}_{A\to B}\in\s{L}_{A\to B}$ and a real number $\varepsilon\in(0,1)$,
\begin{align}
	Q_\abb{retro}\fleft(\ch{N}\fright)&=\lim_{n\to\infty}\frac{1}{n}Q_\abb{retro}^\varepsilon\fleft(\ch{N}^{\otimes n}\fright)=\frac{1}{2}\left(I_{\max}\fleft(\ch{N}\fright)+I_\abb{doe}^\infty\fleft(\ch{N}\fright)\right),
\end{align}
where the regularized Doeblin information is defined in Eq.~\eqref{eq:doeblin-regularized}.
\end{theorem}

\begin{proof}
Since the max- and Doeblin informations are nonnegative, it follows from Theorem~\ref{thm:quantum-oneshot} that
\begin{align}
	Q_\abb{retro}^\varepsilon\fleft(\ch{N}\fright)&=\log_2\left\lfloor\sqrt{\frac{\varepsilon}{1-\varepsilon}2^{I_{\max}\fleft(\ch{N}\fright)+I_\abb{doe}\fleft(\ch{N}\fright)}+1}\right\rfloor \\
	&\leq\frac{1}{2}\log_2\left(\frac{\varepsilon}{1-\varepsilon}2^{I_{\max}\fleft(\ch{N}\fright)+I_\abb{doe}\fleft(\ch{N}\fright)}+1\right) \\
	&\leq\frac{1}{2}\log_2\left(\left(\frac{\varepsilon}{1-\varepsilon}+1\right)2^{I_{\max}\fleft(\ch{N}\fright)+I_\abb{doe}\fleft(\ch{N}\fright)}\right) \\
	&=\frac{1}{2}\left(I_{\max}\fleft(\ch{N}\fright)+I_\abb{doe}\fleft(\ch{N}\fright)+\log_2\frac{1}{1-\varepsilon}\right). \label{pf:quantum-asymptotic-1}
\end{align}
Due to the fact that $\lfloor\sqrt{x+1}\rfloor\geq\frac{1}{2}\sqrt{x}$ for every nonnegative real number $x\in[0,\infty)$, we also have that
\begin{align}
	Q_\abb{retro}^\varepsilon\fleft(\ch{N}\fright)&=\log_2\left\lfloor\sqrt{\frac{\varepsilon}{1-\varepsilon}2^{I_{\max}\fleft(\ch{N}\fright)+I_\abb{doe}\fleft(\ch{N}\fright)}+1}\right\rfloor \\
	&\geq\log_2\left(\frac{1}{2}\sqrt{\frac{\varepsilon}{1-\varepsilon}2^{I_{\max}\fleft(\ch{N}\fright)+I_\abb{doe}\fleft(\ch{N}\fright)}}\right) \\
	&=\frac{1}{2}\left(I_{\max}\fleft(\ch{N}\fright)+I_\abb{doe}\fleft(\ch{N}\fright)+\log_2\frac{\varepsilon}{4\left(1-\varepsilon\right)}\right). \label{pf:quantum-asymptotic-2}
\end{align}
Combining Eqs.~\eqref{pf:quantum-asymptotic-1} and \eqref{pf:quantum-asymptotic-2}, we have that, for every $n\in\spa{N}$,
\begin{align}
	\frac{1}{2n}\left(I_{\max}\fleft(\ch{N}^{\otimes n}\fright)+I_\abb{doe}\fleft(\ch{N}^{\otimes n}\fright)+\log_2\frac{1}{1-\varepsilon}\right)&\leq\frac{1}{n}Q_\abb{retro}^\varepsilon\fleft(\ch{N}^{\otimes n}\fright) \notag\\
	&\leq\frac{1}{2n}\left(I_{\max}\fleft(\ch{N}^{\otimes n}\fright)+I_\abb{doe}\fleft(\ch{N}^{\otimes n}\fright)+\log_2\frac{\varepsilon}{4\left(1-\varepsilon\right)}\right). \label{pf:quantum-asymptotic-3}
\end{align}
Taking the limit $n\to\infty$, it follows from the additivity of the max-information [Eq.~\eqref{eq:max-regularized}] and the definition of the regularized Doeblin information [Eq.~\eqref{eq:doeblin-regularized}] that
\begin{align}
	\lim_{n\to\infty}\frac{1}{n}Q_\abb{retro}^\varepsilon\fleft(\ch{N}^{\otimes n}\fright)&=\frac{1}{2}\left(I_{\max}\fleft(\ch{N}\fright)+I_\abb{doe}^\infty\fleft(\ch{N}\fright)\right),
\end{align}
from which and the definition of the asymptotic retrocausal quantum capacity [Eq.~\eqref{eq:quantum-asymptotic}] the desired statement follows.
\end{proof}

\begin{remark}[Efficiently computable bounds]
\label{rem:bounds}
For quantum-to-classical completely positive maps (which include measurement channels) and mictodiactic covariant completely positive maps (which include depolarizing channels), the Doeblin information is invariant under regularization~\cite[Theorem~3]{george2025QuantumDoeblinCoefficients}, and thus the asymptotic retrocausal quantum capacity has a single-letter expression and is efficiently computable via SDPs for such maps.  While it is unknown whether the asymptotic retrocausal quantum capacity has a single-letter expression for general completely positive maps, single-letter and SDP-computable bounds on it from both sides can be established.  Specifically, it follows from Theorem~\ref{thm:quantum-asymptotic} and Eq.~\eqref{eq:bounds} that
\begin{align}
	\frac{1}{2}\left(I_{\max}\fleft(\ch{N}\fright)+I_\abb{doe}\fleft(\ch{N}\fright)\right)&\leq Q_\abb{retro}\fleft(\ch{N}\fright)\leq\frac{1}{2}\left(I_{\max}\fleft(\ch{N}\fright)+I_\wang\fleft(\ch{N}\fright)\right),
\end{align}
where $I_\wang(\ch{N})$ is defined in Eq.~\eqref{eq:wang}.
\end{remark}

\begin{remark}[Second-order asymptotics]
\label{rem:second}
Equation~\eqref{pf:quantum-asymptotic-3} implies that, for a completely positive map whose regularized Doeblin information is equal to its Doeblin information (e.g., quantum-to-classical completely positive maps and mictodiactic covariant completely positive maps), the second-order asymptotics of retrocausal quantum communication becomes trivial.  However, it is unclear to us whether nontrivial second-order asymptotics exists for general completely positive maps.
\end{remark}

\subsection{Error exponent and strong-converse exponent}
\label{sec:quantum-exponent}

We now show that the error exponent and the strong-converse exponent of retrocausal quantum communication can also be determined based on the one-shot retrocausal quantum capacity.

\begin{theorem}[Error exponent and strong-converse exponent of retrocausal quantum communication]
\label{thm:quantum-exponent}
For a completely positive map $\ch{N}_{A\to B}\in\s{L}_{A\to B}$ and a real number $r\in[0,\infty)$,
\begin{align}
	E_\abb{retro,Q}^r\fleft(\ch{N}\fright)&=\max\left\{0,I_{\max}\fleft(\ch{N}\fright)+I_\abb{doe}^\infty\fleft(\ch{N}\fright)-2r\right\}, \\
	S_\abb{retro,Q}^r\fleft(\ch{N}\fright)&=\max\left\{0,2r-I_{\max}\fleft(\ch{N}\fright)-I_\abb{doe}^\infty\fleft(\ch{N}\fright)\right\}.
\end{align}
\end{theorem}

\begin{proof}
Let $(\ch{E}_{ML\to A},\ch{D}_{B\to\g{M}L})$ be a strategy.  Let $\lambda,\mu\in\spa{R}$ be two real numbers, $\sigma_B\in\s{D}_B$ a state, and $\tau_B\in\aff(\s{D}_B)$ a unit-trace Hermitian operator satisfying Eqs.~\eqref{pf:quantum-converse-2} and \eqref{pf:quantum-converse-3}.  Comparing Eqs.~\eqref{pf:quantum-converse-6} with \eqref{pf:quantum-converse-9} and Eqs.~\eqref{pf:quantum-converse-10} with \eqref{pf:quantum-converse-12}, we have that
\begin{align}
	1-q_\abb{infid}\fleft(\fleft(\ch{E}_{ML\to A},\ch{D}_{B\to\g{M}L}\fright),\ch{N}_{A\to B}\fright)&\leq\frac{\lambda}{d_M^2d_A^2\tr\fleft[\Gamma_A\fleft\{\ch{E}_{ML\to A}\circ\ch{D}_{B\to\g{M}L}\circ\ch{N}_{A\to B}\fright\}\fleft[\Phi_{M'M}\fright]\fright]}, \label{pf:quantum-exponent-1}\\
	\mu q_\abb{infid}\fleft(\fleft(\ch{E}_{ML\to A},\ch{D}_{B\to\g{M}L}\fright),\ch{N}_{A\to B}\fright)&\geq\left(1-\frac{1}{d_M^2}\right)\frac{1}{d_A^2\tr\fleft[\Gamma_A\fleft\{\ch{E}_{ML\to A}\circ\ch{D}_{B\to\g{M}L}\circ\ch{N}_{A\to B}\fright\}\fleft[\Phi_{M'M}\fright]\fright]}. \label{pf:quantum-exponent-2}
\end{align}
Note that combining Eqs.~\eqref{pf:quantum-exponent-1} and \eqref{pf:quantum-exponent-2} implies that
\begin{align}
	q_\abb{infid}\fleft(\fleft(\ch{E}_{ML\to A},\ch{D}_{B\to\g{M}L}\fright),\ch{N}_{A\to B}\fright)&\geq\frac{d_M^2-1}{\lambda\mu+d_M^2-1}.
\end{align}
For a fixed $d_M$, infimizing the left-hand side over every $d_L\in\spa{N}$ and every strategy $(\ch{E}_{ML\to A},\ch{D}_{B\to\g{M}L})$, and supremizing the right-hand side over every $\lambda,\mu\in\spa{R}$, $\sigma_B\in\s{D}_B$, and $\tau_B\in\aff(\s{D}_B)$ satisfying Eqs.~\eqref{pf:quantum-converse-2} and \eqref{pf:quantum-converse-3}, we have that
\begin{align}
	\inf_{\substack{d_L\in\spa{N}, \\
	\ch{E}_{ML\to A}\in\s{C}_{ML\to A}, \\ \ch{D}_{B\to\g{M}L}\in\s{C}_{B\to\g{M}L}}}q_\abb{infid}\fleft(\fleft(\ch{E}_{ML\to A},\ch{D}_{B\to\g{M}L}\fright),\ch{N}_{A\to B}\fright)&\geq\frac{d_M^2-1}{2^{I_{\max}\fleft(\ch{N}\fright)+I_\abb{doe}\fleft(\ch{N}\fright)}+d_M^2-1}. \label{pf:quantum-exponent-3}
\end{align}
Now consider the strategy $(\ch{E}_{ML_1L_2\to A},\ch{D}_{B\to\g{M}L_1L_2})$ defined in Eqs.~\eqref{pf:quantum-achievability-1} and \eqref{pf:quantum-achievability-2}.  It follows from Eq.~\eqref{pf:quantum-achievability-5} that, for every pure state $\psi_{RM}$,
\begin{align}
	&1-F\fleft(\psi_{R\g{M}}\middle|\psi_{RM};\fleft(\ch{E}_{ML_1L_2\to A},\ch{D}_{B\to\g{M}L_1L_2}\fright),\ch{N}_{A\to B}\fright) \notag\\
	&\leq\left(\frac{\tr\fleft[\Phi_{AA'}\left(\ch{K}_{B\to A}^{(0)}\circ\ch{N}_{A\to B}\right)\fleft[\Phi_{AA'}\fright]\fright]}{\left(d_M^2-1\right)\tr\fleft[\Phi_{AA'}\left(\ch{K}_{B\to A}^{(1)}\circ\ch{N}_{A\to B}\right)\fleft[\Phi_{AA'}\fright]\fright]}+1\right)^{-1}.
\end{align}
Supremizing the left-hand side over every $d_R\in\spa{N}$ and every pure state $\psi_{RM}$, and choosing the channels $\ch{K}_{B\to A}^{(0)}$ and $\ch{K}_{B\to A}^{(1)}$ according to Eqs.~\eqref{pf:quantum-achievability-9} and \eqref{pf:quantum-achievability-10}, it follows from Eq.~\eqref{eq:infidelity} that
\begin{align}
	&q_\abb{infid}\fleft(\fleft(\ch{E}_{ML_1L_2\to A},\ch{D}_{B\to\g{M}L_1L_2}\fright),\ch{N}_{A\to B}\fright) \notag\\
	&\leq\left(\frac{\tr\fleft[\Phi_{AA'}\left(\ch{K}_{B\to A}^{(0)}\circ\ch{N}_{A\to B}\right)\fleft[\Phi_{AA'}\fright]\fright]}{\left(d_M^2-1\right)\tr\fleft[\Phi_{AA'}\left(\ch{K}_{B\to A}^{(1)}\circ\ch{N}_{A\to B}\right)\fleft[\Phi_{AA'}\fright]\fright]}+1\right)^{-1} \\
	&=\left(\frac{\max_{\ch{K}_{B\to A}\in\s{C}_{B\to A}}\tr\fleft[\Phi_{AA'}\left(\ch{K}_{B\to A}\circ\ch{N}_{A\to B}\right)\fleft[\Phi_{AA'}\fright]\fright]}{\left(d_M^2-1\right)\min_{\ch{K}_{B\to A}\in\s{C}_{B\to A}}\tr\fleft[\Phi_{AA'}\left(\ch{K}_{B\to A}\circ\ch{N}_{A\to B}\right)\fleft[\Phi_{AA'}\fright]\fright]}+1\right)^{-1} \\
	&=\left(\frac{1}{d_M^2-1}2^{I_{\max}\fleft(\ch{N}\fright)+I_\abb{doe}\fleft(\ch{N}\fright)}+1\right)^{-1} \label{pf:quantum-exponent-4}\\
	&=\frac{d_M^2-1}{2^{I_{\max}\fleft(\ch{N}\fright)+I_\abb{doe}\fleft(\ch{N}\fright)}+d_M^2-1}, \label{pf:quantum-exponent-5}
\end{align}
where Eq.~\eqref{pf:quantum-exponent-4} follows from the interpretation of the max- and Doeblin informations in terms of the maximum and minimum achievable singlet fractions, respectively (Lemma~\ref{lem:fraction}).  Combining Eqs.~\eqref{pf:quantum-exponent-3} and \eqref{pf:quantum-exponent-5}, we have that
\begin{align}
	\inf_{\substack{d_L\in\spa{N}, \\
	\ch{E}_{ML\to A}\in\s{C}_{ML\to A}, \\ \ch{D}_{B\to\g{M}L}\in\s{C}_{B\to\g{M}L}}}q_\abb{infid}\fleft(\fleft(\ch{E}_{ML\to A},\ch{D}_{B\to\g{M}L}\fright),\ch{N}_{A\to B}\fright)&=\frac{d_M^2-1}{2^{I_{\max}\fleft(\ch{N}\fright)+I_\abb{doe}\fleft(\ch{N}\fright)}+d_M^2-1}.
\end{align}
Then it follows from the definition of the error exponent of retrocausal quantum communication [Eq.~\eqref{eq:quantum-error}] that
\begin{align}
	E_\abb{retro,Q}^r\fleft(\ch{N}\fright)&=\lim_{n\to\infty}\sup_{d_M\in\spa{N}}\left\{-\frac{1}{n}\log_2\frac{d_M^2-1}{2^{I_{\max}\fleft(\ch{N}^{\otimes n}\fright)+I_\abb{doe}\fleft(\ch{N}^{\otimes n}\fright)}+d_M^2-1}\colon d_M\geq2^{nr}\right\} \\
	&=\lim_{n\to\infty}-\frac{1}{n}\log_2\frac{2^{2nr}-1}{2^{I_{\max}\fleft(\ch{N}^{\otimes n}\fright)+I_\abb{doe}\fleft(\ch{N}^{\otimes n}\fright)}+2^{2nr}-1} \\
	&=\max\left\{0,I_{\max}\fleft(\ch{N}\fright)+I_\abb{doe}^\infty\fleft(\ch{N}\fright)-2r\right\}. \label{pf:quantum-exponent-6}
\end{align}
where Eq.~\eqref{pf:quantum-exponent-6} follows from the additivity of the max-information [Eq.~\eqref{eq:max-regularized}] and the definition of the regularized Doeblin information [Eq.~\eqref{eq:doeblin-regularized}].  Likewise, it follows from the definition of the strong-converse exponent of retrocausal quantum communication [Eq.~\eqref{eq:quantum-strong}] that
\begin{align}
	S_\abb{retro,Q}^r\fleft(\ch{N}\fright)&=\lim_{n\to\infty}\inf_{d_M\in\spa{N}}\left\{-\frac{1}{n}\log_2\left(1-\frac{d_M^2-1}{2^{I_{\max}\fleft(\ch{N}^{\otimes n}\fright)+I_\abb{doe}\fleft(\ch{N}^{\otimes n}\fright)}+d_M^2-1}\right)\colon d_M\geq2^{nr}\right\} \\
	&=\lim_{n\to\infty}-\frac{1}{n}\log_2\frac{2^{I_{\max}\fleft(\ch{N}^{\otimes n}\fright)+I_\abb{doe}\fleft(\ch{N}^{\otimes n}\fright)}}{2^{I_{\max}\fleft(\ch{N}^{\otimes n}\fright)+I_\abb{doe}\fleft(\ch{N}^{\otimes n}\fright)}+2^{2nr}-1} \\
	&=\max\left\{0,2r-I_{\max}\fleft(\ch{N}\fright)-I_\abb{doe}^\infty\fleft(\ch{N}\fright)\right\},
\end{align}
which completes the proof of the desired statement.
\end{proof}

\section{Retrocausal classical capacity}
\label{sec:classical}

In this section, we precisely determine the one-shot retrocausal classical capacity.  All proofs follow a similar idea to their quantum counterparts.

\begin{theorem}[One-shot retrocausal classical capacity]
\label{thm:classical-oneshot}
The one-shot retrocausal classical capacity has a closed-form expression in terms of the max- and Doeblin informations: for a completely positive map $\ch{N}_{A\to B}\in\s{L}_{A\to B}$ and a real number $\varepsilon\in(0,1)$,
\begin{align}
	C_\abb{retro}^\varepsilon\fleft(\ch{N}\fright)&=\log_2\left\lfloor\frac{\varepsilon}{1-\varepsilon}2^{I_{\max}\fleft(\ch{N}\fright)+I_\abb{doe}\fleft(\ch{N}\fright)}+1\right\rfloor,
\end{align}
where the max- and Doeblin informations are defined in Eqs.~\eqref{eq:max} and \eqref{eq:doeblin}, respectively.
\end{theorem}

As in the quantum case, we prove Theorem~\ref{thm:classical-oneshot} by first establishing the desired converse bound on the one-shot retrocausal classical capacity and then showing that it is achievable by construction.

\subsection{Converse bound}
\label{sec:classical-converse}

We now prove the desired converse bound on the one-shot retrocausal classical capacity.

\begin{proposition}[Converse bound on retrocausal classical capacity]
\label{prop:classical-converse}
For a completely positive map $\ch{N}_{A\to B}\in\s{L}_{A\to B}$ and a real number $\varepsilon\in(0,1)$,
\begin{align}
	C_\abb{retro}^\varepsilon\fleft(\ch{N}\fright)&\leq\log_2\left\lfloor\frac{\varepsilon}{1-\varepsilon}2^{I_{\max}\fleft(\ch{N}\fright)+I_\abb{doe}\fleft(\ch{N}\fright)}+1\right\rfloor.
\end{align}
\end{proposition}

\begin{proof}
Let $(\ch{E}_{ML\to A},\ch{D}_{B\to\g{M}L})$ be a strategy such that
\begin{align}
	p_\abb{error}\fleft(\fleft(\ch{E}_{ML\to A},\ch{D}_{B\to\g{M}L}\fright),\ch{N}_{A\to B}\fright)&\leq\varepsilon. \label{pf:classical-converse-1}
\end{align}
It follows from Eqs.~\eqref{eq:error} and \eqref{eq:probability} that, for every symbol $m$ in the alphabet,
\begin{align}
	1-\varepsilon&\leq1-p_\abb{error}\fleft(\fleft(\ch{E}_{ML\to A},\ch{D}_{B\to\g{M}L}\fright),\ch{N}_{A\to B}\fright) \\
	&\leq\frac{\tr\fleft[\op{m}{m}_\g{M}\Gamma_A\fleft\{\ch{E}_{ML\to A}\circ\ch{D}_{B\to\g{M}L}\circ\ch{N}_{A\to B}\fright\}\fleft[\op{m}{m}_M\fright]\fright]}{\tr\fleft[\Gamma_A\fleft\{\ch{E}_{ML\to A}\circ\ch{D}_{B\to\g{M}L}\circ\ch{N}_{A\to B}\fright\}\fleft[\op{m}{m}_M\fright]\fright]},
\end{align}
which implies that
\begin{align}
	1-\varepsilon&\leq1-p_\abb{error}\fleft(\fleft(\ch{E}_{ML\to A},\ch{D}_{B\to\g{M}L}\fright),\ch{N}_{A\to B}\fright) \\
	&\leq\frac{\sum_{m=0}^{d_M-1}\frac{1}{d_M}\tr\fleft[\op{m}{m}_\g{M}\Gamma_A\fleft\{\ch{E}_{ML\to A}\circ\ch{D}_{B\to\g{M}L}\circ\ch{N}_{A\to B}\fright\}\fleft[\op{m}{m}_M\fright]\fright]}{\sum_{m=0}^{d_M-1}\frac{1}{d_M}\tr\fleft[\Gamma_A\fleft\{\ch{E}_{ML\to A}\circ\ch{D}_{B\to\g{M}L}\circ\ch{N}_{A\to B}\fright\}\fleft[\op{m}{m}_M\fright]\fright]} \\
	&=\frac{\tr\fleft[\Omega_{M'\g{M}}\Gamma_A\fleft\{\ch{E}_{ML\to A}\circ\ch{D}_{B\to\g{M}L}\circ\ch{N}_{A\to B}\fright\}\fleft[\Upsilon_{M'M}\fright]\fright]}{\tr\fleft[\Gamma_A\fleft\{\ch{E}_{ML\to A}\circ\ch{D}_{B\to\g{M}L}\circ\ch{N}_{A\to B}\fright\}\fleft[\Upsilon_{M'M}\fright]\fright]}, \label{pf:classical-converse-2}
\end{align}
where the standard classically maximally correlated state $\Upsilon_{M'M}$ and the classical comparator $\Omega_{M'\g{M}}$ are defined in Sec.~\ref{sec:notation}.  Let $\lambda,\mu\in\spa{R}$ be two real numbers, $\sigma_B\in\s{D}_B$ a state, and $\tau_B\in\aff(\s{D}_B)$ a unit-trace Hermitian operator, such that
\begin{align}
	\Phi_{A'B}^\ch{N}&\leq\lambda\Phi_{A'B}^{\ch{R}^\sigma}, \label{pf:classical-converse-3}\\
	\Phi_{A'B}^{\ch{R}^\tau}&\leq\mu\Phi_{A'B}^\ch{N}. \label{pf:classical-converse-4}
\end{align}
Equation~\eqref{pf:classical-converse-3} means that $\lambda\ch{R}_{A\to B}^\sigma-\ch{N}_{A\to B}$ is a completely positive map, and therefore so is $\lambda\Gamma_A\fleft\{\ch{E}_{ML\to A}\circ\ch{D}_{B\to\g{M}L}\circ\ch{R}_{A\to B}^\sigma\fright\}-\Gamma_A\fleft\{\ch{E}_{ML\to A}\circ\ch{D}_{B\to\g{M}L}\circ\ch{N}_{A\to B}\fright\}$.  This implies that
\begin{align}
	\Gamma_A\fleft\{\ch{E}_{ML\to A}\circ\ch{D}_{B\to\g{M}L}\circ\ch{N}_{A\to B}\fright\}\fleft[\Upsilon_{M'M}\fright]&\leq\lambda\Gamma_A\fleft\{\ch{E}_{ML\to A}\circ\ch{D}_{B\to\g{M}L}\circ\ch{R}_{A\to B}^\sigma\fright\}\fleft[\Upsilon_{M'M}\fright] \\
	&=\frac{\lambda}{d_A^2}\ch{R}_{M\to\g{M}}^\omega\fleft[\Upsilon_{M'M}\fright] \\
	&=\frac{\lambda}{d_A^2}\pi_{M'}\otimes\omega_\g{M}, \label{pf:classical-converse-5}
\end{align}
where $\omega_\g{M}\in\s{D}_\g{M}$ is a state, whose existence is guaranteed by Lemma~\ref{lem:useless}.  Likewise, it follows from Eq.~\eqref{pf:classical-converse-4} that
\begin{align}
	\mu\Gamma_A\fleft\{\ch{E}_{ML\to A}\circ\ch{D}_{B\to\g{M}L}\circ\ch{N}_{A\to B}\fright\}\fleft[\Upsilon_{M'M}\fright]&\geq\Gamma_A\fleft\{\ch{E}_{ML\to A}\circ\ch{D}_{B\to\g{M}L}\circ\ch{R}_{A\to B}^\tau\fright\}\fleft[\Upsilon_{M'M}\fright] \\
	&=\frac{1}{d_A^2}\ch{R}_{M\to\g{M}}^\eta\fleft[\Upsilon_{M'M}\fright] \\
	&=\frac{1}{d_A^2}\pi_{M'}\otimes\eta_\g{M}, \label{pf:classical-converse-6}
\end{align}
where $\eta_\g{M}\in\aff(\s{D}_\g{M})$ is a unit-trace Hermitian operator, whose existence is guaranteed by Lemma~\ref{lem:useless}.  Then it follows from Eq.~\eqref{pf:classical-converse-2} that
\begin{align}
	1-\varepsilon&\leq1-p_\abb{error}\fleft(\fleft(\ch{E}_{ML\to A},\ch{D}_{B\to\g{M}L}\fright),\ch{N}_{A\to B}\fright) \label{pf:classical-converse-7}\\
	&\leq\frac{\tr\fleft[\Omega_{M'\g{M}}\Gamma_A\fleft\{\ch{E}_{ML\to A}\circ\ch{D}_{B\to\g{M}L}\circ\ch{N}_{A\to B}\fright\}\fleft[\Upsilon_{M'M}\fright]\fright]}{\tr\fleft[\Gamma_A\fleft\{\ch{E}_{ML\to A}\circ\ch{D}_{B\to\g{M}L}\circ\ch{N}_{A\to B}\fright\}\fleft[\Upsilon_{M'M}\fright]\fright]} \\
	&\leq\frac{\lambda\tr\fleft[\Omega_{M'\g{M}}\left(\pi_{M'}\otimes\omega_\g{M}\right)\fright]}{d_A^2\tr\fleft[\Gamma_A\fleft\{\ch{E}_{ML\to A}\circ\ch{D}_{B\to\g{M}L}\circ\ch{N}_{A\to B}\fright\}\fleft[\Upsilon_{M'M}\fright]\fright]} \label{pf:classical-converse-8}\\
	&=\frac{\lambda}{d_Md_A^2\tr\fleft[\Gamma_A\fleft\{\ch{E}_{ML\to A}\circ\ch{D}_{B\to\g{M}L}\circ\ch{N}_{A\to B}\fright\}\fleft[\Upsilon_{M'M}\fright]\fright]}, \label{pf:classical-converse-9}
\end{align}
where Eq.~\eqref{pf:classical-converse-8} follows from Eq.~\eqref{pf:classical-converse-5}.  Similarly, it follows from Eq.~\eqref{pf:classical-converse-2} that
\begin{align}
	\mu\varepsilon&\geq\mu p_\abb{error}\fleft(\fleft(\ch{E}_{ML\to A},\ch{D}_{B\to\g{M}L}\fright),\ch{N}_{A\to B}\fright) \label{pf:classical-converse-10}\\
	&\geq\frac{\mu\tr\fleft[\left(\1_{M'\g{M}}-\Omega_{M'\g{M}}\right)\Gamma_A\fleft\{\ch{E}_{ML\to A}\circ\ch{D}_{B\to\g{M}L}\circ\ch{N}_{A\to B}\fright\}\fleft[\Upsilon_{M'M}\fright]\fright]}{\tr\fleft[\Gamma_A\fleft\{\ch{E}_{ML\to A}\circ\ch{D}_{B\to\g{M}L}\circ\ch{N}_{A\to B}\fright\}\fleft[\Upsilon_{M'M}\fright]\fright]} \\
	&\geq\frac{\tr\fleft[\left(\1_{M'\g{M}}-\Omega_{M'\g{M}}\right)\left(\pi_{M'}\otimes\eta_\g{M}\right)\fright]}{d_A^2\tr\fleft[\Gamma_A\fleft\{\ch{E}_{ML\to A}\circ\ch{D}_{B\to\g{M}L}\circ\ch{N}_{A\to B}\fright\}\fleft[\Upsilon_{M'M}\fright]\fright]} \label{pf:classical-converse-11}\\
	&=\left(1-\frac{1}{d_M}\right)\frac{1}{d_A^2\tr\fleft[\Gamma_A\fleft\{\ch{E}_{ML\to A}\circ\ch{D}_{B\to\g{M}L}\circ\ch{N}_{A\to B}\fright\}\fleft[\Upsilon_{M'M}\fright]\fright]}, \label{pf:classical-converse-12}
\end{align}
where Eq.~\eqref{pf:classical-converse-11} follows from Eq.~\eqref{pf:classical-converse-6}.  Note that combining Eqs.~\eqref{pf:classical-converse-9} and \eqref{pf:classical-converse-12} implies that
\begin{align}
	d_M\leq\frac{\varepsilon\lambda\mu}{1-\varepsilon}+1,
\end{align}
which, due to $d_M$ being an integer, further implies that
\begin{align}
	\log_2d_M\leq\log_2\left\lfloor\frac{\varepsilon\lambda\mu}{1-\varepsilon}+1\right\rfloor.
\end{align}
Supremizing the left-hand side over every $d_M,d_L\in\spa{N}$ and every strategy $(\ch{E}_{ML\to A},\ch{D}_{B\to\g{M}L})$ satisfying Eq.~\eqref{pf:classical-converse-1}, and infimizing the right-hand side over every $\lambda,\mu\in\spa{R}$, $\sigma_B\in\s{D}_B$, and $\tau_B\in\aff(\s{D}_B)$ satisfying Eqs.~\eqref{pf:classical-converse-3} and \eqref{pf:classical-converse-4}, the desired statement thus follows from the definitions of the one-shot retrocausal classical capacity and the max- and Doeblin informations [Eqs.~\eqref{eq:classical-oneshot}, \eqref{eq:max}, and \eqref{eq:doeblin}].
\end{proof}

\subsection{Achievability}
\label{sec:classical-achievability}

We now construct a strategy that saturates the converse bound derived in Proposition~\ref{prop:classical-converse} by slightly adapting Strategy~\ref{str:quantum}.

\begin{strategy}[Amplified probabilistic classical teleportation]
\label{str:classical}
The strategy operates as follows (see Fig.~\ref{fig:classical-achievability}).
\begin{itemize}
	\item \textbf{Encoding.}  In the future, the father retrieves a classical memory $L_1$ and a quantum memory $L_2$ (both of which were left by Bella in the past) with $L_1\cong M$ and $L_2\cong B$.  He reads the data $l_1$ in $L_1$ and compares it with his message $m$.  He sets $k\coloneq0$ if $m=l_1$ and sets $k\coloneq1$ otherwise.  Conditioned on $k$, he applies a channel $\ch{K}_{L_2\to A}^{(k)}\in\s{C}_{L_2\to A}$ to $L_2$.  He then feeds $A$ into the entrance of the noisy P-CTC.
	\item \textbf{Decoding.}  In the past, the daughter receives $B$ from the exit of the noisy P-CTC.  She generates a symbol $\g{m}$ from the alphabet uniformly at random, writes a copy of $l_1\coloneq\g{m}$ into $L_1$, renames $B$ as $L_2$, and leaves both $L_1$ and $L_2$ for the father to retrieve in the future.  She then declares that $\g{m}$ is her recovery of the message.
\end{itemize}
\end{strategy}

\begin{figure}[t]
\includegraphics[scale=0.3]{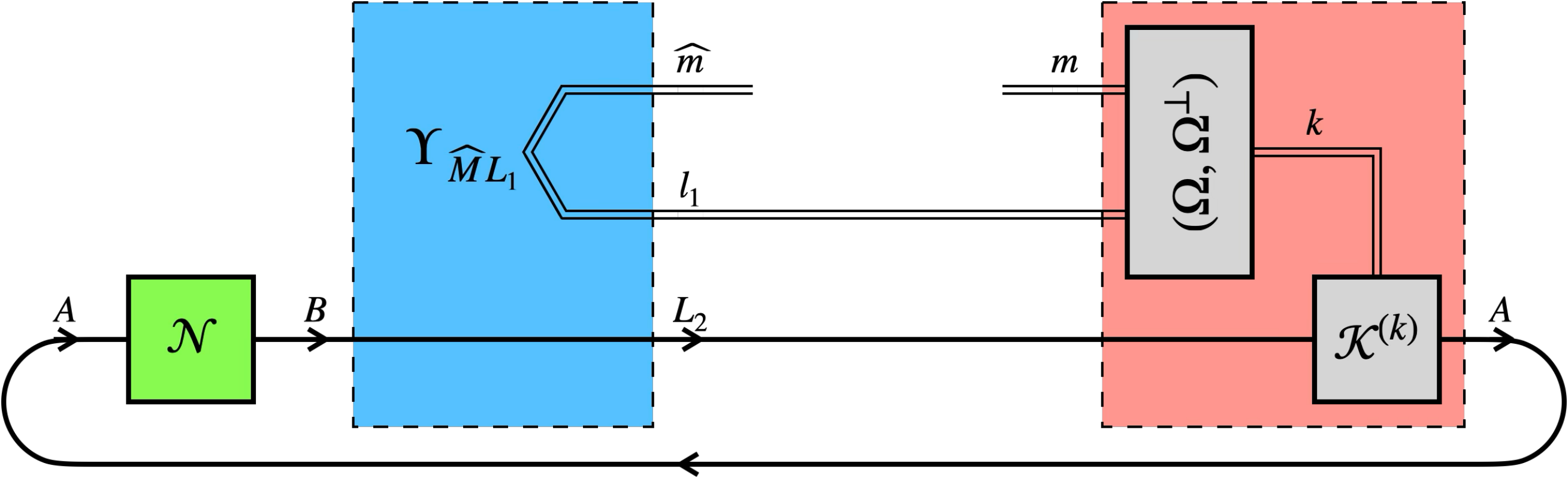}
\caption{The basic idea of amplified probabilistic classical teleportation is to first perform a classical version of probabilistic teleportation from the father to the daughter, and then engineer a causal loop involving the noisy P-CTC such that it maximally amplifies the renormalized probability of the teleportation being successful.  The red and blue boxes encapsulate the father's encoding and the daughter's decoding, respectively.}
\label{fig:classical-achievability}
\end{figure}

The intuition behind Strategy~\ref{str:classical} is the same as that behind its quantum counterpart, Strategy~\ref{str:quantum}.  Instead of trying to code his message into $A$ which enters the noisy P-CTC, the father intends to ``teleport'' the message probabilistically to the daughter via $L_1$, using shared randomness.  Indeed, whenever the father sets $k=0$, his message $m$ has been found to be identical to the data $l_1$ in $L_1$, which is in turn identical to the daughter's recovery $\g{m}$.  This occurs with probability $1/d_M$ without invoking the noisy P-CTC.  The noisy P-CTC and the channels $\ch{K}_{L_2\to A}^{(0)}$ and $\ch{K}_{L_2\to A}^{(1)}$ are then utilized in the same manner as in Strategy~\ref{str:quantum}, to maximally amplify the renormalized probability of the event $k=0$.

We now formalize this analysis and show that Strategy~\ref{str:classical} can indeed achieve the maximum possible efficiency allowed by the converse bound established in Proposition~\ref{prop:classical-converse}.

\begin{proposition}[Achievability of retrocausal classical capacity]
\label{prop:classical-achievability}
For a completely positive map $\ch{N}_{A\to B}\in\s{L}_{A\to B}$ and a real number $\varepsilon\in(0,1)$,
\begin{align}
	C_\abb{retro}^\varepsilon\fleft(\ch{N}\fright)&\geq\log_2\left\lfloor\frac{\varepsilon}{1-\varepsilon}2^{I_{\max}\fleft(\ch{N}\fright)+I_\abb{doe}\fleft(\ch{N}\fright)}+1\right\rfloor. \label{eq:classical-achievability}
\end{align}
Furthermore, the efficiency on the right-hand side of Eq.~\eqref{eq:classical-achievability} is achievable by Strategy~\ref{str:classical} for some channels $\ch{K}_{L_2\to A}^{(0)}$ and $\ch{K}_{L_2\to A}^{(1)}$.
\end{proposition}

\begin{proof}
The encoding and decoding of Strategy~\ref{str:classical} are represented by the following channels, respectively:
\begin{align}
	\ch{E}_{ML_1L_2\to A}\fleft[\cdot\fright]&\coloneq\tr_{ML_1}\fleft[\Omega_{ML_1}\ch{K}_{L_2\to A}^{(0)}\fleft[\cdot\fright]\fright]+\tr_{ML_1}\fleft[\left(\1_{ML_1}-\Omega_{ML_1}\right)\ch{K}_{L_2\to A}^{(1)}\fleft[\cdot\fright]\fright], \label{pf:classical-achievability-1}\\
	\ch{D}_{B\to\g{M}L_1L_2}\fleft[\cdot\fright]&\coloneq\Upsilon_{\g{M}L_1}\otimes\id_{B\to L_2}\fleft[\cdot\fright]. \label{pf:classical-achievability-2}
\end{align}
Let $m$ be a symbol in the alphabet.  Then it follows from the definition of the loop supermap [Eq.~\eqref{eq:loop}] that
\begin{align}
	&\Gamma_A\fleft\{\ch{E}_{ML_1L_2\to A}\circ\ch{D}_{B\to\g{M}L_1L_2}\circ\ch{N}_{A\to B}\fright\}\fleft[\op{m}{m}_M\fright] \notag\\
	&=\tr_{AA'}\fleft[\Phi_{AA'}\left(\ch{E}_{ML_1L_2\to A}\circ\ch{D}_{B\to\g{M}L_1L_2}\circ\ch{N}_{A\to B}\right)\fleft[\op{m}{m}_M\otimes\Phi_{AA'}\fright]\fright] \\
	&=\tr_{ML_1AA'}\fleft[\left(\Omega_{ML_1}\otimes\Phi_{AA'}\right)\left(\ch{K}_{B\to A}^{(0)}\circ\ch{N}_{A\to B}\right)\fleft[\op{m}{m}_M\otimes\Upsilon_{\g{M}L_1}\otimes\Phi_{AA'}\fright]\fright] \notag\\
	&\quad+\tr_{ML_1AA'}\fleft[\left(\left(\1_{ML_1}-\Omega_{ML_1}\right)\otimes\Phi_{AA'}\right)\left(\ch{K}_{B\to A}^{(1)}\circ\ch{N}_{A\to B}\right)\fleft[\op{m}{m}_M\otimes\Upsilon_{\g{M}L_1}\otimes\Phi_{AA'}\fright]\fright] \\
	&=\frac{1}{d_M}\tr\fleft[\Phi_{AA'}\left(\ch{K}_{B\to A}^{(0)}\circ\ch{N}_{A\to B}\right)\fleft[\Phi_{AA'}\fright]\fright]\op{m}{m}_\g{M} \notag\\
	&\quad+\frac{1}{d_M}\tr\fleft[\Phi_{AA'}\left(\ch{K}_{B\to A}^{(1)}\circ\ch{N}_{A\to B}\right)\fleft[\Phi_{AA'}\fright]\fright]\left(\1_\g{M}-\op{m}{m}_\g{M}\right). \label{pf:classical-achievability-3}
\end{align}
It follows from Eq.~\eqref{eq:probability} that
\begin{align}
	&1-\Pr\fleft\{m\middle|m;\fleft(\ch{E}_{ML_1L_2\to A},\ch{D}_{B\to\g{M}L_1L_2}\fright),\ch{N}_{A\to B}\fright\} \notag\\
	&=1-\frac{\tr\fleft[\op{m}{m}_\g{M}\Gamma_A\fleft\{\ch{E}_{ML_1L_2\to A}\circ\ch{D}_{B\to\g{M}L_1L_2}\circ\ch{N}_{A\to B}\fright\}\fleft[\op{m}{m}_M\fright]\fright]}{\tr\fleft[\Gamma_A\fleft\{\ch{E}_{ML_1L_2\to A}\circ\ch{D}_{B\to\g{M}L_1L_2}\circ\ch{N}_{A\to B}\fright\}\fleft[\op{m}{m}_M\fright]\fright]} \\
	&\leq1-\frac{\frac{1}{d_M}\tr\fleft[\Phi_{AA'}\left(\ch{K}_{B\to A}^{(0)}\circ\ch{N}_{A\to B}\right)\fleft[\Phi_{AA'}\fright]\fright]}{\frac{1}{d_M}\tr\fleft[\Phi_{AA'}\left(\ch{K}_{B\to A}^{(0)}\circ\ch{N}_{A\to B}\right)\fleft[\Phi_{AA'}\fright]\fright]+\frac{d_M-1}{d_M}\tr\fleft[\Phi_{AA'}\left(\ch{K}_{B\to A}^{(1)}\circ\ch{N}_{A\to B}\right)\fleft[\Phi_{AA'}\fright]\fright]} \label{pf:classical-achievability-4}\\
	&=\left(\frac{\tr\fleft[\Phi_{AA'}\left(\ch{K}_{B\to A}^{(0)}\circ\ch{N}_{A\to B}\right)\fleft[\Phi_{AA'}\fright]\fright]}{\left(d_M-1\right)\tr\fleft[\Phi_{AA'}\left(\ch{K}_{B\to A}^{(1)}\circ\ch{N}_{A\to B}\right)\fleft[\Phi_{AA'}\fright]\fright]}+1\right)^{-1}, \label{pf:classical-achievability-5}
\end{align}
where Eq.~\eqref{pf:classical-achievability-4} follows from Eq.~\eqref{pf:classical-achievability-3}.  Since this holds for every symbol $m$, it follows from Eq.~\eqref{eq:error} that, as long as
\begin{align}
	\left(\frac{\tr\fleft[\Phi_{AA'}\left(\ch{K}_{B\to A}^{(0)}\circ\ch{N}_{A\to B}\right)\fleft[\Phi_{AA'}\fright]\fright]}{\left(d_M-1\right)\tr\fleft[\Phi_{AA'}\left(\ch{K}_{B\to A}^{(1)}\circ\ch{N}_{A\to B}\right)\fleft[\Phi_{AA'}\fright]\fright]}+1\right)^{-1}&\leq\varepsilon, \label{pf:classical-achievability-6}
\end{align}
we have that
\begin{align}
	p_\abb{error}\fleft(\fleft(\ch{E}_{ML_1L_2\to A},\ch{D}_{B\to\g{M}L_1L_2}\fright),\ch{N}_{A\to B}\fright)&=1-\min_{m}\Pr\fleft\{m\middle|m;\fleft(\ch{E}_{ML_1L_2\to A},\ch{D}_{B\to\g{M}L_1L_2}\fright),\ch{N}_{A\to B}\fright\} \\
	&\leq\varepsilon. \label{pf:classical-achievability-7}
\end{align}
Since Eq.~\eqref{pf:classical-achievability-6} is equivalent to
\begin{align}
	d_M&\leq\frac{\varepsilon\tr\fleft[\Phi_{AA'}\left(\ch{K}_{B\to A}^{(0)}\circ\ch{N}_{A\to B}\right)\fleft[\Phi_{AA'}\fright]\fright]}{\left(1-\varepsilon\right)\tr\fleft[\Phi_{AA'}\left(\ch{K}_{B\to A}^{(1)}\circ\ch{N}_{A\to B}\right)\fleft[\Phi_{AA'}\fright]\fright]}+1, \label{pf:classical-achievability-8}
\end{align}
by choosing
\begin{align}
	\ch{K}_{B\to A}^{(0)}&\coloneq\arg\max_{\ch{K}_{B\to A}\in\s{C}_{B\to A}}\tr\fleft[\Phi_{AA'}\left(\ch{K}_{B\to A}\circ\ch{N}_{A\to B}\right)\fleft[\Phi_{AA'}\fright]\fright], \label{pf:classical-achievability-9}\\
	\ch{K}_{B\to A}^{(1)}&\coloneq\arg\min_{\ch{K}_{B\to A}\in\s{C}_{B\to A}}\tr\fleft[\Phi_{AA'}\left(\ch{K}_{B\to A}\circ\ch{N}_{A\to B}\right)\fleft[\Phi_{AA'}\fright]\fright], \label{pf:classical-achievability-10}
\end{align}
which maximize the right-hand side of Eq.~\eqref{pf:classical-achievability-8}, the following efficiency can be achieved within error $\varepsilon$ (i.e., subject to Eq.~\eqref{pf:classical-achievability-7}) and is thus a lower bound on the one-shot retrocausal classical capacity [see Eq.~\eqref{eq:classical-oneshot}]:
\begin{align}
	\log_2d_M&=\log_2\left\lfloor\frac{\varepsilon\tr\fleft[\Phi_{AA'}\left(\ch{K}_{B\to A}^{(0)}\circ\ch{N}_{A\to B}\right)\fleft[\Phi_{AA'}\fright]\fright]}{\left(1-\varepsilon\right)\tr\fleft[\Phi_{AA'}\left(\ch{K}_{B\to A}^{(1)}\circ\ch{N}_{A\to B}\right)\fleft[\Phi_{AA'}\fright]\fright]}+1\right\rfloor \\
	&=\log_2\left\lfloor\frac{\varepsilon\max_{\ch{K}_{B\to A}\in\s{C}_{B\to A}}\tr\fleft[\Phi_{AA'}\left(\ch{K}_{B\to A}\circ\ch{N}_{A\to B}\right)\fleft[\Phi_{AA'}\fright]\fright]}{\left(1-\varepsilon\right)\min_{\ch{K}_{B\to A}\in\s{C}_{B\to A}}\tr\fleft[\Phi_{AA'}\left(\ch{K}_{B\to A}\circ\ch{N}_{A\to B}\right)\fleft[\Phi_{AA'}\fright]\fright]}+1\right\rfloor \\
	&=\log_2\left\lfloor\frac{\varepsilon}{1-\varepsilon}2^{I_{\max}\fleft(\ch{N}\fright)+I_\abb{doe}\fleft(\ch{N}\fright)}+1\right\rfloor, \label{pf:classical-achievability-11}
\end{align}
where Eq.~\eqref{pf:classical-achievability-11} follows from the interpretation of the max- and Doeblin informations in terms of the maximum and minimum achievable singlet fractions, respectively (Lemma~\ref{lem:fraction}), implying the desired statement.
\end{proof}

Combining Propositions~\ref{prop:classical-converse} and \ref{prop:classical-achievability} completes the proof of Theorem~\ref{thm:classical-oneshot}.

\subsection{Asymptotics}
\label{sec:classical-asymptotics}

With the one-shot retrocausal classical capacity determined, the asymptotic retrocausal classical capacity immediately follows.

\begin{theorem}[Asymptotic retrocausal classical capacity]
\label{thm:classical-asymptotic}
The asymptotic retrocausal classical capacity is equal to the sum of the max- and regularized Doeblin informations: for a completely positive map $\ch{N}_{A\to B}\in\s{L}_{A\to B}$ and a real number $\varepsilon\in(0,1)$,
\begin{align}
	C_\abb{retro}\fleft(\ch{N}\fright)&=\lim_{n\to\infty}\frac{1}{n}Q_\abb{retro}^\varepsilon\fleft(\ch{N}^{\otimes n}\fright)=I_{\max}\fleft(\ch{N}\fright)+I_\abb{doe}^\infty\fleft(\ch{N}\fright),
\end{align}
where the regularized Doeblin information is defined in Eq.~\eqref{eq:doeblin-regularized}.
\end{theorem}

\begin{proof}
Since the max- and Doeblin informations are nonnegative, it follows from Theorem~\ref{thm:classical-oneshot} that
\begin{align}
	C_\abb{retro}^\varepsilon\fleft(\ch{N}\fright)&=\log_2\left\lfloor\frac{\varepsilon}{1-\varepsilon}2^{I_{\max}\fleft(\ch{N}\fright)+I_\abb{doe}\fleft(\ch{N}\fright)}+1\right\rfloor \\
	&\leq\log_2\left(\frac{\varepsilon}{1-\varepsilon}2^{I_{\max}\fleft(\ch{N}\fright)+I_\abb{doe}\fleft(\ch{N}\fright)}+1\right) \\
	&\leq\log_2\left(\left(\frac{\varepsilon}{1-\varepsilon}+1\right)2^{I_{\max}\fleft(\ch{N}\fright)+I_\abb{doe}\fleft(\ch{N}\fright)}\right) \\
	&=I_{\max}\fleft(\ch{N}\fright)+I_\abb{doe}\fleft(\ch{N}\fright)+\log_2\frac{1}{1-\varepsilon}. \label{pf:classical-asymptotic-1}
\end{align}
We also have that
\begin{align}
	C_\abb{retro}^\varepsilon\fleft(\ch{N}\fright)&=\log_2\left\lfloor\frac{\varepsilon}{1-\varepsilon}2^{I_{\max}\fleft(\ch{N}\fright)+I_\abb{doe}\fleft(\ch{N}\fright)}+1\right\rfloor \\
	&\geq\log_2\left(\frac{\varepsilon}{1-\varepsilon}2^{I_{\max}\fleft(\ch{N}\fright)+I_\abb{doe}\fleft(\ch{N}\fright)}\right) \\
	&=I_{\max}\fleft(\ch{N}\fright)+I_\abb{doe}\fleft(\ch{N}\fright)+\log_2\frac{\varepsilon}{1-\varepsilon}. \label{pf:classical-asymptotic-2}
\end{align}
Combining Eqs.~\eqref{pf:classical-asymptotic-1} and \eqref{pf:classical-asymptotic-2}, we have that, for every $n\in\spa{N}$,
\begin{align}
	\frac{1}{n}\left(I_{\max}\fleft(\ch{N}^{\otimes n}\fright)+I_\abb{doe}\fleft(\ch{N}^{\otimes n}\fright)+\log_2\frac{1}{1-\varepsilon}\right)&\leq\frac{1}{n}C_\abb{retro}^\varepsilon\fleft(\ch{N}^{\otimes n}\fright) \\
	&\leq\frac{1}{n}\left(I_{\max}\fleft(\ch{N}^{\otimes n}\fright)+I_\abb{doe}\fleft(\ch{N}^{\otimes n}\fright)+\log_2\frac{\varepsilon}{1-\varepsilon}\right).
\end{align}
Taking the limit $n\to\infty$, it follows from the additivity of the max-information [Eq.~\eqref{eq:max-regularized}] and the definition of the regularized Doeblin information [Eq.~\eqref{eq:doeblin-regularized}] that
\begin{align}
	\lim_{n\to\infty}\frac{1}{n}C_\abb{retro}^\varepsilon\fleft(\ch{N}^{\otimes n}\fright)&=I_{\max}\fleft(\ch{N}\fright)+I_\abb{doe}^\infty\fleft(\ch{N}\fright),
\end{align}
from which and the definition of the asymptotic retrocausal classical capacity [Eq.~\eqref{eq:classical-asymptotic}] the desired statement follows.
\end{proof}

\subsection{Error exponent and strong-converse exponent}
\label{sec:classical-exponent}

We now show that the error exponent and the strong-converse exponent of retrocausal classical communication can also be determined based on the one-shot retrocausal classical capacity.

\begin{theorem}[Error exponent and strong-converse exponent of retrocausal classical communication]
\label{thm:classical-exponent}
For a completely positive map $\ch{N}_{A\to B}\in\s{L}_{A\to B}$ and a real number $r\in[0,\infty)$,
\begin{align}
	E_\abb{retro,C}^r\fleft(\ch{N}\fright)&=\max\left\{0,I_{\max}\fleft(\ch{N}\fright)+I_\abb{doe}^\infty\fleft(\ch{N}\fright)-r\right\}, \\
	S_\abb{retro,C}^r\fleft(\ch{N}\fright)&=\max\left\{0,r-I_{\max}\fleft(\ch{N}\fright)-I_\abb{doe}^\infty\fleft(\ch{N}\fright)\right\}.
\end{align}
\end{theorem}

\begin{proof}
Let $(\ch{E}_{ML\to A},\ch{D}_{B\to\g{M}L})$ be a strategy.  Let $\lambda,\mu\in\spa{R}$ be two real numbers, $\sigma_B\in\s{D}_B$ a state, and $\tau_B\in\aff(\s{D}_B)$ a unit-trace Hermitian operator satisfying Eqs.~\eqref{pf:classical-converse-3} and \eqref{pf:classical-converse-4}.  Comparing Eqs.~\eqref{pf:classical-converse-7} with \eqref{pf:classical-converse-9} and Eqs.~\eqref{pf:classical-converse-10} with \eqref{pf:classical-converse-12}, we have that
\begin{align}
	1-p_\abb{error}\fleft(\fleft(\ch{E}_{ML\to A},\ch{D}_{B\to\g{M}L}\fright),\ch{N}_{A\to B}\fright)&\leq\frac{\lambda}{d_Md_A^2\tr\fleft[\Gamma_A\fleft\{\ch{E}_{ML\to A}\circ\ch{D}_{B\to\g{M}L}\circ\ch{N}_{A\to B}\fright\}\fleft[\Upsilon_{M'M}\fright]\fright]}, \label{pf:classical-exponent-1}\\
	\mu p_\abb{error}\fleft(\fleft(\ch{E}_{ML\to A},\ch{D}_{B\to\g{M}L}\fright),\ch{N}_{A\to B}\fright)&\geq\left(1-\frac{1}{d_M}\right)\frac{1}{d_A^2\tr\fleft[\Gamma_A\fleft\{\ch{E}_{ML\to A}\circ\ch{D}_{B\to\g{M}L}\circ\ch{N}_{A\to B}\fright\}\fleft[\Upsilon_{M'M}\fright]\fright]}. \label{pf:classical-exponent-2}
\end{align}
Note that combining Eqs.~\eqref{pf:classical-exponent-1} and \eqref{pf:classical-exponent-2} implies that
\begin{align}
	p_\abb{error}\fleft(\fleft(\ch{E}_{ML\to A},\ch{D}_{B\to\g{M}L}\fright),\ch{N}_{A\to B}\fright)&\geq\frac{d_M-1}{\lambda\mu+d_M-1}.
\end{align}
For a fixed $d_M$, infimizing the left-hand side over every $d_L\in\spa{N}$ and every strategy $(\ch{E}_{ML\to A},\ch{D}_{B\to\g{M}L})$, and supremizing the right-hand side over every $\lambda,\mu\in\spa{R}$, $\sigma_B\in\s{D}_B$, and $\tau_B\in\aff(\s{D}_B)$ satisfying Eqs.~\eqref{pf:classical-converse-3} and \eqref{pf:classical-converse-4}, we have that
\begin{align}
	\inf_{\substack{d_L\in\spa{N}, \\
	\ch{E}_{ML\to A}\in\s{C}_{ML\to A}, \\ \ch{D}_{B\to\g{M}L}\in\s{C}_{B\to\g{M}L}}}p_\abb{error}\fleft(\fleft(\ch{E}_{ML\to A},\ch{D}_{B\to\g{M}L}\fright),\ch{N}_{A\to B}\fright)&\geq\frac{d_M-1}{2^{I_{\max}\fleft(\ch{N}\fright)+I_\abb{doe}\fleft(\ch{N}\fright)}+d_M-1}. \label{pf:classical-exponent-3}
\end{align}
Now consider the strategy $(\ch{E}_{ML_1L_2\to A},\ch{D}_{B\to\g{M}L_1L_2})$ defined in Eqs.~\eqref{pf:classical-achievability-1} and \eqref{pf:classical-achievability-2}.  It follows from Eq.~\eqref{pf:classical-achievability-5} that, for every symbol $m$ in the alphabet,
\begin{align}
	1-\Pr\fleft\{m\middle|m;\fleft(\ch{E}_{ML_1L_2\to A},\ch{D}_{B\to\g{M}L_1L_2}\fright),\ch{N}_{A\to B}\fright\}&\leq\left(\frac{\tr\fleft[\Phi_{AA'}\left(\ch{K}_{B\to A}^{(0)}\circ\ch{N}_{A\to B}\right)\fleft[\Phi_{AA'}\fright]\fright]}{\left(d_M^2-1\right)\tr\fleft[\Phi_{AA'}\left(\ch{K}_{B\to A}^{(1)}\circ\ch{N}_{A\to B}\right)\fleft[\Phi_{AA'}\fright]\fright]}+1\right)^{-1}.
\end{align}
Supremizing the left-hand side over every $d_R\in\spa{N}$ and every symbol $m$ in the alphabet, and choosing the channels $\ch{K}_{B\to A}^{(0)}$ and $\ch{K}_{B\to A}^{(1)}$ according to Eqs.~\eqref{pf:classical-achievability-9} and \eqref{pf:classical-achievability-10}, it follows from Eq.~\eqref{eq:error} that
\begin{align}
	&p_\abb{error}\fleft(\fleft(\ch{E}_{ML_1L_2\to A},\ch{D}_{B\to\g{M}L_1L_2}\fright),\ch{N}_{A\to B}\fright) \notag\\
	&=\left(\frac{\tr\fleft[\Phi_{AA'}\left(\ch{K}_{B\to A}^{(0)}\circ\ch{N}_{A\to B}\right)\fleft[\Phi_{AA'}\fright]\fright]}{\left(d_M-1\right)\tr\fleft[\Phi_{AA'}\left(\ch{K}_{B\to A}^{(1)}\circ\ch{N}_{A\to B}\right)\fleft[\Phi_{AA'}\fright]\fright]}+1\right)^{-1} \\
	&=\left(\frac{\max_{\ch{K}_{B\to A}\in\s{C}_{B\to A}}\tr\fleft[\Phi_{AA'}\left(\ch{K}_{B\to A}\circ\ch{N}_{A\to B}\right)\fleft[\Phi_{AA'}\fright]\fright]}{\left(d_M-1\right)\min_{\ch{K}_{B\to A}\in\s{C}_{B\to A}}\tr\fleft[\Phi_{AA'}\left(\ch{K}_{B\to A}\circ\ch{N}_{A\to B}\right)\fleft[\Phi_{AA'}\fright]\fright]}+1\right)^{-1} \\
	&=\left(\frac{1}{d_M-1}2^{I_{\max}\fleft(\ch{N}\fright)+I_\abb{doe}\fleft(\ch{N}\fright)}+1\right)^{-1} \label{pf:classical-exponent-4}\\
	&=\frac{d_M-1}{2^{I_{\max}\fleft(\ch{N}\fright)+I_\abb{doe}\fleft(\ch{N}\fright)}+d_M-1}, \label{pf:classical-exponent-5}
\end{align}
where Eq.~\eqref{pf:classical-exponent-4} follows from the interpretation of the max- and Doeblin informations in terms of the maximum and minimum achievable singlet fractions, respectively (Lemma~\ref{lem:fraction}).  Combining Eqs.~\eqref{pf:classical-exponent-3} and \eqref{pf:classical-exponent-5}, we have that
\begin{align}
	\inf_{\substack{d_L\in\spa{N}, \\
	\ch{E}_{ML\to A}\in\s{C}_{ML\to A}, \\ \ch{D}_{B\to\g{M}L}\in\s{C}_{B\to\g{M}L}}}p_\abb{error}\fleft(\fleft(\ch{E}_{ML\to A},\ch{D}_{B\to\g{M}L}\fright),\ch{N}_{A\to B}\fright)&=\frac{d_M-1}{2^{I_{\max}\fleft(\ch{N}\fright)+I_\abb{doe}\fleft(\ch{N}\fright)}+d_M-1}.
\end{align}
Then it follows from the definition of the error exponent of retrocausal classical communication [Eq.~\eqref{eq:classical-error}] that
\begin{align}
	E_\abb{retro,C}^r\fleft(\ch{N}\fright)&=\lim_{n\to\infty}\sup_{d_M\in\spa{N}}\left\{-\frac{1}{n}\log_2\frac{d_M-1}{2^{I_{\max}\fleft(\ch{N}^{\otimes n}\fright)+I_\abb{doe}\fleft(\ch{N}^{\otimes n}\fright)}+d_M-1}\colon d_M\geq2^{nr}\right\} \\
	&=\lim_{n\to\infty}-\frac{1}{n}\log_2\frac{2^{nr}-1}{2^{I_{\max}\fleft(\ch{N}^{\otimes n}\fright)+I_\abb{doe}\fleft(\ch{N}^{\otimes n}\fright)}+2^{nr}-1} \\
	&=\max\left\{0,I_{\max}\fleft(\ch{N}\fright)+I_\abb{doe}^\infty\fleft(\ch{N}\fright)-r\right\}. \label{pf:classical-exponent-6}
\end{align}
where Eq.~\eqref{pf:classical-exponent-6} follows from the additivity of the max-information [Eq.~\eqref{eq:max-regularized}] and the definition of the regularized Doeblin information [Eq.~\eqref{eq:doeblin-regularized}].  Likewise, it follows from the definition of the strong-converse exponent of retrocausal classical communication [Eq.~\eqref{eq:classical-strong}] that
\begin{align}
	S_\abb{retro,C}^r\fleft(\ch{N}\fright)&=\lim_{n\to\infty}\inf_{d_M\in\spa{N}}\left\{-\frac{1}{n}\log_2\left(1-\frac{d_M-1}{2^{I_{\max}\fleft(\ch{N}^{\otimes n}\fright)+I_\abb{doe}\fleft(\ch{N}^{\otimes n}\fright)}+d_M-1}\right)\colon d_M\geq2^{nr}\right\} \\
	&=\lim_{n\to\infty}-\frac{1}{n}\log_2\frac{2^{I_{\max}\fleft(\ch{N}^{\otimes n}\fright)+I_\abb{doe}\fleft(\ch{N}^{\otimes n}\fright)}}{2^{I_{\max}\fleft(\ch{N}^{\otimes n}\fright)+I_\abb{doe}\fleft(\ch{N}^{\otimes n}\fright)}+2^{nr}-1} \\
	&=\max\left\{0,r-I_{\max}\fleft(\ch{N}\fright)-I_\abb{doe}^\infty\fleft(\ch{N}\fright)\right\},
\end{align}
which completes the proof of the desired statement.
\end{proof}

\section{Examples}
\label{sec:example}

In this section, we apply our results to some special classes of channels, and we numerically compare the retrocausal capacity of these channels to their entanglement-assisted and unassisted capacities.

\subsection{Depolarizing channels}
\label{sec:depolarizing}

A depolarizing channel $\ch{D}_{A\to A}^p\in\s{C}_{A\to A}$ with parameter $p\in[0,1]$ is defined as
\begin{align}
	\ch{D}_{A\to A}^p&\equiv\left(1-p\right)\id_{A\to A}+p\ch{R}_{A\to A}^\pi. \label{eq:depolarizing}
\end{align}
Since $\ch{D}_{A\to A}^p$ is a mictodiactic covariant channel, its regularized Doeblin information is equal to its Doeblin information~\cite[Theorem~4]{george2025QuantumDoeblinCoefficients}, both of which are given by
\begin{align}
	I_\abb{doe}^\infty\fleft(\ch{D}^p\fright)&=I_\abb{doe}\fleft(\ch{D}^p\fright) \\
	&=\inf_{\lambda\in\spa{R}}\left\{\log_2\lambda\colon\pi_{A'}\otimes\pi_A\leq\lambda\left(\left(1-p\right)\Phi_{A'A}+p\pi_{A'}\otimes\pi_A\right)\right\} \label{pf:depolarizing-1}\\
	&=-\log_2p,
\end{align}
where Eq.~\eqref{pf:depolarizing-1} follows from Ref.~\cite[Lemma~3]{george2025QuantumDoeblinCoefficients} and Eq.~\eqref{eq:depolarizing}.  The max-information of $\ch{D}_{A\to A}^p$ is equal to~\cite[Eq.~(66)]{fang2020QuantumChannelSimulation}
\begin{align}
	I_{\max}\fleft(\ch{D}^p\fright)&=\log_2\left(\left(1-p\right)d_A^2+p\right).
\end{align}
By Theorem~\ref{thm:classical-asymptotic}, the asymptotic retrocausal classical capacity of $\ch{D}_{A\to A}^p$ is thus given by
\begin{align}
	C_\abb{retro}\fleft(\ch{D}^p\fright)&=I_{\max}\fleft(\ch{D}^p\fright)+I_\abb{doe}^\infty\fleft(\ch{D}^p\fright) \\
	&=\log_2\left(\left(1-p\right)d_A^2+p\right)-\log_2p.
\end{align}
The asymptotic entanglement-assisted classical capacity of $\ch{D}_{A\to A}^p$ is equal to~\cite[Eq.~(3)]{bennett1999EntanglementassistedClassicalCapacity}
\begin{align}
	C_\abb{EA}\fleft(\ch{D}^p\fright)&=2\log_2d_A-h_2\fleft(\frac{p\left(d_A^2-1\right)}{d_A^2}\fright)-\frac{p\left(d_A^2-1\right)}{d_A^2}\log_2\left(d_A^2-1\right),
\end{align}
where $h_2(x)\equiv-x\log_2x-(1-x)\log_2(1-x)$ is the binary entropy.  The asymptotic (unassisted) classical capacity of $\ch{D}_{A\to A}^p$ is equal to~\cite[Theorem~1]{king2003CapacityQuantumDepolarizing}
\begin{align}
	C\fleft(\ch{D}^p\fright)&=\log_2d_A-h_2\fleft(\frac{pd_A}{d_A+1}\fright)-\frac{pd_A}{d_A+1}\log_2\left(d_A-1\right).
\end{align}
See Fig.~\ref{fig:depolarizing} for a numerical comparison between the asymptotic retrocausal classical capacity, the asymptotic entanglement-assisted classical capacity, and the asymptotic (unassisted) classical capacity of depolarizing channels for when $d_A=2$.

\begin{figure}[t]
\includegraphics[scale=0.9]{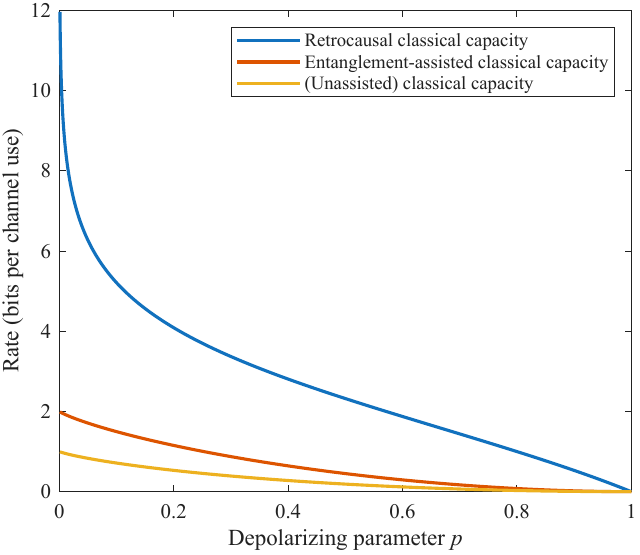}
\caption{For qubit depolarizing channels, the asymptotic retrocausal classical capacity, the asymptotic entanglement-assisted classical capacity, and the asymptotic (unassisted) classical capacity are compared numerically as functions of the depolarizing parameter $p\in[0,1]$.  The asymptotic retrocausal classical capacity approaches $\infty$ as $p\searrow0$.}
\label{fig:depolarizing}
\end{figure}

\subsection{Erasure channels}
\label{sec:erasure}

An erasure channel $\ch{E}_{A\to B}^p\in\s{C}_{A\to B}$ with $\spa{H}_B=\spa{H}_A\oplus\spn(\ket{e})$ and parameter $p\in[0,1]$ is defined as
\begin{align}
	\ch{E}_{A\to B}^p&\equiv\left(1-p\right)\overline{\id}_{A\to B}+p\ch{R}_{A\to B}^\op{e}{e},
\end{align}
where $\overline{\id}_{A\to B}[\cdot]\equiv\sum_{i,j=0}^{d_A-1}\tr[\op{j}{i}_A(\cdot)_A]\op{i}{j}_B$ with $\{\ket{i}\}_{i=0}^{d_A-1}$ an orthonormal basis of $\spa{H}_A$.  It follows from Eq.~\eqref{eq:wang} that
\begin{align}
	I_\wang\fleft(\ch{E}^p\fright)&=\inf_{\substack{\lambda\in\spa{R}, \\ \tau_B\in\aff\fleft(\s{D}_B\fright)}}\left\{\log_2\lambda\colon-\lambda\Phi_{A'B}^{\ch{E}^p}\leq\pi_{A'}\otimes\tau_B\leq\lambda\Phi_{A'B}^{\ch{E}^p}\right\} \\
	&\leq\inf_{\lambda\in\spa{R}}\left\{\log_2\lambda\colon-\lambda\Phi_{A'B}^{\ch{E}^p}\leq\pi_{A'}\otimes\op{e}{e}_B\leq\lambda\Phi_{A'B}^{\ch{E}^p}\right\} \\
	&=\inf_{\lambda\in\spa{R}}\left\{\log_2\lambda\colon\pi_{A'}\otimes\op{e}{e}_B\leq\lambda\left(\left(1-p\right)\overline{\Phi}_{A'B}+p\pi_{A'}\otimes\op{e}{e}_B\right)\right\} \\
	&=-\log_2 p, \label{pf:erasure-1}
\end{align}
where $\overline{\Phi}_{A'B}\equiv(1/d_A)\sum_{i,j=0}^{d_A-1}\op{i}{j}_{A'}\otimes\op{i}{j}_B$.  Now define the following channel:
\begin{align}
	\ch{P}_{B\to A}\fleft[\cdot\fright]&\coloneq\sum_{i,j=0}^{d_A-1}\tr\fleft[\op{j}{i}_B\left(\cdot\right)_B\fright]\op{\left(i+1\right)\bmod d_A}{\left(j+1\right)\bmod d_A}_A+\tr\fleft[\op{e}{e}_B\left(\cdot\right)_B\fright]\pi_A.
\end{align}
It follows from the interpretation of the Doeblin information in terms of the minimum achievable singlet fraction (Lemma~\ref{lem:fraction}) that
\begin{align}
	I_\abb{doe}\fleft(\ch{E}^p\fright)&=-2\log_2d_A-\log_2\min_{\ch{K}_{B\to A}\in\s{C}_{B\to A}}\tr\fleft[\Phi_{AA'}\left(\ch{K}_{B\to A}\circ\ch{E}_{A\to B}^p\right)\fleft[\Phi_{AA'}\fright]\fright] \\
	&\geq-2\log_2d_A-\log_2\tr\fleft[\Phi_{AA'}\left(\ch{P}_{B\to A}\circ\ch{E}_{A\to B}^p\right)\fleft[\Phi_{AA'}\fright]\fright] \\
	&=-2\log_2d_A-\log_2\left(p\tr\fleft[\Phi_{AA'}\left(\pi_A\otimes\pi_{A'}\right)\fright]\right) \\
	&=-\log_2p. \label{pf:erasure-2}
\end{align}
Combining Eqs.~\eqref{pf:erasure-1} and \eqref{pf:erasure-2}, it follows from Eq.~\eqref{eq:bounds} that
\begin{align}
	-\log_2p&\leq I_\abb{doe}\fleft(\ch{E}^p\fright)\leq I_\abb{doe}^\infty\fleft(\ch{E}^p\fright)\leq I_\wang\fleft(\ch{E}^p\fright)\leq-\log_2p,
\end{align}
which implies that
\begin{align}
	I_\abb{doe}\fleft(\ch{E}^p\fright)&=I_\abb{doe}^\infty\fleft(\ch{E}^p\fright)=I_\wang\fleft(\ch{E}^p\fright)=-\log_2p.
\end{align}
The max-information of $\ch{E}_{A\to B}^p$ is equal to~\cite[Eq.~(75)]{fang2020QuantumChannelSimulation}
\begin{align}
	I_{\max}\fleft(\ch{E}^p\fright)&=\log_2\left(\left(1-p\right)d_A^2+p\right).
\end{align}
By Theorem~\ref{thm:classical-asymptotic}, the asymptotic retrocausal classical capacity of $\ch{E}_{A\to B}^p$ is thus given by
\begin{align}
	C_\abb{retro}\fleft(\ch{E}^p\fright)&=I_{\max}\fleft(\ch{E}^p\fright)+I_\abb{doe}^\infty\fleft(\ch{E}^p\fright) \\
	&=\log_2\left(\left(1-p\right)d_A^2+p\right)-\log_2p.
\end{align}
The asymptotic entanglement-assisted classical capacity of $\ch{E}_{A\to B}^p$ is equal to~\cite[p.~3083]{bennett1999EntanglementassistedClassicalCapacity}
\begin{align}
	C_\abb{EA}\fleft(\ch{E}^p\fright)&=2\left(1-p\right)\log_2d_A.
\end{align}
The asymptotic (unassisted) classical capacity of $\ch{E}_{A\to B}^p$ is equal to~\cite{bennett1997CapacitiesQuantumErasure}
\begin{align}
	C\fleft(\ch{E}^p\fright)&=\left(1-p\right)\log_2d_A.
\end{align}
See Fig.~\ref{fig:erasure} for a numerical comparison between the asymptotic retrocausal classical capacity, the asymptotic entanglement-assisted classical capacity, and the asymptotic (unassisted) classical capacity of erasure channels for when $d_A=2$.

\begin{figure}[t]
\includegraphics[scale=0.9]{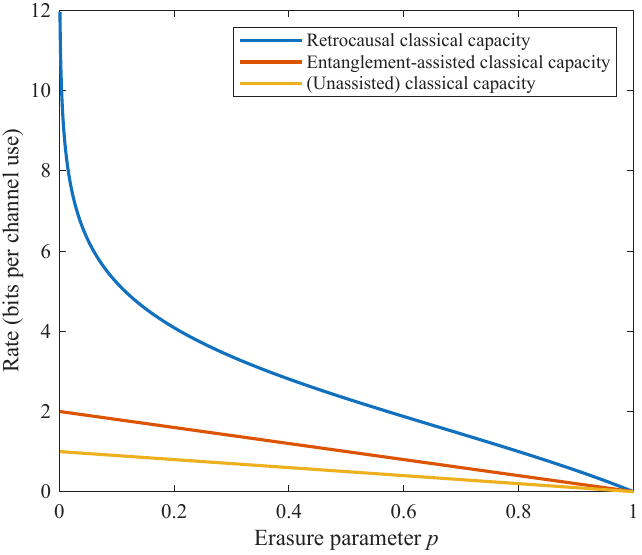}
\caption{For qubit erasure channels, the asymptotic retrocausal classical capacity, the asymptotic entanglement-assisted classical capacity, and the asymptotic (unassisted) classical capacity are compared numerically as functions of the erasure parameter $p\in[0,1]$.  The asymptotic retrocausal classical capacity approaches $\infty$ as $p\searrow0$.}
\label{fig:erasure}
\end{figure}

\subsection{Amplitude-damping channels}
\label{sec:damping}

An amplitude-damping channel $\ch{A}_{A\to A}^\gamma\in\s{C}_{A\to A}$ with $d_A=2$ and parameter $\gamma\in[0,1]$ is defined as
\begin{align}
	\ch{A}_{A\to A}^\gamma\fleft[\cdot\fright]&\equiv\left(\op{0}{0}_A+\sqrt{1-\gamma}\op{1}{1}_A\right)\left(\cdot\right)_A\left(\op{0}{0}_A+\sqrt{1-\gamma}\op{1}{1}_A\right)+\gamma\op{0}{1}_A\left(\cdot\right)_A\op{1}{0}_A.
\end{align}
The Doeblin information of $\ch{A}_{A\to A}^\gamma$ is equal to~\cite[Lemma~9]{george2025QuantumDoeblinCoefficients}
\begin{align}
	I_\abb{doe}\fleft(\ch{A}^\gamma\fright)&=-2\log_2\left(1-\sqrt{1-\gamma}\right).
\end{align}
The max-information of $\ch{A}_{A\to A}^\gamma$ is given by~\cite[Eq.~(69)]{fang2020QuantumChannelSimulation}
\begin{align}
	I_{\max}\fleft(\ch{A}^\gamma\fright)&=\log_2\left(2\left(1+\sqrt{1-\gamma}\right)-\gamma\right).
\end{align}
By Theorem~\ref{thm:quantum-asymptotic}, the asymptotic retrocausal quantum capacity of $\ch{A}_{A\to A}^\gamma$ has the following lower bound:
\begin{align}
	Q_\abb{retro}\fleft(\ch{A}^\gamma\fright)&=\frac{1}{2}\left(I_{\max}\fleft(\ch{A}^\gamma\fright)+I_\abb{doe}^\infty\fleft(\ch{A}^\gamma\fright)\right) \\
	&\geq\frac{1}{2}\left(I_{\max}\fleft(\ch{A}^\gamma\fright)+I_\abb{doe}\fleft(\ch{A}^\gamma\fright)\right) \\
	&=\frac{1}{2}\left(\log_2\left(2\left(1+\sqrt{1-\gamma}\right)-\gamma\right)-2\log_2\left(1-\sqrt{1-\gamma}\right)\right). \label{eq:damping}
\end{align}
The asymptotic entanglement-assisted quantum capacity of $\ch{A}_{A\to A}^\gamma$ is equal to~\cite[Eq.~(38)]{giovannetti2005InformationcapacityDescriptionSpinchain}
\begin{align}
	Q_\abb{EA}\fleft(\ch{A}^\gamma\fright)&=\max_{\lambda\in[0,1]}\left\{h_2\fleft(\lambda\fright)+h_2\fleft(\left(1-\gamma\right)\lambda\fright)-h_2\fleft(\gamma\lambda\fright)\right\}.
\end{align}
The asymptotic (unassisted) quantum capacity of $\ch{A}_{A\to A}^\gamma$ is equal to~\cite[Eq.~(36)]{giovannetti2005InformationcapacityDescriptionSpinchain}
\begin{align}
	Q\fleft(\ch{A}^\gamma\fright)&=\max_{\lambda\in[0,1]}\left\{h_2\fleft(\left(1-\gamma\right)\lambda\fright)-h_2\fleft(\gamma\lambda\fright)\right\}.
\end{align}
See Fig.~\ref{fig:damping} for a numerical comparison between the achievable asymptotic rate for retrocausal quantum communication in Eq.~\eqref{eq:damping}, the asymptotic entanglement-assisted quantum capacity, and the asymptotic (unassisted) quantum capacity of amplitude-damping channels.

\begin{figure}[t]
\includegraphics[scale=0.9]{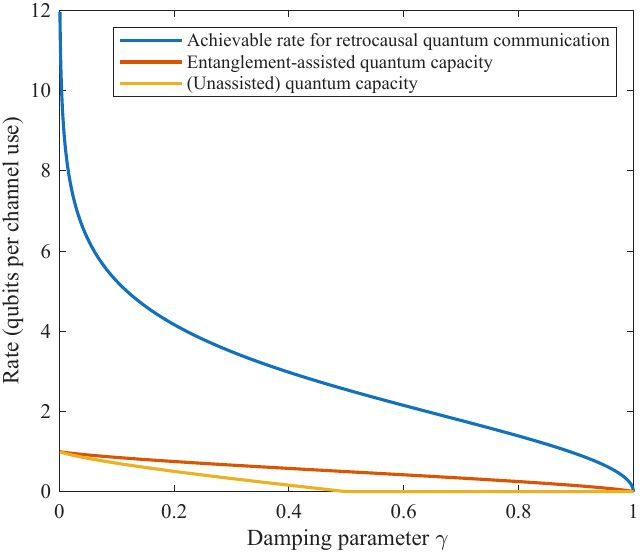}
\caption{For amplitude-damping channels, an achievable asymptotic rate for retrocausal quantum communication, the asymptotic entanglement-assisted quantum capacity, and the asymptotic (unassisted) quantum capacity are compared numerically as functions of the damping parameter $\gamma\in[0,1]$.  The asymptotic retrocausal quantum capacity approaches $\infty$ as $\gamma\searrow0$.}
\label{fig:damping}
\end{figure}

\end{document}